\documentclass[12pt]{article}
\usepackage{amsfonts}
\usepackage{stmaryrd}
\usepackage{array}
\usepackage{graphicx}
\usepackage{graphics}
\usepackage{epstopdf}
\usepackage{epsfig}
\usepackage{amsmath}
\usepackage{amsthm}
\usepackage{amssymb}
\usepackage{float}
\usepackage{cite}
\usepackage{multirow}
\usepackage{verbatim}
\usepackage{fancyhdr}
\usepackage{subfigure}
\usepackage{color}
\usepackage{mathtools}
\usepackage{sectsty}
\usepackage[title]{appendix}
\usepackage{threeparttable}
\usepackage{dcolumn}
\usepackage{booktabs}
\usepackage{indentfirst}
\usepackage{setspace}
\usepackage{bm}
\usepackage{enumerate}
\usepackage[labelfont=bf,labelsep=period]{caption}

\newcolumntype{d}[1]{D{.}{.}{#1}}
\theoremstyle{definition}

\theoremstyle{plain}

\newtheorem{Pro}{Proposition}[section]

\numberwithin{equation}{section}

\makeatletter
\renewcommand{\maketag@@@}[1]{\hbox{\m@th\normalsize\normalfont#1}}%
\makeatother

\allowdisplaybreaks
\begin{document}
\pagenumbering{arabic}
\baselineskip=1.4pc

\vspace*{0.5in}

\begin{center}

{{\bf \Large
A Conservative Time-Accurate Local Time-Stepping DG Scheme Based on a Weakly Compressible Model for Unsteady Low-Mach-Number Flows}}

\end{center}

\vspace{.03in}

\centerline{
Shihao Liu \footnote{
	Digital Twin Research Center, Institute of Engineering Thermophysics, Chinese Academy of Sciences, Beijing 100190, PR China. Supported by the Strategic Priority Research Program of the Chinese Academy of Sciences, Grant No. XDA0390501
	E-mail: liushihao@iet.cn.},
Keli Zhang \footnote{
	School of Aeronautic Science and Engineering, Beihang University, 100191 Beijing, PR China.
    Rankyee Technology Beijing, 100036 Beijing,  PR China.
	E-mail: zkl70@buaa.edu.cn.},
Yuning Luan \footnote{
	School of Mechanical Engineering, Tianjin University, 300072 Tianjin, PR China.
	E-mail: luanyuning@163.com},
Kai Liu \footnote{
	Rankyee Technology Beijing, 100036 Beijing,  PR China.
	E-mail: liukai@rankyee.com.}}

\vspace{.1in}

\noindent
\textbf{Abstract:}
This paper presents a conservative high-order discontinuous Galerkin (DG) method with time-accurate local time stepping for the accurate simulation of low-Mach-number unsteady flows based on a weakly compressible formulation. In this model, pressure is prescribed solely as a function of density, which avoids the solution of a global pressure Poisson equation required in traditional incompressible approaches, while retaining the locality and conservation properties of compressible formulations. This formulation is suitable for the numerical simulation of low-speed unsteady flows and related aeroacoustic problems. The spatial discretization is constructed in a strong-form nodal DG spectral element method (DGSEM) framework based on Gauss--Lobatto--Legendre (GLL) points, in which inviscid interface coupling is realized through numerical fluxes designed for the weakly compressible system, while viscous terms are discretized by the incomplete interior penalty Galerkin (IIPG) method. For the convective part, in addition to the local Lax--Friedrichs flux, a two-rarefaction approximate Riemann flux is formulated for the present constant-sound-speed barotropic equation of state by specializing the Riemann-invariant relations to the local normal direction. For temporal discretization, the method adopts a CERK continuous extension to obtain a cell-local predictor polynomial for continuous-in-time reconstruction of the volume terms. The face correction is split into a purely interior contribution and a common-flux contribution: the former is integrated with the same temporal quadrature as the volume term, whereas the latter is integrated by face-wise piecewise Gaussian quadrature using interface states evaluated from the continuous predictors. This split treatment preserves the discrete summation-by-parts cancellation in the nodal DGSEM formulation and keeps the common-flux exchange conservative across element interfaces. Across the numerical examples, periodic tests verify roundoff-level mass conservation, the cylinder test indicates weaker far-field pressure fluctuations with the TR flux, and the cavity and 30P30N cases demonstrate viscous-flow resolution and applicability to complex three-dimensional configurations. The proposed method combines high-order accuracy, geometric flexibility, temporal accuracy for unsteady simulations, support for local time stepping, and conservation under local time stepping, and provides a conservative LTS-DG framework for high-fidelity simulations of complex low-Mach-number unsteady flows.

\bigskip


\pagenumbering{arabic}

\baselineskip=1.4pc

\section{Introduction}

Low-Mach-number flows arise in a wide range of engineering applications, including aeroacoustic prediction, flow control, and complex unsteady vortex-dominated flows. High-fidelity numerical simulation of such problems usually requires not only accurate resolution of local fine-scale flow structures, but also long-time accuracy in the evolution of both hydrodynamic and acoustic fields. However, conventional compressible flow solvers often suffer from pronounced numerical stiffness in the low-Mach-number regime, because the acoustic wave speed is much larger than the characteristic flow velocity. As a result, the stable time step of explicit schemes is severely restricted by both the smallest mesh size and the largest local wave speed, leading to substantial computational cost, especially for high-order simulations on nonuniform meshes.

For low-speed flow computation, traditional approaches more commonly rely on incompressible models together with pressure-velocity coupling algorithms. Typical examples include SIMPLE-type methods based on pressure correction \cite{PatankarSpalding1972,Patankar1980} and projection methods based on Helmholtz decomposition \cite{Chorin1968,Temam1969}. These methods have been widely used in incompressible flow simulations and are well known for their maturity, robustness, and clear physical interpretation. However, a common feature of such approaches is that they require the solution of a pressure Poisson equation, or an equivalent global elliptic problem, in order to enforce the divergence constraint on the velocity field. For large-scale unsteady simulations, complex geometries, and high-order local discretizations, this global coupling step often introduces considerable algorithmic complexity and communication overhead, and may partially offset the locality and parallel efficiency of explicit methods.

To improve the applicability of compressible-type numerical methods to low-speed flows, weakly compressible models have attracted continued attention in recent years. By introducing a barotropic equation of state in which pressure depends only on density, such models retain the conservative form and the implementation convenience of compressible formulations, while improving their suitability for low-Mach-number unsteady flow simulations. Compared with traditional incompressible models, one direct advantage of weakly compressible formulations is that no additional pressure Poisson equation needs to be solved. Instead, pressure evolution is obtained through local updates of the conservative variables together with the equation of state, which makes it easier to construct locally dependent numerical algorithms. Furthermore, this formulation is compatible with explicit time integration, thereby avoiding large-scale matrix assembly and the solution of global linear systems. As a consequence, it reduces global-coupling overhead and preserves locality, which is important for parallel implementation, local time stepping, and discretization on complex unstructured meshes. In addition, weakly compressible models can be naturally coupled with aeroacoustic post-processing techniques, which has motivated their use in the prediction of low-speed flows and sound generation.

Among numerical approaches based on weakly compressible formulations, the lattice Boltzmann method (LBM) has been widely used for low-Mach-number flow simulations \cite{ChenDoolen1998,Succi2001,AidunClausen2010}. Owing to its explicit streaming-collision structure, LBM possesses strong locality, high parallel efficiency, and advantages for local time advancement. However, LBM generally relies on structured lattice arrangements and imposes relatively strict requirements on mesh regularity, lattice connectivity, and the compatibility between discrete velocity directions and the underlying grid. Consequently, it is less flexible for strongly nonuniform meshes, local mesh refinement, curved boundaries, and general unstructured discretizations. Moreover, boundary treatment in LBM is often nontrivial. The accurate implementation of no-slip walls, curved boundaries, inflow/outflow conditions, and coupled boundary physics usually requires specially designed bounce-back, interpolation, or reconstruction procedures \cite{ZouHe1997,Bouzidi2001,GuoZhengShi2002,AnsumaliKarlin2002}, whose robustness and accuracy may depend sensitively on the local lattice arrangement. Another representative approach is the smoothed particle hydrodynamics (SPH) method, which also often adopts a weakly compressible formulation in order to avoid solving a pressure Poisson equation \cite{Monaghan1992,Monaghan1994,MorrisFoxZhu1997,HuAdams2006}. As a meshfree Lagrangian framework, SPH is attractive for problems involving large deformation, free surfaces, and moving boundaries, but it also faces challenges related to particle disorder, numerical dissipation, kernel consistency, and reduced boundary accuracy.

Meanwhile, high-order discontinuous Galerkin (DG) methods have become an important tool for high-fidelity flow computation because of their high-order accuracy, compact stencil, geometric flexibility, and favorable parallel performance \cite{CockburnShu2001,HesthavenWarburton2008}. For multiscale unsteady problems with local mesh refinement, local time-stepping (LTS) strategies can effectively alleviate the efficiency bottleneck caused by the globally smallest allowable time step. In recent years, several LTS frameworks have been developed for DG methods \cite{GassnerHindenlangMunz2011,GassnerDumbserHindenlangMunz2011,DumbserBalsaraToroMunz2008}. In particular, continuous explicit Runge--Kutta (CERK) methods provide a natural continuous-in-time predictor for one-step formulations and time-accurate local time stepping \cite{OwrenZennaro1992,GassnerHindenlangMunz2011}. For weakly compressible low-Mach-number flows, it is of interest to incorporate the corresponding physical model, spatial discretization, numerical flux treatment, and time integration strategy effectively into a high-order DG framework with local time stepping.

Motivated by these considerations, this paper studies a weakly compressible DG method with local time stepping for low-Mach-number flows. First, a weakly compressible model is incorporated into a high-order DG framework with local time stepping, so as to combine the suitability of the model for low-Mach-number flows with the flexibility of high-order discretization. Second, the viscous terms are discretized by the incomplete interior penalty Galerkin (IIPG) method \cite{Arnold1982,ArnoldBrezziCockburnMarini2002,Riviere2008} to obtain a stable viscous treatment. For the inviscid part, a two-rarefaction approximate Riemann flux is formulated according to the characteristics of the weakly compressible system and the philosophy of approximate Riemann solvers \cite{Toro2009}, so that interfacial information can be transmitted in a consistent and robust manner. Since the present spatial discretization is a nodal DGSEM based on GLL points, conservation under local time stepping also requires the discrete summation-by-parts cancellation between the volume term and the purely interior part of the strong-form face correction to be retained. For this reason, the surface contribution is split into a purely interior part and a common-flux part: the former is integrated together with the volume term using element-local quadrature, whereas only the common-flux part is integrated by face-wise time matching.

The methodology developed here has already been integrated into the commercial software DIMAXER, and has been extended to the flux reconstruction (FR) framework and applied to aeroacoustic simulations \cite{zhang2026sweep}. The present paper focuses on a description of its original DG formulation.

Therefore, the main contributions of this work are to incorporate a weakly compressible model into a high-order DG framework with local time stepping, to adopt the IIPG treatment for viscous terms, to formulate a two-rarefaction approximate Riemann flux specialized to the weakly compressible system, and to construct a CERK-based time-accurate local time-stepping formulation with a conservation-preserving split treatment of the face contribution. With these ingredients, the resulting method retains high-order accuracy while supporting local time stepping and preserving the conservative interface exchange of the underlying nodal DG discretization. A series of numerical experiments are carried out to assess the accuracy, conservation behavior, robustness, and applicability of the proposed method.

\section{Governing Equations}

We consider a weakly compressible model for low-Mach-number flows. In two space dimensions, the governing equations are written in conservative form as
\begin{equation}
\frac{\partial \rho}{\partial t}
+\frac{\partial (\rho u)}{\partial x}
+\frac{\partial (\rho v)}{\partial y}=0,
\label{eq:wc_cont}
\end{equation}
\begin{equation}
\frac{\partial (\rho u)}{\partial t}
+\frac{\partial (\rho u^2+p)}{\partial x}
+\frac{\partial (\rho uv)}{\partial y}
=\mu\left(\frac{\partial^2 u}{\partial x^2}
+\frac{\partial^2 u}{\partial y^2}\right),
\label{eq:wc_momx}
\end{equation}
\begin{equation}
\frac{\partial (\rho v)}{\partial t}
+\frac{\partial (\rho uv)}{\partial x}
+\frac{\partial (\rho v^2+p)}{\partial y}
=\mu\left(\frac{\partial^2 v}{\partial x^2}
+\frac{\partial^2 v}{\partial y^2}\right),
\label{eq:wc_momy}
\end{equation}
where $\rho$ is the density, $(u,v)^T$ is the velocity vector, $p$ is the pressure, and $\mu$ is the dynamic viscosity. When a reference density is used, the corresponding reference kinematic viscosity is $\mu/\rho_0$. The pressure is assumed to satisfy a barotropic equation of state,
\begin{equation*}
p=p(\rho),
\end{equation*}
with
\begin{equation*}
p'(\rho)>0,\qquad \rho>0,
\end{equation*}
so that the convective subsystem remains hyperbolic.

Introducing the conservative variable vector
\begin{equation*}
U=
\begin{pmatrix}
\rho\\
\rho u\\
\rho v
\end{pmatrix}
=
\begin{pmatrix}
\rho\\
m_1\\
m_2
\end{pmatrix},
\end{equation*}
where $m_1=\rho u$ and $m_2=\rho v$ denote the momentum components in the $x$- and $y$-directions, respectively, the system can be expressed in a compact form as
\begin{equation}
\frac{\partial U}{\partial t}
+\frac{\partial F_c(U)}{\partial x}
+\frac{\partial G_c(U)}{\partial y}
=
\frac{\partial F_d(U,\nabla U)}{\partial x}
+\frac{\partial G_d(U,\nabla U)}{\partial y},
\label{eq:wc_compact}
\end{equation}
where the convective fluxes are
\begin{equation*}
F_c(U)=
\begin{pmatrix}
\rho u\\
\rho u^2+p\\
\rho uv
\end{pmatrix},
\qquad
G_c(U)=
\begin{pmatrix}
\rho v\\
\rho uv\\
\rho v^2+p
\end{pmatrix},
\end{equation*}
and the diffusive fluxes are
\begin{equation*}
F_d(U,\nabla U)=
\begin{pmatrix}
0\\
\mu\,u_x\\
\mu\,v_x
\end{pmatrix},
\qquad
G_d(U,\nabla U)=
\begin{pmatrix}
0\\
\mu\,u_y\\
\mu\,v_y
\end{pmatrix}.
\end{equation*}
Since the numerical flux construction relies on the characteristic structure of the hyperbolic part, we next consider the corresponding one-dimensional convective subsystem in the $x$-direction.

Neglecting the $y$-derivative and viscous terms, the governing equations reduce to
\begin{equation*}
\frac{\partial \rho}{\partial t}
+\frac{\partial (\rho u)}{\partial x}=0,
\end{equation*}
\begin{equation*}
\frac{\partial m_1}{\partial t}
+\frac{\partial \left(\dfrac{m_1^2}{\rho}+p\right)}{\partial x}=0,
\end{equation*}
\begin{equation*}
\frac{\partial m_2}{\partial t}
+\frac{\partial \left(\dfrac{m_1m_2}{\rho}\right)}{\partial x}=0.
\end{equation*}

To derive the Riemann invariants, it is convenient to rewrite the one-dimensional system in primitive variables $(\rho,u,v)$. Expanding the governing equations yields
\begin{equation*}
\rho_t+u\rho_x+\rho u_x=0,
\end{equation*}
\begin{equation*}
u_t+uu_x+\frac{1}{\rho}p'(\rho)\rho_x=0,
\end{equation*}
\begin{equation*}
v_t+uv_x=0.
\end{equation*}
Equivalently, one has
\begin{equation*}
\begin{pmatrix}
\rho\\ u\\ v
\end{pmatrix}_t
+
B(\rho,u,v)
\begin{pmatrix}
\rho\\ u\\ v
\end{pmatrix}_x
=0,
\end{equation*}
with
\begin{equation*}
B(\rho,u,v)=
\begin{pmatrix}
u & \rho & 0\\[1mm]
\dfrac{p'(\rho)}{\rho} & u & 0\\[2mm]
0 & 0 & u
\end{pmatrix}.
\end{equation*}

\begin{Pro}
Let
\begin{equation*}
c=\sqrt{p'(\rho)}.
\end{equation*}
Then the characteristic speeds of the one-dimensional barotropic system are
\begin{equation}
\lambda_1=u-c,\qquad
\lambda_2=u,\qquad
\lambda_3=u+c.
\label{eq:eigenvalues_A}
\end{equation}
Moreover, a corresponding set of Riemann invariants is given by
\begin{equation}
w_1=u-\int \frac{\sqrt{p'(\rho)}}{\rho}\,d\rho,\qquad
w_2=v,\qquad
w_3=u+\int \frac{\sqrt{p'(\rho)}}{\rho}\,d\rho.
\label{eq:riemann_general}
\end{equation}
In smooth regions, these quantities satisfy
\begin{equation*}
(w_i)_t+\lambda_i(w_i)_x=0,\qquad i=1,2,3.
\end{equation*}
\end{Pro}

\begin{proof}
We verify the statement directly. The characteristic polynomial of $B(\rho,u,v)$ is
\begin{equation*}
\det(B-\lambda I)
=(u-\lambda)\left[(u-\lambda)^2-p'(\rho)\right].
\end{equation*}
Thus, with $c=\sqrt{p'(\rho)}$, the characteristic speeds are
\begin{equation*}
\lambda_1=u-c,\qquad
\lambda_2=u,\qquad
\lambda_3=u+c.
\end{equation*}

It remains to verify the stated Riemann invariants. Define
\begin{equation*}
H(\rho)=\int \frac{\sqrt{p'(\rho)}}{\rho}\,d\rho,
\qquad H'(\rho)=\frac{c}{\rho}.
\end{equation*}
For $w_2=v$, the third primitive equation gives immediately
\begin{equation*}
(w_2)_t+\lambda_2(w_2)_x
=v_t+uv_x=0.
\end{equation*}

For $w_1=u-H(\rho)$, using the first two primitive equations yields
\begin{equation*}
\begin{aligned}
(w_1)_t+(u-c)(w_1)_x
&=
u_t+(u-c)u_x
-\frac{c}{\rho}\bigl(\rho_t+(u-c)\rho_x\bigr)\\
&=
u_t+u u_x+\frac{c^2}{\rho}\rho_x\\
&=0.
\end{aligned}
\end{equation*}

Similarly, for $w_3=u+H(\rho)$,
\begin{equation*}
\begin{aligned}
(w_3)_t+(u+c)(w_3)_x
&=
u_t+(u+c)u_x
+\frac{c}{\rho}\bigl(\rho_t+(u+c)\rho_x\bigr)\\
&=
u_t+u u_x+\frac{c^2}{\rho}\rho_x\\
&=0.
\end{aligned}
\end{equation*}
Therefore the quantities in \eqref{eq:riemann_general} satisfy
$(w_i)_t+\lambda_i(w_i)_x=0$ in smooth regions, as claimed.
\end{proof}

In the present work, the equation of state is specified as
\begin{equation}
p(\rho)=c_0^2(\rho-\rho_0)+p_0,
\label{eq:specific_eos}
\end{equation}
where $c_0$ is a prescribed artificial sound speed, and $\rho_0$ and $p_0$ denote the reference density and pressure, respectively. This weakly compressible relation is also widely used in SPH-type formulations to avoid solving a pressure Poisson equation \cite{Monaghan1994,HuAdams2006}. Following common weakly compressible SPH practice, the artificial sound speed is usually chosen sufficiently larger than the characteristic flow speed, often about ten times the maximum velocity scale, so that the artificial Mach number remains small and density fluctuations are limited \cite{Monaghan1994,MorrisFoxZhu1997,HuAdams2006}. From \eqref{eq:specific_eos}, one has
\begin{equation*}
p'(\rho)=c_0^2,
\qquad
c=\sqrt{p'(\rho)}=c_0.
\end{equation*}
Therefore, the integral in \eqref{eq:riemann_general} can be evaluated explicitly:
\begin{equation*}
\int \frac{\sqrt{p'(\rho)}}{\rho}\,d\rho
=
\int \frac{c_0}{\rho}\,d\rho
=
c_0\ln\left(\frac{\rho}{\rho_0}\right)+C.
\end{equation*}
The Riemann invariants therefore reduce to
\begin{equation}
w_1=u-c_0\ln\left(\frac{\rho}{\rho_0}\right),\qquad
w_2=v,\qquad
w_3=u+c_0\ln\left(\frac{\rho}{\rho_0}\right).
\label{eq:riemann_special}
\end{equation}
Compared with a general barotropic law, this constant-sound-speed barotropic equation of state yields an explicit logarithmic form of the characteristic relations, which considerably simplifies the formulation of interface relations and the construction of numerical fluxes.

For a general barotropic law satisfying $p'(\rho)>0$, the characteristic structure $u-c$, $u$, and $u+c$ remains unchanged. The difference is that the sound speed $c=\sqrt{p'(\rho)}$ now depends on $\rho$, and the integral in \eqref{eq:riemann_general} cannot, in general, be expressed in closed form. In such cases, the characteristic relations may still be formulated through the integral representation \eqref{eq:riemann_general}. In contrast, for the present constant-sound-speed barotropic equation of state, the explicit form \eqref{eq:riemann_special} is available and will be exploited in the numerical flux construction presented in the next section.

\section{Formulation}

In this section, the numerical formulation of the proposed method is presented. We first describe the spatial discretization framework on curved quadrilateral and hexahedral elements, with emphasis on the strong-form nodal DGSEM formulation and the associated elemental semi-discrete system. The common-flux treatment for the weakly compressible convective terms and the IIPG discretization for the viscous terms are then introduced. These ingredients together provide the basis for the subsequent time-accurate local time-stepping strategy.

\subsection{Nodal DGSEM spatial discretization on curved quadrilateral and hexahedral elements}
\label{subsec:spatial_dgsem}

Since only quadrilateral and hexahedral meshes are considered in the present work, the spatial discretization is formulated within a tensor-product nodal DG spectral element method (DGSEM) framework on reference elements. In this work, DGSEM denotes the strong-form nodal DG discretization in which the interpolation, quadrature rules, and discrete derivative operators are built from Gauss--Lobatto--Legendre (GLL) points. For curved quadrilateral or hexahedral elements, the geometric treatment is based on a reference-element mapping: each physical element is mapped to a standard reference element, and the basis functions, quadrature rules, and discrete operators are constructed in the computational coordinates. For quadrilateral and hexahedral elements, this strategy not only provides a natural representation of curved boundaries and curved geometries, but also fully exploits the efficiency and locality associated with tensor-product structures. The basic geometric treatment adopted here is consistent with the spectral-element DG idea used by Min and Lee \cite{MinLee2011}.

Let $\Omega^e$ denote a physical element. In two dimensions, the corresponding reference element is
\begin{equation*}
I=[-1,1]^2,
\end{equation*}
while in three dimensions it is
\begin{equation*}
I=[-1,1]^3.
\end{equation*}
For a two-dimensional quadrilateral element, a mapping from the reference coordinates $(\xi,\eta)\in I$ to the physical coordinates $(x,y)\in\Omega^e$ is introduced as
\begin{equation*}
x=x^e(\xi,\eta),\qquad y=y^e(\xi,\eta).
\end{equation*}
The associated Jacobian matrix and determinant are
\begin{equation*}
J^e(\xi,\eta)=
\begin{pmatrix}
x_\xi & x_\eta\\
y_\xi & y_\eta
\end{pmatrix},
\qquad
|J^e|=x_\xi y_\eta-x_\eta y_\xi.
\end{equation*}
Accordingly, an integral over the physical element can be transformed to the reference element as
\begin{equation*}
\int_{\Omega^e} q(x,y)\,d\Omega
=
\int_I q\bigl(x^e(\xi,\eta),y^e(\xi,\eta)\bigr)\,|J^e(\xi,\eta)|\,d\xi d\eta.
\end{equation*}
Meanwhile, the differential operators in the physical coordinates are computed from the reference-space derivatives through the geometric mapping as
\begin{equation*}
\nabla_x q = J^{-T}\nabla_{\xi} q,
\end{equation*}
where the two-dimensional Jacobian matrix is
\begin{equation*}
J=
\begin{pmatrix}
x_\xi & x_\eta\\
y_\xi & y_\eta
\end{pmatrix}.
\end{equation*}
The extension to the three-dimensional hexahedral case is completely analogous.

For notational convenience, define the total flux by
\begin{equation*}
\bm{F}(U,\nabla U)=\bm{F}_c(U)-\bm{F}_d(U,\nabla U),
\end{equation*}
where $\bm{F}_c$ denotes the convective flux and $\bm{F}_d$ denotes the diffusive flux. The governing equations may then be written as
\begin{equation*}
\frac{\partial U}{\partial t}+\nabla\cdot \bm{F}(U,\nabla U)=0.
\end{equation*}
Here, the flux $\bm{F}$ corresponds to the full flux operator introduced in the governing equations, namely, the combined contribution of the convective and diffusive parts. In the DG discretization, however, the convective and diffusive parts will be treated separately in order to account for their different numerical properties.

On each element, a local approximation is introduced using tensor-product interpolation at Gauss--Lobatto--Legendre (GLL) points. In the two-dimensional case, the approximate solution is written as
\begin{equation*}
U_h^e(\xi,\eta,t)
=
\sum_{i=0}^{k}\sum_{j=0}^{k}
\hat{U}_{ij}^e(t)\,\ell_i(\xi)\ell_j(\eta),
\end{equation*}
where $k$ is the polynomial degree, and $\ell_i(\xi)$ and $\ell_j(\eta)$ are one-dimensional Lagrange basis functions. The same tensor-product construction can be extended naturally to three-dimensional hexahedral elements. Here, $\hat{U}_{ij}^e$ denotes the local degree of freedom coefficient associated with the corresponding nodal basis function. By stacking these coefficients at all interpolation nodes in the element according to a fixed ordering, one obtains the elemental degree of freedom vector $\hat{U}^e$. In other words, $\hat{U}^e$ is the column vector formed by all coefficients associated with the nodal basis functions on $\Omega^e$, and it contains all local unknowns of the discrete solution on that element. A key property of this nodal GLL construction is the summation-by-parts (SBP) relation between the GLL quadrature matrix and the nodal derivative matrix, which provides the discrete analogue of integration by parts and will be used in the conservation discussion below.

Let $\phi_h^e$ be an arbitrary local test function. The elemental weak form reads
\begin{equation*}
\int_{\Omega^e}\frac{\partial U_h^e}{\partial t}\phi_h^e\,d\Omega
-
\int_{\Omega^e}\bm{F}(U_h^e,\nabla U_h^e)\cdot\nabla\phi_h^e\,d\Omega
+
\int_{\partial\Omega^e}(\widehat{\bm{F}}\cdot\bm{n})\phi_h^e\,dS
=0.
\end{equation*}
Applying integration by parts yields the corresponding strong form
\begin{equation*}
\int_{\Omega^e}
\left(
\frac{\partial U_h^e}{\partial t}
+\nabla\cdot\bm{F}(U_h^e,\nabla U_h^e)
\right)\phi_h^e\,d\Omega
+
\int_{\partial\Omega^e}
\left(
\widehat{\bm{F}}-\bm{F}^{-}
\right)\cdot\bm{n}\,\phi_h^e\,dS
=0,
\end{equation*}
where $\bm{F}^{-}$ denotes the physical flux evaluated from the interior trace of the element solution.

To obtain the semi-discrete system corresponding to all local degrees of freedom in the element, we further collect the nodal basis functions in the same ordering as $\hat{U}^e$ and denote the resulting basis-function vector by $\bm{\varphi}^e$. Weighting the above strong form componentwise by $\bm{\varphi}^e$ and integrating over the element then gives the elemental semi-discrete system
\begin{equation*}
M^e\frac{d\hat{U}^e}{dt}
=
R_{\mathrm{vol}}^e(\hat{U}^e)+R_{\mathrm{surf}}^e(\hat{U}^{-},\hat{U}^{+}),
\end{equation*}
where $M^e$ is the elemental mass matrix, and the volume residual and interface correction are defined by
\begin{equation*}
R_{\mathrm{vol}}^e(\hat{U}^e)
=
-\int_{\Omega^e}\nabla\cdot\bm{F}(U_h^e,\nabla U_h^e)\,\bm{\varphi}^e\,d\Omega,
\end{equation*}
\begin{equation*}
R_{\mathrm{surf}}^e(\hat{U}^{-},\hat{U}^{+})
=
-\int_{\partial\Omega^e}
\left(
\widehat{\bm{F}}-\bm{F}^{-}
\right)\cdot\bm{n}\,\bm{\varphi}^e\,dS.
\end{equation*}
Here, both the volume residual and the interface correction are vector-valued residuals associated with all local degrees of freedom of the element. Therefore, their weighting object is no longer an arbitrary scalar test function, but rather the basis-function vector $\bm{\varphi}^e$ formed by the nodal basis functions in an ordering consistent with $\hat{U}^e$. With this definition, every term in the semi-discrete system is dimensionally and structurally consistent with the elemental degree of freedom vector $\hat{U}^e$, which is convenient for the subsequent construction of time integration schemes.

We now explain the treatment of the mass matrix. If the elemental $L^2$ inner product is retained in its standard form, the general mass matrix is defined by
\begin{equation*}
(M^e)_{\alpha\beta}
=
\int_{\Omega^e}\varphi_\alpha^e\,\varphi_\beta^e\,d\Omega,
\end{equation*}
where $\{\varphi_\alpha^e\}$ denotes the local nodal basis functions on the element.

In the present work, the mass matrix is evaluated by numerical quadrature on the GLL points instead of exact integration, which yields a diagonal mass matrix, referred to here as the lumped mass matrix. Accordingly, it is written as
\begin{equation*}
M_L^e=\mathrm{diag}(M_\alpha^e),
\end{equation*}
where $M_\alpha^e$ denotes the diagonal mass coefficient associated with the $\alpha$-th local degree of freedom. The elemental semi-discrete system therefore becomes
\begin{equation*}
M_L^e\frac{d\hat{U}^e}{dt}
=
R_{\mathrm{vol}}^e(\hat{U}^e)+R_{\mathrm{surf}}^e(\hat{U}^{-},\hat{U}^{+}).
\end{equation*}

The main purpose of this treatment is to simplify the inversion of the mass matrix in time integration, improve the compactness and data locality of the semi-discrete system, and provide a convenient algebraic structure for future extensions to implicit--explicit hybrid time integration. Under the present nodal GLL discretization, the interface correction acts directly on the corresponding boundary degrees of freedom.

Based on this unified spatial-discretization framework, the following subsections specify the numerical treatment of the weakly compressible convective flux, the IIPG discretization of the viscous terms, and the semi-discrete formulation compatible with local time stepping.

\subsection{Numerical flux evaluation}

In the present DG discretization, the coupling between adjacent elements is realized through numerical fluxes defined on the common interfaces. For the convective part, two different Riemann-type numerical fluxes are implemented, namely the local Lax--Friedrichs (local LF) flux and an approximate Riemann flux based on the two-rarefaction-wave assumption. The former introduces interfacial dissipation through a unified local wave-speed estimate, whereas the latter makes use of the characteristic structure of the weakly compressible model and evaluates the numerical flux through an interfacial star state. For the viscous part, an IIPG-type numerical diffusive flux is employed to enhance the stability of the diffusion discretization.

Before introducing the specific flux formulas, a local normal--tangential coordinate system is first established on the interface. Let the unit normal and tangential vectors be
\begin{equation}
\bm{n}=(n_x,n_y)^T,\qquad
\bm{t}=(-n_y,n_x)^T,
\end{equation}
respectively. For any velocity $\bm{u}=(u,v)^T$, the normal velocity and tangential velocity are defined as
\begin{equation}
q=\bm{u}\cdot\bm{n},\qquad
s=\bm{u}\cdot\bm{t},
\end{equation}
so that
\begin{equation*}
\bm{u}=q\bm{n}+s\bm{t}.
\end{equation*}
For convenience, once the face normal is fixed, the interior and exterior traces on the interface are also denoted by the left and right states, respectively, i.e. $U_L=U^-$ and $U_R=U^+$. Accordingly, the normal velocities on the two sides of the interface are denoted by
\begin{equation*}
q_L=\bm{u}_L\cdot\bm{n},\qquad
q_R=\bm{u}_R\cdot\bm{n},
\end{equation*}
and the tangential velocities are denoted by
\begin{equation*}
s_L=\bm{u}_L\cdot\bm{t},\qquad
s_R=\bm{u}_R\cdot\bm{t}.
\end{equation*}

\subsubsection*{(1) Local LF convective numerical flux}

Let $U$ denote the conservative variables and $\bm{F}_c(U)$ the convective flux tensor. The normal convective flux is defined as
\begin{equation}
\bm{F}_{c,n}(U)=\bm{F}_c(U)\cdot\bm{n}.
\end{equation}
For the left and right states $U_L$ and $U_R$, the local LF numerical flux is written as
\begin{equation}
\widehat{\bm{F}}_{c,n}^{\,LF}
=
\frac{1}{2}\Big(\bm{F}_{c,n}(U_L)+\bm{F}_{c,n}(U_R)\Big)
-\frac{1}{2}\alpha(U_R-U_L),
\end{equation}
where $\alpha$ is a local estimate of the maximum characteristic speed on the interface. For the present weakly compressible model, $\alpha$ is taken as
\begin{equation}
\alpha=\max\big(|q_L|+c_0,\;|q_R|+c_0\big),
\end{equation}
where $c_0$ is the prescribed artificial sound speed.

\subsubsection*{(2) Convective numerical flux based on the two-rarefaction-wave assumption}

In addition to the local LF flux, a two-rarefaction approximate Riemann flux is also employed. This assumption is motivated by the fact that the present work mainly concerns low-Mach-number flows, for which strong shock structures are generally absent. Therefore, the local interfacial Riemann problem can be reasonably approximated by a wave pattern composed of two rarefaction waves. The cylinder-flow results presented later suggest that, compared with the local LF flux, this flux produces weaker far-field pressure fluctuations in the present open-boundary test.

The two-rarefaction approximation is a standard approximate-state Riemann-solver idea in compressible-flow computations \cite{Toro2009}, and related two-rarefaction or exact Riemann solvers have also been used in artificial-compressibility, artificial-equation-of-state, and weakly compressible formulations \cite{ElsworthToro1992,Massa2022,ZhangHuAdams2017}. In the present work, this idea is specialized to the constant-sound-speed barotropic equation of state \eqref{eq:specific_eos}. Owing to the explicit logarithmic Riemann invariants of this model, the interfacial star density and normal velocity can be written in closed form. These closed-form interfacial states are used to define a low-dissipation flux for smooth weakly compressible low-Mach-number flows within the conservative LTS-DG framework.

The present specialization considers a local one-dimensional Riemann problem in the interface-normal direction and determines the interfacial star state from the invariance of the Riemann invariants across rarefaction waves. Since the present model is weakly compressible, the artificial sound speed $c_0$ is much larger than the characteristic flow velocity, namely
\begin{equation}
c_0\gg |\bm{u}|.
\end{equation}
As a result, for the local normal system the acoustic characteristic speeds satisfy
\begin{equation}
\lambda_1=q-c_0<0,\qquad
\lambda_3=q+c_0>0.
\end{equation}
This implies that, under the present weakly compressible assumption, the interface location $\hat{x}/t=0$ lies between the left- and right-running acoustic waves, i.e. within the star region, where $\hat{x}$ denotes the local coordinate along the interface-normal direction. Accordingly, under the present weakly compressible assumption, the interfacial state at $\hat{x}/t=0$ is identified with the star state, and the numerical convective flux is evaluated from that state. Figure~\ref{fig:riemann-star-region} illustrates the corresponding local two-dimensional Riemann structure in the $\hat{x}$--$t$ plane.

\begin{figure}[htbp]
\centering
\includegraphics[width=0.62\textwidth]{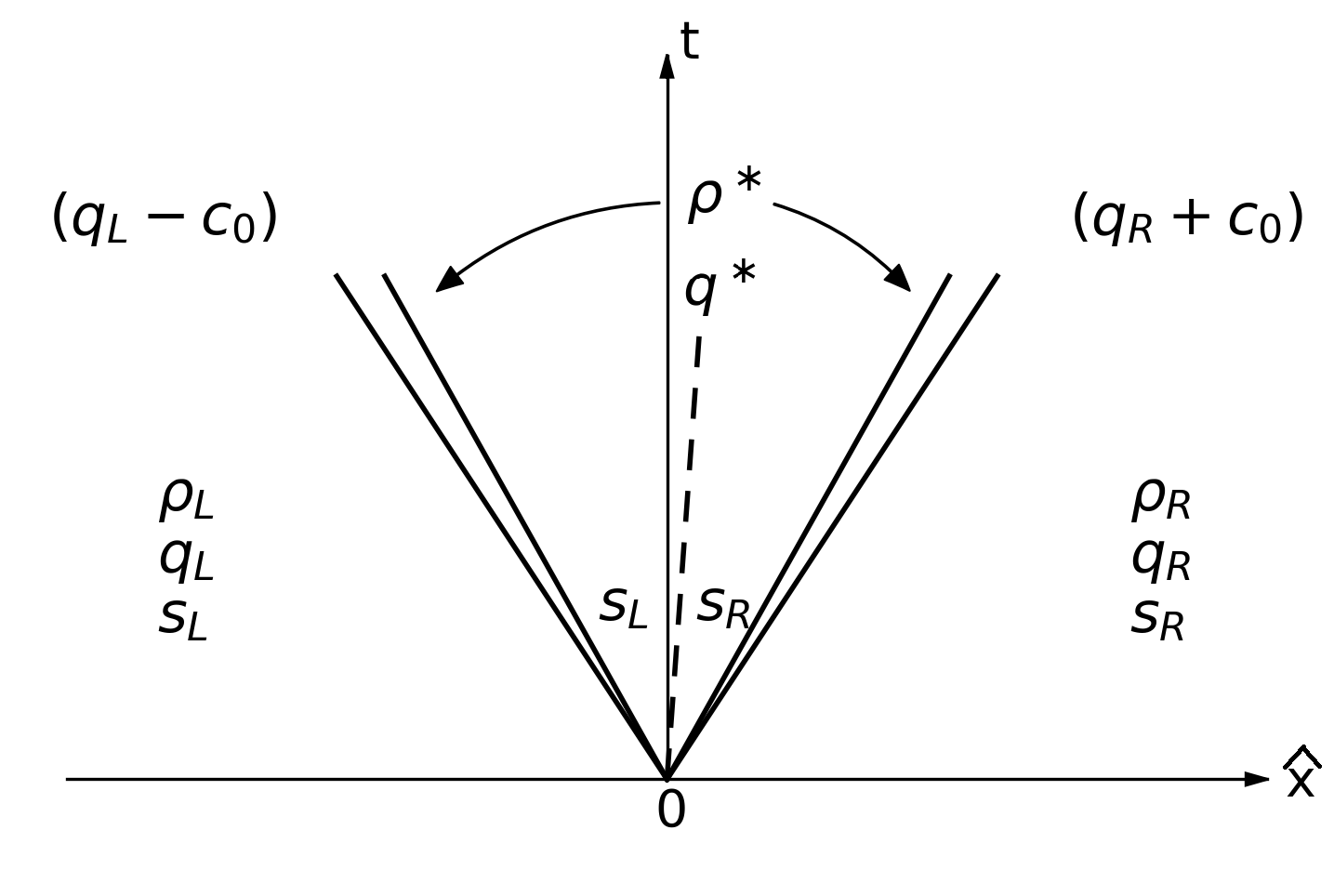}
\caption{Schematic of the local two-dimensional Riemann problem under the two-rarefaction-wave assumption. Under the present weakly compressible assumption, the interfacial state at $\hat{x}/t=0$ is identified with the star state located between the left- and right-running acoustic waves.}
\label{fig:riemann-star-region}
\end{figure}

In the above local normal coordinate system, the weakly compressible convective subsystem is treated as a one-dimensional hyperbolic system in terms of $(\rho,q,s)^T$, with characteristic speeds
\begin{equation}
\lambda_1=q-c_0,\qquad
\lambda_2=q,\qquad
\lambda_3=q+c_0.
\end{equation}
A set of Riemann invariants can then be chosen as
\begin{equation}
w_1=q-c_0\ln\left(\frac{\rho}{\rho_0}\right),\qquad
w_2=s,\qquad
w_3=q+c_0\ln\left(\frac{\rho}{\rho_0}\right),
\end{equation}
where $\rho_0$ is the reference density.

These invariants are related to the three characteristic fields as follows: across a 1-wave, $w_2$ and $w_3$ remain invariant; across a 2-wave, $w_1$ and $w_3$ remain invariant; across a 3-wave, $w_1$ and $w_2$ remain invariant. Therefore, under the two-rarefaction-wave approximation, with a 1-rarefaction wave connecting the left state to the star region and a 3-rarefaction wave connecting the right state to the star region, one has
\begin{equation}
q^\ast+c_0\ln\left(\frac{\rho^\ast}{\rho_0}\right)
=
q_L+c_0\ln\left(\frac{\rho_L}{\rho_0}\right),
\end{equation}
and
\begin{equation}
q^\ast-c_0\ln\left(\frac{\rho^\ast}{\rho_0}\right)
=
q_R-c_0\ln\left(\frac{\rho_R}{\rho_0}\right).
\end{equation}
For convenience, define
\begin{equation*}
r_L=q_L+c_0\ln\left(\frac{\rho_L}{\rho_0}\right),\qquad
r_R=q_R-c_0\ln\left(\frac{\rho_R}{\rho_0}\right).
\end{equation*}

\paragraph{Star-state normal velocity and density.}

Adding the above two relations yields the star-state normal velocity
\begin{equation}
q^\ast=\frac{1}{2}(r_L+r_R).
\end{equation}
Subtracting them gives
\begin{equation}
2c_0\ln\left(\frac{\rho^\ast}{\rho_0}\right)=r_L-r_R,
\end{equation}
from which the star-state density is obtained as
\begin{equation}
\rho^\ast=\rho_0\exp\left(\frac{r_L-r_R}{2c_0}\right).
\end{equation}

\paragraph{Star-state tangential velocity and velocity reconstruction.}

Once the star-state normal velocity $q^\ast$ is obtained, the tangential velocity is selected according to the upwind rule,
\begin{equation}
s^\ast=
\begin{cases}
s_L, & q^\ast>0,\\
s_R, & q^\ast\le 0.
\end{cases}
\end{equation}
The star-state velocity is then reconstructed as
\begin{equation*}
\bm{u}^\ast=q^\ast\bm{n}+s^\ast\bm{t}.
\end{equation*}

\paragraph{Remark.}

For additional transported variables such as temperature, no separate intermediate state is constructed from the acoustic Riemann relations. Instead, they are taken according to the upwind direction of the interfacial mass transport, i.e. the left state is used when $q^\ast>0$, and the right state is used when $q^\ast\le 0$. This treatment remains consistent with the convective transport direction at the interface while avoiding the introduction of additional approximate wave structures for non-dominant variables.

\paragraph{Star-state variables and numerical convective flux.}

The star-state conservative variables $U^\ast$ are then constructed from $\rho^\ast$ and $\bm{u}^\ast$, while the corresponding pressure is evaluated from the equation of state,
\begin{equation*}
p^\ast=p(\rho^\ast).
\end{equation*}
The numerical convective flux under the two-rarefaction-wave approximation is finally defined as the physical normal flux evaluated at the star state,
\begin{equation}
\widehat{\bm{F}}_{c,n}^{\,TR}
=
\bm{F}_{c,n}(U^\ast).
\end{equation}
Compared with the local LF flux, this flux evaluates the interfacial state from an approximate solution of the local Riemann problem and thereby more naturally represents acoustic propagation, density variation, and pressure coupling in weakly compressible flows.

\paragraph{Remark.}

The two-rarefaction flux used here is designed for smooth weakly compressible low-Mach-number flows, for which $c_0\gg |\bm{u}|$ and density fluctuations remain small. It is not intended as a general-purpose shock-capturing Riemann solver for strongly compressible, transonic, or near-vacuum flows.

\begin{Pro}[Consistency and interface-orientation consistency of the TR flux]
Let $\widehat{\bm{F}}_{c,n}^{\,TR}(U_L,U_R;\bm n)$ denote the TR numerical flux evaluated with unit normal $\bm n$ pointing from the left state to the right state. For any admissible state $U$,
\begin{equation*}
\widehat{\bm{F}}_{c,n}^{\,TR}(U,U;\bm n)=\bm{F}_c(U)\cdot\bm n .
\end{equation*}
Moreover, if the same interface is viewed from the opposite side, namely the normal is changed to $-\bm n$ and the left and right states are exchanged, then
\begin{equation*}
\widehat{\bm{F}}_{c,n}^{\,TR}(U_R,U_L;-\bm n)
=
-\widehat{\bm{F}}_{c,n}^{\,TR}(U_L,U_R;\bm n).
\end{equation*}
\end{Pro}
\begin{proof}
When $U_L=U_R=U$, one has $q_L=q_R=q$, $\rho_L=\rho_R=\rho$, and $s_L=s_R=s$. Hence
\begin{equation*}
r_L=q+c_0\ln\left(\frac{\rho}{\rho_0}\right),\qquad
r_R=q-c_0\ln\left(\frac{\rho}{\rho_0}\right),
\end{equation*}
which gives $q^\ast=q$ and $\rho^\ast=\rho$. The upwind selection also gives $s^\ast=s$, and therefore $U^\ast=U$. The consistency relation follows directly from the definition of the flux.

For the orientation property, let $\bm n'=-\bm n$ and $\bm t'=-\bm t$, and exchange the left and right states. The corresponding normal and tangential velocities satisfy
\begin{equation*}
q'_L=-q_R,\qquad q'_R=-q_L,\qquad
s'_L=-s_R,\qquad s'_R=-s_L .
\end{equation*}
Consequently,
\begin{equation*}
r'_L=-r_R,\qquad r'_R=-r_L,
\end{equation*}
and thus $q'^\ast=-q^\ast$ and $\rho'^\ast=\rho^\ast$. For $q^\ast\ne0$, the upwind rule gives $s'^\ast=-s^\ast$, so the reconstructed physical velocity is unchanged:
\begin{equation*}
\bm u'^\ast=q'^\ast\bm n'+s'^\ast\bm t'
=q^\ast\bm n+s^\ast\bm t=\bm u^\ast .
\end{equation*}
If $q^\ast=0$, the advective part of the normal convective flux is proportional to $q^\ast$, and the remaining pressure contribution is independent of the tangential velocity; hence the same conclusion holds at the level of the normal flux. Therefore the star-state physical flux is the same while the normal direction is reversed, which yields
\[
\bm F_c(U^\ast)\cdot(-\bm n)
=-\bm F_c(U^\ast)\cdot\bm n .
\]
This proves the stated interface-orientation consistency.
\end{proof}

\subsubsection*{(3) Boundary treatment}

Only solid-wall and far-field boundaries are considered in the present computations. Boundary faces are treated within the same numerical-flux framework as interior faces by constructing a ghost-neighbor state and a ghost-neighbor derivative value at each boundary quadrature point. Let $U^-$ and $(\nabla U)^-$ denote the interior trace and its derivative on a boundary face. For a far-field boundary, the exterior boundary state is prescribed by the far-field state,
\begin{equation}
U_b^+=U_\infty,
\end{equation}
while the derivative value used in the diffusive flux is extrapolated from the interior element,
\begin{equation}
(\nabla U)_b^+=(\nabla U)^- .
\end{equation}
For a solid wall with wall velocity $\bm{u}_w$, the exterior density is copied from the interior trace and the exterior velocity is reflected with respect to the wall velocity,
\begin{equation}
\rho_b^+=\rho^-,\qquad
\bm{u}_b^+=2\bm{u}_w-\bm{u}^- .
\end{equation}
Thus, for a stationary wall, $\bm{u}_b^+=-\bm{u}^-$. The conservative exterior state is then assembled from $\rho_b^+$ and $\bm{u}_b^+$, and the derivative value on the exterior side is again taken from the interior derivative unless additional wall-gradient data are prescribed. After these ghost-neighbor values are constructed, the convective Riemann flux and the IIPG diffusive flux are evaluated in exactly the same way as on an interior face.

\subsubsection*{(4) IIPG diffusive numerical flux}

For the diffusive part, an IIPG-type numerical diffusive flux is adopted. Let the normal diffusive physical flux be
\begin{equation}
\bm{F}_{d,n}(U,\nabla U)=\bm{F}_d(U,\nabla U)\cdot\bm{n}.
\end{equation}
Since the diffusive flux has a zero mass component and acts only on the momentum equations, the penalty term is also applied only to the momentum components. For the two-dimensional system considered above, with the same idea extending to three dimensions, define
\begin{equation*}
\bm{D}_{\mathrm{m}}=\operatorname{diag}(0,1,1),
\end{equation*}
where the first entry corresponds to the mass-conservation equation.
The numerical diffusive flux on the interface is written as
\begin{equation}
\widehat{\bm{F}}_{d,n}
=
\{\!\!\{\bm{F}_{d,n}\}\!\!\}
-\tau\bm{D}_{\mathrm{m}}[\![U]\!],
\end{equation}
where $\{\!\!\{\cdot\}\!\!\}$ denotes the average of the two traces, $[\![U]\!]$ the jump of the solution, and $\tau$ the penalty parameter.

For a polynomial approximation of degree $k$, the penalty parameter is taken as
\begin{equation}
\tau=C_{\mathrm{IP}}\frac{\mu}{\rho_0}\frac{(k+1)^2}{h_f},
\end{equation}
where $h_f$ is the characteristic length associated with the interface, and $C_{\mathrm{IP}}$ is a dimensionless constant depending on the mesh type and the stability requirement.

For nonuniform meshes, let $h_L$ and $h_R$ denote the local characteristic lengths of the left and right elements in the interface-normal direction. One may further choose
\begin{equation*}
h_f=\min(h_L,h_R),
\end{equation*}
so that
\begin{equation}
\tau=C_{\mathrm{IP}}\frac{\mu}{\rho_0}\frac{(k+1)^2}{\min(h_L,h_R)}.
\end{equation}
Here, $h_L$ and $h_R$ represent the local mesh sizes of the two adjacent elements measured along the interface-normal direction. Taking the smaller one helps provide sufficient penalization on the finer side of the interface, which is beneficial for maintaining the stability of the diffusive discretization.

The above penalty term weakly suppresses inter-element jumps in the momentum equations and thereby improves the stability of the diffusion discretization without adding a penalty contribution to the mass equation. Therefore, the interfacial numerical-flux treatment adopted in the present work can be summarized as follows: for the convective part, either the local LF flux or the two-rarefaction-wave approximate Riemann flux is used; for the diffusive part, the IIPG numerical diffusive flux is employed.

\subsection{Time-Accurate Local Time-Stepping Formulation}
\label{sec:3.3}

To support long-time integration accuracy for weakly compressible low-Mach-number flows together with element-local time advancement on nonuniform meshes, a time-accurate local time-stepping method is constructed on the semi-discrete DG system, following the predictor--corrector philosophy of time-accurate local time stepping developed by Gassner and co-workers \cite{GassnerHindenlangMunz2011,GassnerDumbserHindenlangMunz2011}. The algorithm proceeds as follows. For each element, a local time step is selected independently. Inside the element, a continuous explicit Runge--Kutta method is used to construct a continuous-in-time predictor driven by the volume residual. In the subsequent correction, the face contribution is split into a purely interior part and a common-flux part. The purely interior face contribution is local to the element and is integrated with the same temporal quadrature as the volume contribution, whereas the common-flux contribution is integrated face by face using piecewise temporal quadrature based on the time matching between neighboring elements. Once the element reaches its target time, these integrated contributions are combined to obtain the updated elemental degrees of freedom.

Based on the spatial discretization introduced above, the semi-discrete equation on element $\Omega^e$ is written as
\begin{equation}
M_L^e \frac{d\hat U^e}{dt}
=
R_{\mathrm{vol}}^e(\hat U^e)
+
R_{\mathrm{surf}}^e(\hat U^{e,-},\hat U^{e,+}).
\label{eq:lts_semidiscrete}
\end{equation}
Here $M_L^e$ is the lumped mass matrix, $R_{\mathrm{vol}}^e$ denotes the elemental volume residual, and $R_{\mathrm{surf}}^e$ denotes the surface residual assembled from the interface correction terms.

For convenience, the semi-discrete operators associated with the volume and surface terms are defined by
\begin{equation}
\mathcal L^e(\hat U^e)
=
(M_L^e)^{-1}R_{\mathrm{vol}}^e(\hat U^e),
\label{eq:lts_volume_operator}
\end{equation}
\begin{equation}
\mathcal C^e(\hat U^{e,-},\hat U^{e,+})
=
(M_L^e)^{-1}R_{\mathrm{surf}}^e(\hat U^{e,-},\hat U^{e,+}).
\label{eq:lts_surface_operator}
\end{equation}
Then the ordinary differential equation on element $\Omega^e$ can be written as
\begin{equation}
\frac{d\hat U^e}{dt}
=
\mathcal L^e(\hat U^e)
+
\mathcal C^e(\hat U^{e,-},\hat U^{e,+}).
\label{eq:lts_ode}
\end{equation}

If the boundary of the element is decomposed into individual faces $f\subset\partial\Omega^e$, the surface operator can be further written as
\begin{equation}
\mathcal C^e(\hat U^{e,-},\hat U^{e,+})
=
\sum_{f\subset\partial\Omega^e}
\left[
\mathcal C_{f,\mathrm{int}}^e(\hat U^{e,-})
+
\mathcal C_{f,\mathrm{cf}}^e(\hat U^{e,-},\hat U^{e,+})
\right],
\label{eq:lts_surface_face_split}
\end{equation}
where $\mathcal C_{f,\mathrm{int}}^e$ denotes the purely interior contribution on face $f$, associated with the interior physical flux, and $\mathcal C_{f,\mathrm{cf}}^e$ denotes the common-flux contribution, associated with the numerical flux shared by the two elements adjacent to the face.

Each element $\Omega^e$ is assigned an independent local time step $\Delta t^e$, whose value is determined by the local mesh scale, the flow state, and the stability constraints. Accordingly, the current local time interval of element $\Omega^e$ is denoted by
\begin{equation}
I_n^e=[t_n^e,\; t_{n+1}^e],\qquad
t_{n+1}^e=t_n^e+\Delta t^e.
\label{eq:lts_local_interval}
\end{equation}
Since neighboring elements may use different local time steps, the common-flux contribution generally requires the reconstruction of neighboring states at requested times. Therefore, a continuous-in-time predictor plays a key role in the local time-stepping procedure.

On the interval $I_n^e$, the normalized local time variable is introduced as
\begin{equation}
\theta=\frac{t-t_n^e}{\Delta t^e},\qquad 0\le \theta\le 1.
\label{eq:lts_theta}
\end{equation}

For element $\Omega^e$, an $s$-stage continuous explicit Runge--Kutta predictor is constructed as in the CERK-based local predictor framework \cite{OwrenZennaro1992,GassnerHindenlangMunz2011,GassnerDumbserHindenlangMunz2011}
\begin{equation}
\mathcal P^e(t)
=
\hat U^{e,n}
+
\Delta t^e\sum_{i=1}^{s} b_i(\theta)\,K_i^e,
\qquad t\in I_n^e.
\label{eq:lts_cerk_predictor}
\end{equation}
Here $b_i(\theta)$ are the continuous extension weights, and $K_i^e$ are the stage residuals. In the predictor step, the stage residuals are generated only by the volume operator, namely,
\begin{equation}
K_i^e
=
\mathcal L^e\!\left(
\hat U^{e,n}
+
\Delta t^e\sum_{j=1}^{i-1} a_{ij}K_j^e
\right),
\qquad i=1,2,\ldots,s.
\label{eq:lts_stage_residual}
\end{equation}

The resulting continuous predictor $\mathcal P^e(t)$ describes the local temporal evolution of the elemental degree of freedom vector driven by the volume term. In the present formulation, the volume contribution and the purely interior face contribution are both local to element $\Omega^e$, and are therefore evaluated with the same Gaussian quadrature over the local time interval $[t_n^e,t_{n+1}^e]$. The common-flux contribution on a given face, however, depends on the neighboring predictor state. Its temporal integration interval is therefore partitioned into subintervals according to the local time levels of the two adjacent elements, and Gaussian quadrature is applied separately on each subinterval.

When element $\Omega^e$ reaches the target time $t_{n+1}^e$, its updated degrees of freedom are obtained from the time-integrated volume and surface contributions as
\begin{equation}
\begin{aligned}
&\hat U^{e,n+1}
=
\hat U^{e,n}
+
\int_{t_n^e}^{t_{n+1}^e}\mathcal L^e(\mathcal P^e(t))\,dt
\\
&\quad
+
\sum_{f\subset\partial\Omega^e}
\int_{t_n^e}^{t_{n+1}^e}
\mathcal C_{f,\mathrm{int}}^e\bigl(\mathcal P^{e,-}(t)\bigr)\,dt
\\
&\quad
+
\sum_{f\subset\partial\Omega^e}
\int_{t_n^e}^{t_{n+1}^e}
\mathcal C_{f,\mathrm{cf}}^e\bigl(\mathcal P^{e,-}(t),\mathcal P^{e,+}(t)\bigr)\,dt.
\end{aligned}
\label{eq:lts_update_integral}
\end{equation}

Accordingly, the accumulated local contribution over the interval $[t_n^e,t_{n+1}^e]$, consisting of the volume term and the purely interior face terms, can be written as
\begin{equation*}
\int_{t_n^e}^{t_{n+1}^e}
\left[
\mathcal L^e(\mathcal P^e(t))
+
\sum_{f\subset\partial\Omega^e}
\mathcal C_{f,\mathrm{int}}^e(\mathcal P^{e,-}(t))
\right]dt.
\end{equation*}
Using the same temporal quadrature points as for the volume term, one obtains
\begin{equation*}
\begin{aligned}
&
\int_{t_n^e}^{t_{n+1}^e}
\left[
\mathcal L^e(\mathcal P^e(t))
+
\sum_{f\subset\partial\Omega^e}
\mathcal C_{f,\mathrm{int}}^e(\mathcal P^{e,-}(t))
\right]dt
\\
&\quad
\approx
\sum_{m=1}^{N_{\mathrm{vol}}^e}
\omega_m^{\,e}\,
\left[
\mathcal L^e\bigl(\mathcal P^e(t_m^{\,e})\bigr)
+
\sum_{f\subset\partial\Omega^e}
\mathcal C_{f,\mathrm{int}}^e\bigl(\mathcal P^{e,-}(t_m^{\,e})\bigr)
\right],
\end{aligned}
\end{equation*}
where $t_m^{\,e}$ and $\omega_m^{\,e}$ denote the temporal quadrature points and weights used for the volume contribution in element $\Omega^e$. The same points and weights are used for $\mathcal C_{f,\mathrm{int}}^e$ because this term involves only the interior trace of the same element.

For any face $f\subset\partial\Omega^e$ of element $\Omega^e$, let
\begin{equation*}
\mathcal C_{f,\mathrm{cf}}^e(\hat U^{e,-},\hat U^{e,+})
\end{equation*}
denote the common-flux contribution to the surface operator. At time $t$, the states on the two sides of the interface are reconstructed from the continuous predictors of the neighboring elements as
\begin{equation*}
\mathcal P^{e,-}(t),\qquad \mathcal P^{e,+}(t),
\end{equation*}
which gives the common-flux contribution at time $t$ in the form
\begin{equation*}
\mathcal C_{f,\mathrm{cf}}^e\bigl(\mathcal P^{e,-}(t),\mathcal P^{e,+}(t)\bigr).
\end{equation*}

Since the temporal matching relations are generally different from face to face, the temporal quadrature points for the common-flux contribution are, in general, not identical on different faces. Therefore, the temporal integration of the common-flux term must be carried out separately on each face, which is consistent with the face-wise time-accurate LTS treatment advocated in \cite{GassnerHindenlangMunz2011}. For face $f$, if the interval associated with that face is partitioned into $N_{\mathrm{sub}}^f$ subintervals,
\begin{equation*}
[t_n^e,t_{n+1}^e]
=
\bigcup_{k=1}^{N_{\mathrm{sub}}^f}
[t_{k,\mathrm{L}}^{\,f},\,t_{k,\mathrm{R}}^{\,f}],
\end{equation*}
then the common-flux contribution is evaluated in a piecewise manner as
\begin{equation*}
\int_{t_n^e}^{t_{n+1}^e}
\mathcal C_{f,\mathrm{cf}}^e\bigl(\mathcal P^{e,-}(t),\mathcal P^{e,+}(t)\bigr)\,dt
=
\sum_{k=1}^{N_{\mathrm{sub}}^f}
\int_{t_{k,\mathrm{L}}^{\,f}}^{t_{k,\mathrm{R}}^{\,f}}
\mathcal C_{f,\mathrm{cf}}^e\bigl(\mathcal P^{e,-}(t),\mathcal P^{e,+}(t)\bigr)\,dt.
\end{equation*}
Applying Gaussian quadrature on each subinterval gives
\begin{equation*}
\int_{t_n^e}^{t_{n+1}^e}
\mathcal C_{f,\mathrm{cf}}^e\bigl(\mathcal P^{e,-}(t),\mathcal P^{e,+}(t)\bigr)\,dt
\approx
\sum_{k=1}^{N_{\mathrm{sub}}^f}
\sum_{g=1}^{N_t^{f,k}}
w_g^{\,f,k}\,
\mathcal C_{f,\mathrm{cf}}^e\bigl(\mathcal P^{e,-}(t_g^{\,f,k}),\mathcal P^{e,+}(t_g^{\,f,k})\bigr),
\end{equation*}
where $t_g^{\,f,k}$ and $w_g^{\,f,k}$ are the Gaussian quadrature points and weights on the $k^\text{th}$ subinterval of face $f$.

Figure~\ref{fig:lts_1d_time_quadrature} illustrates the temporal quadrature points used in a one-dimensional local time-stepping setting. The blue circles indicate the Gaussian quadrature points $t_m^{\,e}$ employed for both the volume contribution and the purely interior face contribution of element $\Omega^e$ over its local interval $[t_n^e,t_{n+1}^e]$, whereas the red triangles indicate the Gaussian quadrature points $t_g^{\,f,k}$ used for the piecewise temporal integration of the common-flux contribution. Owing to the mismatch between the local time step of $\Omega^e$ and those of its neighboring elements, the common-flux integrals are evaluated over several subintervals, and each subinterval is equipped with its own set of quadrature points.

\begin{figure}[htbp]
\centering
\includegraphics[width=0.42\textwidth]{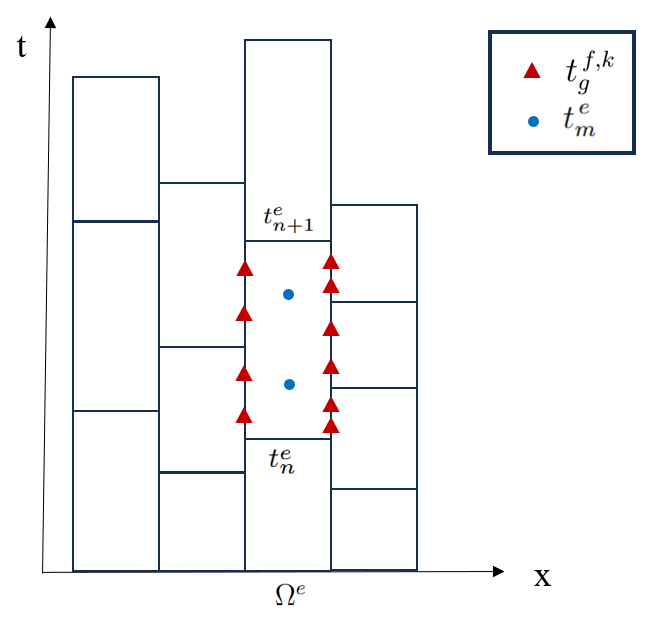}
\caption{Schematic illustration of temporal quadrature points for the volume and face contributions in a one-dimensional local time-stepping procedure. The blue circles denote the quadrature points $t_m^{\,e}$ used for the volume contribution and the purely interior face contribution of element $\Omega^e$, while the red triangles denote the quadrature points $t_g^{\,f,k}$ used for the common-flux contribution on the piecewise subintervals induced by the time-step mismatch across neighboring elements.}
\label{fig:lts_1d_time_quadrature}
\end{figure}

Accordingly, the total surface contribution over the local time interval is obtained by summing the purely interior and common-flux contributions from all faces, namely,
\begin{equation*}
\int_{t_n^e}^{t_{n+1}^e}
\mathcal C^e\bigl(\mathcal P^{e,-}(t),\mathcal P^{e,+}(t)\bigr)\,dt
=
\sum_{f\subset\partial\Omega^e}
\int_{t_n^e}^{t_{n+1}^e}
\left[
\mathcal C_{f,\mathrm{int}}^e\bigl(\mathcal P^{e,-}(t)\bigr)
+
\mathcal C_{f,\mathrm{cf}}^e\bigl(\mathcal P^{e,-}(t),\mathcal P^{e,+}(t)\bigr)
\right]dt.
\end{equation*}
Substituting the two temporal quadrature rules gives
\begin{equation*}
\begin{aligned}
&
\int_{t_n^e}^{t_{n+1}^e}
\mathcal C^e\bigl(\mathcal P^{e,-}(t),\mathcal P^{e,+}(t)\bigr)\,dt
\approx
\sum_{f\subset\partial\Omega^e}
\sum_{m=1}^{N_{\mathrm{vol}}^e}
\omega_m^{\,e}\,
\mathcal C_{f,\mathrm{int}}^e\bigl(\mathcal P^{e,-}(t_m^{\,e})\bigr)
\\
&\quad
+
\sum_{f\subset\partial\Omega^e}
\sum_{k=1}^{N_{\mathrm{sub}}^f}
\sum_{g=1}^{N_t^{f,k}}
w_g^{\,f,k}\,
\mathcal C_{f,\mathrm{cf}}^e\bigl(\mathcal P^{e,-}(t_g^{\,f,k}),\mathcal P^{e,+}(t_g^{\,f,k})\bigr).
\end{aligned}
\end{equation*}

The complete integral update has been given in \eqref{eq:lts_update_integral}. If these terms are written in quadrature form, the update formula becomes
\begin{equation}
\begin{aligned}
&\hat U^{e,n+1}
\approx
\hat U^{e,n}
+
\sum_{m=1}^{N_{\mathrm{vol}}^e}
\omega_m^{\,e}\,
\left[
\mathcal L^e\bigl(\mathcal P^e(t_m^{\,e})\bigr)
+
\sum_{f\subset\partial\Omega^e}
\mathcal C_{f,\mathrm{int}}^e\bigl(\mathcal P^{e,-}(t_m^{\,e})\bigr)
\right]
\\
&\quad
+
\sum_{f\subset\partial\Omega^e}
\sum_{k=1}^{N_{\mathrm{sub}}^f}
\sum_{g=1}^{N_t^{f,k}}
w_g^{\,f,k}\,
\mathcal C_{f,\mathrm{cf}}^e\bigl(\mathcal P^{e,-}(t_g^{\,f,k}),\mathcal P^{e,+}(t_g^{\,f,k})\bigr).
\end{aligned}
\label{eq:lts_update_quadrature}
\end{equation}

This update formula shows that the local right-hand-side contributions are separated according to their data dependence. The volume term and the purely interior face contribution are evaluated with the element-local quadrature points $t_m^{\,e}$, whereas only the common-flux contribution requires face-wise temporal matching and reconstruction of neighboring predictor states. The updated state at time $t_{n+1}^e$ is therefore obtained by combining the integrated volume contribution, the integrated purely interior face contribution, and the face-wise integrated common-flux contribution.

\subsection{Conservation of the split time integration}
\label{sec:lts_conservation}

We next verify that, even when different elements use different local time intervals, the split integration described above preserves discrete conservation by pairing the volume term with the purely interior face contribution and treating the common-flux contribution through face-wise time matching. Let $\mathbf{1}^e$ denote the vector corresponding to the constant test function on element $\Omega^e$. The discrete conserved quantity in this element is
\begin{equation}
Q^e(t)
=
(\mathbf{1}^e)^T M_L^e \hat U^e(t)
\approx
\int_{\Omega^e} U_h^e(t)\,d\Omega,
\label{eq:lts_element_conserved_quantity}
\end{equation}
where the relation is understood componentwise for the conservative variables. Acting with $(\mathbf{1}^e)^T M_L^e$ on the volume operator and the purely interior face contribution gives
\begin{equation}
\begin{aligned}
&(\mathbf{1}^e)^T M_L^e
\left[
\mathcal L^e(\hat U^e)
+
\sum_{f\subset\partial\Omega^e}
\mathcal C_{f,\mathrm{int}}^e(\hat U^{e,-})
\right]
\\
&\quad
=
-\int_{\Omega^e}
\nabla\cdot\bm F(U_h^e,\nabla U_h^e)\,d\Omega
+
\sum_{f\subset\partial\Omega^e}
\int_f
\bm F^-\cdot\bm n^e\,dS
=0.
\end{aligned}
\label{eq:lts_local_interior_cancellation}
\end{equation}
For the nodal DGSEM introduced in Section~\ref{subsec:spatial_dgsem}, the interpolation nodes and quadrature nodes coincide. In one spatial direction, let $W=\operatorname{diag}(\omega_0,\omega_1,\ldots,\omega_k)$ be the diagonal GLL quadrature matrix and let $D$ be the derivative matrix of the Lagrange basis evaluated at the GLL nodes. These operators satisfy the SBP identity
\begin{equation}
WD+D^T W=B,\qquad
B=\operatorname{diag}(-1,0,\ldots,0,1),
\end{equation}
which is the discrete analogue of integration by parts for the nodal GLL collocation form of DGSEM \cite{HesthavenWarburton2008}. After tensor-product assembly in multiple dimensions, this identity implies that, for a constant test function, the discrete volume divergence and the interior physical-flux part of the strong-form surface correction cancel elementwise, leaving only boundary or common-flux contributions. In the fully discrete implementation, \eqref{eq:lts_local_interior_cancellation} is precisely this SBP cancellation. It holds at each element-local temporal quadrature point because the same temporal quadrature is used for the volume term and the purely interior face contribution.

Thus, over any local interval $I_r^e=[t_r^e,t_{r+1}^e]$ of element $\Omega^e$, the change of the element conserved quantity is due only to the common-flux contributions:
\begin{equation}
Q^{e,r+1}-Q^{e,r}
=
-\sum_{f\subset\partial\Omega^e}
\int_{I_r^e}
\int_f
\widehat{\bm F}_{f,n}^e(t)\,dS\,dt,
\label{eq:lts_element_local_conservation_update}
\end{equation}
where $\widehat{\bm F}_{f,n}^e$ denotes the normal common flux on face $f$ with the outward normal of element $\Omega^e$. Equation~\eqref{eq:lts_element_local_conservation_update} is local in time; no global time interval is assumed in its definition.

To show global conservation, consider a synchronization window $[T^N,T^{N+1}]$ that contains several element-local intervals. For an interior face $f=\partial\Omega^e\cap\partial\Omega^{e'}$, form a face-time partition by taking the ordered union of all local time levels of the two adjacent elements within this window,
\begin{equation}
T^N=\tau_0^f<\tau_1^f<\cdots<\tau_{M_f}^f=T^{N+1}.
\label{eq:lts_face_time_partition}
\end{equation}
On each face-time slab $f\times[\tau_\ell^f,\tau_{\ell+1}^f]$, both adjacent continuous predictors are available, and the common flux is integrated once with the physical quadrature associated with that slab. With the outward normal of $\Omega^e$, define
\begin{equation}
\Phi_{f,\ell}^{e}
=
\int_{\tau_\ell^f}^{\tau_{\ell+1}^f}
\int_f
\widehat{\bm F}_{f,n}^{e}(t)\,dS\,dt
\approx
\sum_{g=1}^{N_t^{f,\ell}}
w_g^{\,f,\ell}
\int_f
\widehat{\bm F}_{f,n}^{e}(t_g^{\,f,\ell})\,dS .
\label{eq:lts_face_time_slab_flux}
\end{equation}
Because the same numerical flux is used on the same physical face and the two outward normals are opposite,
\begin{equation}
\widehat{\bm F}_{f,n}^{e'}(t)=-\widehat{\bm F}_{f,n}^{e}(t),
\qquad
\Phi_{f,\ell}^{e'}=-\Phi_{f,\ell}^{e}.
\label{eq:lts_common_flux_antisymmetry}
\end{equation}
The face-wise LTS quadrature therefore assigns equal and opposite contributions to the two elements for every face-time slab. These two contributions may be applied during different local element updates, or one of them may be stored until the neighboring element reaches the corresponding local time level, but the algebraic contribution associated with the slab is antisymmetric:
\begin{equation}
\Phi_{f,\ell}^{e}+\Phi_{f,\ell}^{e'}=0.
\label{eq:lts_pairwise_face_cancellation}
\end{equation}

Summing \eqref{eq:lts_element_local_conservation_update} over all local intervals of all elements contained in $[T^N,T^{N+1}]$, every interior face-time slab appears exactly twice with opposite signs. Hence all interior common-flux contributions cancel pairwise, and the total discrete conserved quantity at synchronization times satisfies
\begin{equation}
\sum_e Q^e(T^{N+1})-\sum_e Q^e(T^N)
=
-\sum_{f\subset\partial\Omega}
\sum_{\ell=0}^{M_f-1}
\int_{\tau_\ell^f}^{\tau_{\ell+1}^f}
\int_f
\widehat{\bm F}_{f,n}^{\,b}(t)\,dS\,dt,
\label{eq:lts_global_conservation}
\end{equation}
where the right-hand side contains only physical boundary faces. For periodic boundaries, the boundary faces are paired in the same way as interior faces; for closed boundaries with zero normal flux, the boundary contribution vanishes. Therefore, even though $t_r^e$ and $t_{r+1}^e$ are local quantities and differ from element to element, the split evaluation of the volume, purely interior face, and common-flux contributions is conservative over each synchronization window.

\begin{Pro}[Conservation of the split local time integration]
For the split LTS update described above, define the total discrete conserved quantity by
\begin{equation*}
Q_{\mathrm{tot}}(T)=\sum_e(\mathbf{1}^e)^T M_L^e \hat U^e(T).
\end{equation*}
Over any synchronization window $[T^N,T^{N+1}]$, the change of this quantity is produced only by the physical boundary faces,
\begin{equation*}
Q_{\mathrm{tot}}(T^{N+1})-Q_{\mathrm{tot}}(T^N)
=
-\sum_{f\in\mathcal F_{\mathrm b}}
\int_{T^N}^{T^{N+1}}
\int_f
\widehat{\bm F}_{f,n}^{\,b}(t)\,dS\,dt,
\end{equation*}
where $\mathcal F_{\mathrm b}$ denotes the set of physical boundary faces. In particular, in the absence of physical boundary fluxes, for example under periodic boundary conditions, the total discrete conserved quantity is unchanged.
\end{Pro}

\paragraph{Remark.}

The split treatment used here should be distinguished from the CERK-based LTS formulation of Gassner and co-workers \cite{GassnerHindenlangMunz2011,GassnerDumbserHindenlangMunz2011}. In their modal DG setting, the volume contribution is incorporated into the element-local CERK predictor, while the face contribution is advanced as an unsplit surface term. Reusing the notation introduced above, this can be written schematically as
\begin{equation}
\hat U^{e,n+1}
=
\mathcal P^e(t_{n+1}^e)
+
\int_{t_n^e}^{t_{n+1}^e}
\mathcal C^e\bigl(\mathcal P^{e,-}(t),\mathcal P^{e,+}(t)\bigr)\,dt .
\label{eq:gassner_lts_unsplit_surface}
\end{equation}
Here $\mathcal P^e(t_{n+1}^e)$ is the continuous predictor defined in \eqref{eq:lts_cerk_predictor}, evaluated at the end of the local interval. For a modal DG discretization, this form is natural because the cell-average equation is associated directly with the constant modal coefficient, and the surface term represents the conservative exchange across element interfaces. In the nodal GLL DG discretization used in this work, however, the cancellation between the volume divergence and the interior physical-flux part of the strong-form surface correction is a discrete summation-by-parts cancellation. If the unsplit surface term in \eqref{eq:gassner_lts_unsplit_surface} is applied directly together with face-wise LTS quadrature, the purely interior part $\mathcal C_{f,\mathrm{int}}^e$ in \eqref{eq:lts_surface_face_split} is no longer evaluated at the same temporal quadrature points as the volume term, and the cancellation in \eqref{eq:lts_local_interior_cancellation} is generally lost. Therefore, the present formulation separates the purely interior face contribution from the common-flux contribution: the former is integrated together with the volume term using element-local quadrature, while only the latter is integrated with face-wise temporal matching. This modification is introduced to preserve the discrete conservation property of the nodal DG scheme under local time stepping.

\subsection{Algorithmic organization for CPU--GPU execution}
\label{subsec:cpu_gpu_algorithm}

The CPU--GPU implementation follows the data-dependence structure of the local time-stepping algorithm rather than introducing a separate numerical formulation. During each advancement stage, the CPU maintains mesh and geometric data, identifies the active element groups, builds the face-time tasks required by the common-flux quadrature, and dispatches batches of element-local or face-local work. The GPU evaluates the tensor-product interpolation, differentiation, volume residuals, predictor stages, purely interior face contributions, and common-flux quadrature for the dispatched batches. This organization allows CPU-side task preparation and GPU-side numerical evaluation to overlap, while keeping the device-resident working set limited to the active computational units \cite{LiuLu2019KeplerGPU}.

The resulting algorithm is naturally split into two types of kernels. Element-local kernels require only data from a single element and cover the predictor construction, modal decomposition when used, volume-residual evaluation, and the purely interior part of the face correction. Face-local kernels require data from the two adjacent predictors and are used only for the common-flux contribution on each face-time slab. This separation mirrors the split residual formulation and reduces the amount of neighbor-dependent communication.

For tensor-product discretizations on quadrilateral and hexahedral elements, the multidimensional operations are further evaluated by applying one-dimensional operators successively in each coordinate direction. Thus derivatives, interpolations, modal transforms, quadrature evaluations, and related operator applications are assembled from one-dimensional data rather than full multidimensional matrices. This is the computational form used by the GPU kernels and reduces both memory traffic and storage requirements for operator data.

\section{Numerical Experiments}
\label{sec:results}

In this section, a series of numerical experiments are carried out to assess the accuracy, robustness, and applicability of the proposed weakly compressible DG method with time-accurate local time stepping. The tests include a one-dimensional weakly compressible simple wave and a two-dimensional weakly compressible shear wave for accuracy and conservation verification, the $Re=100$ flow past a circular cylinder for numerical-flux comparison, the lid-driven cavity flow, and a three-dimensional 30P30N slat noise benchmark case. These examples are chosen to examine different aspects of the method, including formal order of accuracy, conservation under local time stepping, pressure-field behavior of the numerical fluxes in the cylinder test, viscous-flow resolution, and applicability to a complex three-dimensional flow configuration.

Unless otherwise stated, the weakly compressible equation of state \eqref{eq:specific_eos} is used in all computations. The spatial discretization is based on the strong-form nodal DGSEM formulation described in Section~3, and the temporal advancement is performed by the proposed CERK-based time-accurate local time-stepping procedure. The element-local time step is selected to satisfy
\begin{equation*}
\Delta t_e\le \mathrm{CFL}\,
\min\!\left\{
\frac{h_e}{(2k+1)\max_{\Omega^e}(\|\bm{u}\|+c)},\;
\frac{h_e^2}{(2k+1)^2(\mu/\rho_0)}
\right\},
\end{equation*}
where $h_e$ is the characteristic size of element $\Omega^e$, and $c=\sqrt{p'(\rho)}=c_0$ for the equation of state \eqref{eq:specific_eos}; the second term is omitted for inviscid cases. All numerical examples use $\mathrm{CFL}=0.8$. In each example, the degree of the temporal predictor polynomial constructed by the CERK continuous extension is chosen to match the spatial DG polynomial degree $k$. For viscous flows, the IIPG method is employed with a penalty constant $C_{\mathrm{IP}}=1.0$. Unless otherwise specified, the convective flux is handled by the two-rarefaction approximate Riemann flux.

\subsection{One-dimensional weakly compressible simple wave}
\label{subsec:wc_shear_1d}

The first verification problem is a one-dimensional periodic calculation on $x\in[0,2\pi]$ with conservative state $U=(\rho,\rho u)^T$. The initial condition is
\begin{equation*}
\rho(x,0)=\rho_0(1+\epsilon\sin x),\qquad
u(x,0)=u_0+c_0\ln\frac{\rho(x,0)}{\rho_0},\qquad
p(x,0)=c_0^2(\rho(x,0)-\rho_0)+p_0 .
\end{equation*}
Thus the Riemann invariant $u-c_0\ln(\rho/\rho_0)=u_0$ is constant, and the smooth right-running exact solution is obtained from the characteristic footpoint $\xi$ satisfying
\begin{equation*}
x=\xi+\big(u(\xi,0)+c_0\big)t\pmod{2\pi}.
\end{equation*}
The exact state is
\begin{equation*}
\rho(x,t)=\rho(\xi,0),\qquad
u(x,t)=u(\xi,0),\qquad
p(x,t)=c_0^2(\rho(x,t)-\rho_0)+p_0.
\end{equation*}
The parameters are
\[
c_0=5,\qquad p_0=25,\qquad \rho_0=1,\qquad
\epsilon=0.02,\qquad u_0=0.05,\qquad t_f=0.2,
\]
and the viscous coefficient is set to zero for this inviscid simple-wave test. The reported accuracy quantity is the density error evaluated at the DG solution/quadrature points.

This test uses nonuniform one-dimensional meshes. The coarsest line is generated by
\begin{equation*}
x_i=x_{\min}+(x_{\max}-x_{\min})
\left[
s_i+\frac{0.16}{2\pi}\sin(2\pi s_i)-\frac{0.04}{4\pi}\sin(4\pi s_i)
\right],
\qquad s_i=\frac{i}{n_0},
\end{equation*}
with $x_{\min}=0$, $x_{\max}=2\pi$, and $n_0=4$. The finer $n_x=8,16,32,64$ meshes are obtained by midpoint insertion on every existing cell. This produces nested nonuniform periodic grids. To assess conservation, the mass variation is computed as $\Delta M=M_h(t_f)-M_h(0)$.

\begin{table}[htbp]
\centering
\caption{Density errors and mass variation for the one-dimensional simple wave on nonuniform meshes.}
\label{tab:wc_shear_1d_accuracy}
{\small
\setlength{\tabcolsep}{3.5pt}
\begin{tabular}{c c c c c c c c}
\toprule
Flux & $k$ & $n_x$ & $L^2$ error & Order & $L^\infty$ error & Order & $\Delta M$ \\
\midrule
LF & 1 & 4  & $1.571{\times}10^{-2}$ & --    & $1.065{\times}10^{-2}$ & --    & $0.000{\times}10^{0}$ \\
LF & 1 & 8  & $4.580{\times}10^{-3}$ & 1.778 & $3.691{\times}10^{-3}$ & 1.529 & $2.665{\times}10^{-15}$ \\
LF & 1 & 16 & $1.247{\times}10^{-3}$ & 1.877 & $9.807{\times}10^{-4}$ & 1.912 & $-1.776{\times}10^{-15}$ \\
LF & 1 & 32 & $3.182{\times}10^{-4}$ & 1.971 & $2.620{\times}10^{-4}$ & 1.904 & $1.776{\times}10^{-15}$ \\
LF & 1 & 64 & $8.006{\times}10^{-5}$ & 1.991 & $6.662{\times}10^{-5}$ & 1.975 & $-8.882{\times}10^{-16}$ \\
TR & 1 & 4  & $1.576{\times}10^{-2}$ & --    & $1.071{\times}10^{-2}$ & --    & $-8.882{\times}10^{-16}$ \\
TR & 1 & 8  & $4.606{\times}10^{-3}$ & 1.774 & $3.706{\times}10^{-3}$ & 1.530 & $1.776{\times}10^{-15}$ \\
TR & 1 & 16 & $1.255{\times}10^{-3}$ & 1.876 & $9.865{\times}10^{-4}$ & 1.910 & $-8.882{\times}10^{-16}$ \\
TR & 1 & 32 & $3.205{\times}10^{-4}$ & 1.970 & $2.632{\times}10^{-4}$ & 1.906 & $5.329{\times}10^{-15}$ \\
TR & 1 & 64 & $8.065{\times}10^{-5}$ & 1.991 & $6.709{\times}10^{-5}$ & 1.972 & $-2.665{\times}10^{-15}$ \\
\midrule
LF & 2 & 4  & $2.010{\times}10^{-3}$ & --    & $2.707{\times}10^{-3}$ & --    & $-8.882{\times}10^{-16}$ \\
LF & 2 & 8  & $2.329{\times}10^{-4}$ & 3.109 & $3.018{\times}10^{-4}$ & 3.165 & $-8.882{\times}10^{-16}$ \\
LF & 2 & 16 & $3.079{\times}10^{-5}$ & 2.919 & $4.172{\times}10^{-5}$ & 2.855 & $0.000{\times}10^{0}$ \\
LF & 2 & 32 & $3.841{\times}10^{-6}$ & 3.003 & $5.320{\times}10^{-6}$ & 2.971 & $-8.882{\times}10^{-16}$ \\
LF & 2 & 64 & $4.858{\times}10^{-7}$ & 2.983 & $6.625{\times}10^{-7}$ & 3.006 & $-5.329{\times}10^{-15}$ \\
TR & 2 & 4  & $2.017{\times}10^{-3}$ & --    & $2.729{\times}10^{-3}$ & --    & $-1.776{\times}10^{-15}$ \\
TR & 2 & 8  & $2.332{\times}10^{-4}$ & 3.112 & $3.016{\times}10^{-4}$ & 3.178 & $-1.776{\times}10^{-15}$ \\
TR & 2 & 16 & $3.077{\times}10^{-5}$ & 2.922 & $4.171{\times}10^{-5}$ & 2.854 & $1.776{\times}10^{-15}$ \\
TR & 2 & 32 & $3.837{\times}10^{-6}$ & 3.003 & $5.320{\times}10^{-6}$ & 2.971 & $0.000{\times}10^{0}$ \\
TR & 2 & 64 & $4.850{\times}10^{-7}$ & 2.984 & $6.627{\times}10^{-7}$ & 3.005 & $1.776{\times}10^{-15}$ \\
\midrule
LF & 3 & 4  & $1.937{\times}10^{-4}$ & --    & $2.538{\times}10^{-4}$ & --    & $-8.882{\times}10^{-16}$ \\
LF & 3 & 8  & $1.110{\times}10^{-5}$ & 4.126 & $1.323{\times}10^{-5}$ & 4.262 & $-1.776{\times}10^{-15}$ \\
LF & 3 & 16 & $7.006{\times}10^{-7}$ & 3.986 & $9.374{\times}10^{-7}$ & 3.819 & $-1.776{\times}10^{-15}$ \\
LF & 3 & 32 & $4.199{\times}10^{-8}$ & 4.060 & $5.461{\times}10^{-8}$ & 4.102 & $2.665{\times}10^{-15}$ \\
LF & 3 & 64 & $2.633{\times}10^{-9}$ & 3.995 & $3.389{\times}10^{-9}$ & 4.010 & $-2.665{\times}10^{-15}$ \\
TR & 3 & 4  & $1.937{\times}10^{-4}$ & --    & $2.523{\times}10^{-4}$ & --    & $0.000{\times}10^{0}$ \\
TR & 3 & 8  & $1.117{\times}10^{-5}$ & 4.117 & $1.341{\times}10^{-5}$ & 4.234 & $8.882{\times}10^{-16}$ \\
TR & 3 & 16 & $7.053{\times}10^{-7}$ & 3.985 & $9.420{\times}10^{-7}$ & 3.831 & $8.882{\times}10^{-16}$ \\
TR & 3 & 32 & $4.229{\times}10^{-8}$ & 4.060 & $5.506{\times}10^{-8}$ & 4.097 & $3.553{\times}10^{-15}$ \\
TR & 3 & 64 & $2.652{\times}10^{-9}$ & 3.995 & $3.403{\times}10^{-9}$ & 4.016 & $-5.329{\times}10^{-15}$ \\
\bottomrule
\end{tabular}
}
\end{table}

The results show the expected convergence behavior for the one-dimensional simple wave. Both LF and TR fluxes produce essentially the same density errors. The observed orders on the finest levels are approximately two, three, and four for $k=1,2,3$, respectively. The mass variation remains at roundoff level, demonstrating that the conservative LTS implementation preserves the total mass for this periodic one-dimensional test.

\subsection{Two-dimensional weakly compressible shear wave}
\label{subsec:wc_shear_2d}

The second exact-solution test is a two-dimensional diagonal shear wave on the periodic domain $[0,1]\times[0,1]$. Let
\[
\theta=2\pi(x+y)-2\pi(U_c+V_c)t+\phi .
\]
The parameters are
\[
c_0=1.5,\qquad p_0=2.25,\qquad \rho_0=\rho_\ast=1,\qquad
A=0.05,\qquad U_c=0.1,\qquad V_c=0.05,
\]
with $\mu=0.1$, $\phi=0.3$, and $t_f=0.05$. The initial condition is the exact state at $t=0$, and the full exact solution is
\begin{equation*}
\rho(x,y,t)=\rho_\ast,\qquad
p(x,y,t)=c_0^2(\rho_\ast-\rho_0)+p_0,
\end{equation*}
\begin{equation*}
u(x,y,t)=U_c-\frac{A}{\sqrt{2}}\exp(-8\pi^2\mu t/\rho_\ast)\sin\theta,\qquad
v(x,y,t)=V_c+\frac{A}{\sqrt{2}}\exp(-8\pi^2\mu t/\rho_\ast)\sin\theta .
\end{equation*}
The reported error is the velocity-vector error
\[
e_u=\sqrt{(u_h-u_{\mathrm{exact}})^2+(v_h-v_{\mathrm{exact}})^2}.
\]

All rows use the nonuniform midpoint quadrilateral mesh family with $n_x=n_y=4,8,16,32$. The coarse $4\times4$ grid is formed from tensor-product vertex lines. The $x$ vertex line uses the same smooth perturbation as the one-dimensional mesh, with coefficients $(0.16,-0.04)$, while the $y$ vertex line uses coefficients $(-0.12,0.03)$. The finer levels are generated by midpoint subdivision of every coarse cell edge in both coordinate directions, so the nonuniform grids remain nested and the periodic faces are connected by unit translations. The mass variation is computed as $\Delta M=M_h(t_f)-M_h(0)$.

\begin{table}[htbp]
\centering
\caption{Velocity errors and mass variation for the two-dimensional weakly compressible shear wave on nonuniform meshes.}
\label{tab:wc_shear_2d_accuracy}
{\small
\setlength{\tabcolsep}{3.5pt}
\begin{tabular}{c c c c c c c c}
\toprule
Flux & $k$ & Grid & $L^2$ error & Order & $L^\infty$ error & Order & $\Delta M$ \\
\midrule
LF & 1 & $4\times4$   & $1.766{\times}10^{-3}$ & --    & $3.487{\times}10^{-3}$ & --    & $-2.220{\times}10^{-16}$ \\
LF & 1 & $8\times8$   & $6.822{\times}10^{-4}$ & 1.372 & $1.482{\times}10^{-3}$ & 1.235 & $-2.665{\times}10^{-15}$ \\
LF & 1 & $16\times16$ & $1.738{\times}10^{-4}$ & 1.973 & $3.856{\times}10^{-4}$ & 1.942 & $4.108{\times}10^{-15}$ \\
LF & 1 & $32\times32$ & $4.428{\times}10^{-5}$ & 1.973 & $1.085{\times}10^{-4}$ & 1.829 & $1.876{\times}10^{-14}$ \\
TR & 1 & $4\times4$   & $2.001{\times}10^{-3}$ & --    & $3.926{\times}10^{-3}$ & --    & $-1.110{\times}10^{-16}$ \\
TR & 1 & $8\times8$   & $8.950{\times}10^{-4}$ & 1.161 & $1.851{\times}10^{-3}$ & 1.085 & $-1.998{\times}10^{-15}$ \\
TR & 1 & $16\times16$ & $2.480{\times}10^{-4}$ & 1.852 & $5.306{\times}10^{-4}$ & 1.803 & $4.108{\times}10^{-15}$ \\
TR & 1 & $32\times32$ & $5.626{\times}10^{-5}$ & 2.140 & $1.411{\times}10^{-4}$ & 1.911 & $1.776{\times}10^{-14}$ \\
\midrule
LF & 2 & $4\times4$   & $9.873{\times}10^{-4}$ & --    & $3.428{\times}10^{-3}$ & --    & $2.220{\times}10^{-16}$ \\
LF & 2 & $8\times8$   & $2.308{\times}10^{-4}$ & 2.097 & $6.865{\times}10^{-4}$ & 2.320 & $1.998{\times}10^{-15}$ \\
LF & 2 & $16\times16$ & $5.711{\times}10^{-5}$ & 2.015 & $1.281{\times}10^{-4}$ & 2.422 & $-4.885{\times}10^{-15}$ \\
LF & 2 & $32\times32$ & $1.489{\times}10^{-5}$ & 1.939 & $2.708{\times}10^{-5}$ & 2.242 & $-1.910{\times}10^{-14}$ \\
TR & 2 & $4\times4$   & $1.268{\times}10^{-3}$ & --    & $4.643{\times}10^{-3}$ & --    & $2.220{\times}10^{-16}$ \\
TR & 2 & $8\times8$   & $2.796{\times}10^{-4}$ & 2.181 & $8.694{\times}10^{-4}$ & 2.417 & $1.998{\times}10^{-15}$ \\
TR & 2 & $16\times16$ & $6.394{\times}10^{-5}$ & 2.128 & $1.453{\times}10^{-4}$ & 2.581 & $-6.772{\times}10^{-15}$ \\
TR & 2 & $32\times32$ & $1.584{\times}10^{-5}$ & 2.013 & $2.901{\times}10^{-5}$ & 2.325 & $-1.998{\times}10^{-14}$ \\
\midrule
LF & 3 & $4\times4$   & $5.284{\times}10^{-5}$ & --    & $2.527{\times}10^{-4}$ & --    & $9.992{\times}10^{-16}$ \\
LF & 3 & $8\times8$   & $3.578{\times}10^{-6}$ & 3.884 & $1.583{\times}10^{-5}$ & 3.996 & $-2.887{\times}10^{-15}$ \\
LF & 3 & $16\times16$ & $2.409{\times}10^{-7}$ & 3.893 & $1.029{\times}10^{-6}$ & 3.943 & $-3.553{\times}10^{-15}$ \\
LF & 3 & $32\times32$ & $1.606{\times}10^{-8}$ & 3.907 & $6.410{\times}10^{-8}$ & 4.005 & $1.112{\times}10^{-13}$ \\
TR & 3 & $4\times4$   & $5.849{\times}10^{-5}$ & --    & $2.904{\times}10^{-4}$ & --    & $1.443{\times}10^{-15}$ \\
TR & 3 & $8\times8$   & $3.718{\times}10^{-6}$ & 3.976 & $1.633{\times}10^{-5}$ & 4.153 & $-2.220{\times}10^{-15}$ \\
TR & 3 & $16\times16$ & $2.447{\times}10^{-7}$ & 3.926 & $1.040{\times}10^{-6}$ & 3.972 & $1.554{\times}10^{-15}$ \\
TR & 3 & $32\times32$ & $1.617{\times}10^{-8}$ & 3.920 & $6.496{\times}10^{-8}$ & 4.002 & $1.090{\times}10^{-13}$ \\
\bottomrule
\end{tabular}
}
\end{table}

The results show the expected convergence behavior and roundoff-level mass conservation for the proposed split LTS treatment. The LF and TR fluxes remain close in accuracy. The $k=1$ and $k=3$ cases recover approximately second- and fourth-order behavior, respectively. For $k=2$, the observed order is close to two rather than the optimal $k+1=3$, which is attributed to the IIPG viscous discretization: for even polynomial degrees it typically yields only $k$-th order accuracy, consistent with the known parity-dependent behavior of incomplete or nonsymmetric interior-penalty DG discretizations \cite{ArnoldBrezziCockburnMarini2002,Riviere2008}.

To highlight the role of the split treatment in the nodal DG discretization, an additional comparison is performed on the same two-dimensional nonuniform shear-wave problem with $k=1$ on the $16\times16$ grid. The comparison directly uses the Gassner-type unsplit surface update represented by \eqref{eq:gassner_lts_unsplit_surface} instead of the conservative split update used in the present method. Table~\ref{tab:wc_shear_2d_nonconserved} compares the result with the conservative formulation.

\begin{table}[htbp]
\centering
\caption{Conservation comparison for the two-dimensional nonuniform shear wave with $k=1$ on the $16\times16$ grid.}
\label{tab:wc_shear_2d_nonconserved}
{\small
\begin{tabular}{l c c c c}
\toprule
Update & Flux & $L^2$ error & $L^\infty$ error & $\Delta M$ \\
\midrule
Conservative split update & LF & $1.738{\times}10^{-4}$ & $3.856{\times}10^{-4}$ & $4.108{\times}10^{-15}$ \\
Gassner-type unsplit update & LF & $1.398{\times}10^{-4}$ & $4.194{\times}10^{-4}$ & $8.345{\times}10^{-8}$ \\
Conservative split update & TR & $2.480{\times}10^{-4}$ & $5.306{\times}10^{-4}$ & $4.108{\times}10^{-15}$ \\
Gassner-type unsplit update & TR & $2.098{\times}10^{-4}$ & $5.537{\times}10^{-4}$ & $8.160{\times}10^{-8}$ \\
\bottomrule
\end{tabular}
}
\end{table}

The comparison shows that directly applying the unsplit Gassner-type update to the nodal DGSEM discretization destroys the roundoff-level mass conservation: the mass drift increases from about $10^{-15}$ to about $10^{-8}$. This behavior is consistent with the discussion in Section~\ref{sec:lts_conservation}: in the nodal GLL setting, the purely interior face contribution must be integrated together with the volume term in order to retain the discrete SBP cancellation and hence the conservation property.

\subsection{Flow past a circular cylinder at $Re=100$}
\label{subsec:cylinder_re100}

The flow past a circular cylinder at $Re=100$ is used to compare the pressure-field behavior of different inviscid numerical fluxes in the weakly compressible formulation. The free-stream velocity is
\begin{equation*}
U_\infty=1.0~\mathrm{m\,s^{-1}},
\end{equation*}
the artificial sound speed is set to
\begin{equation*}
c_0=10~\mathrm{m\,s^{-1}},
\end{equation*}
and the cylinder diameter is
\begin{equation*}
D=0.012~\mathrm{m}.
\end{equation*}
The reference kinematic viscosity is therefore $\mu/\rho_0=U_\infty D/Re=1.2\times10^{-4}~\mathrm{m^2\,s^{-1}}$. The polynomial degree is $k=3$, and the computation is advanced to the final time $t_f=3.0~\mathrm{s}$.
\begin{figure}[htbp]
    \centering
    \begin{minipage}[t]{0.48\textwidth}
        \centering
        \includegraphics[width=\linewidth]{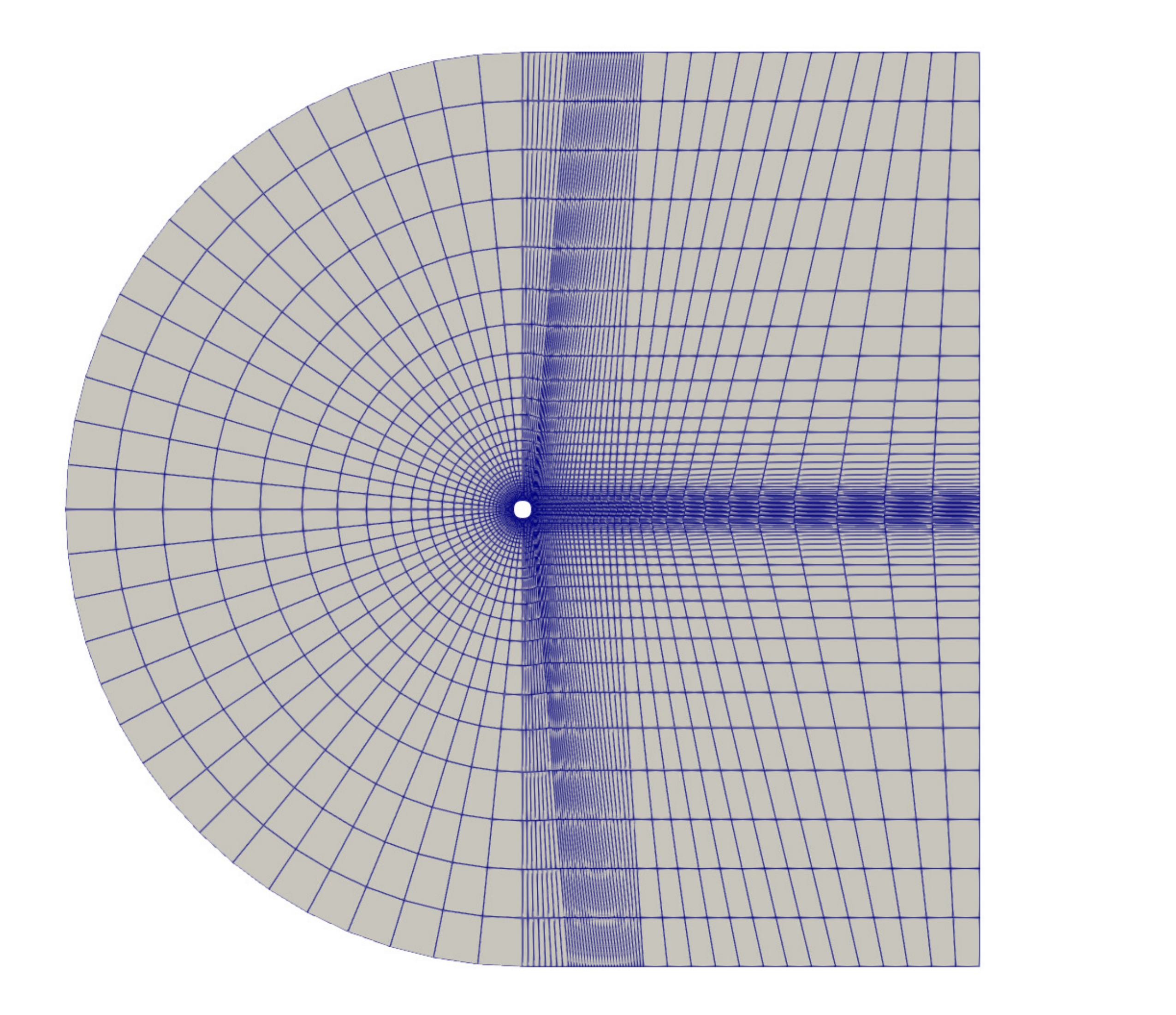}\\[-1mm]
        {\footnotesize (a) Computational mesh}
    \end{minipage}
    \hfill
    \begin{minipage}[t]{0.48\textwidth}
        \centering
        \includegraphics[width=\linewidth]{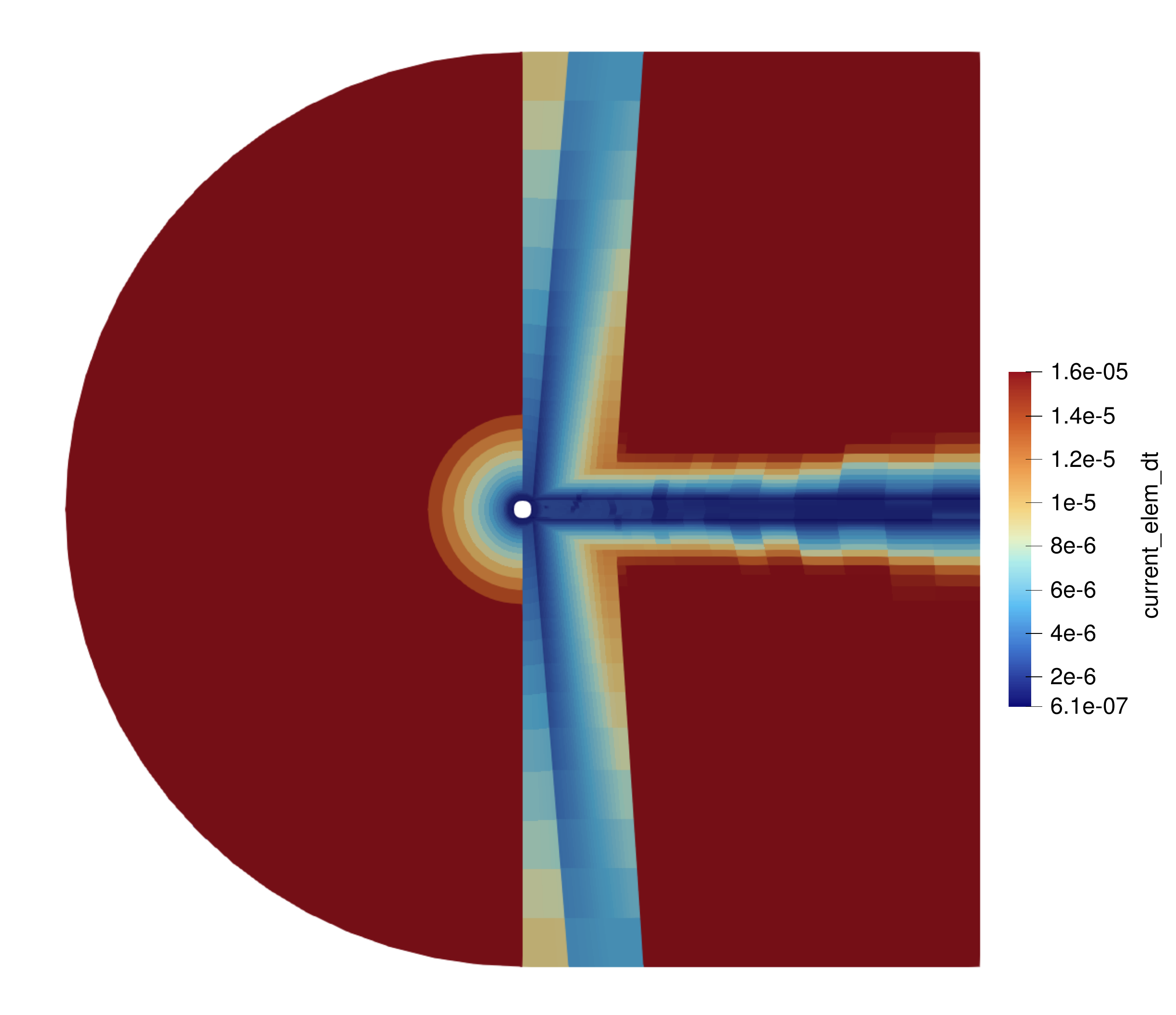}\\[-1mm]
        {\footnotesize (b) Local time-step size at $t=3.0~\mathrm{s}$}
    \end{minipage}
    \caption{Mesh and instantaneous local time-step distribution for the $Re=100$ cylinder-flow computation.}
    \label{fig:re100_mesh_dt}
\end{figure}

Figure~\ref{fig:re100_mesh_dt} shows the computational mesh and the local time-step distribution at $t=3.0~\mathrm{s}$. The smallest local time steps are concentrated near the cylinder and in the refined wake region, where the mesh size and local gradients impose the strongest stability constraint. Away from the body, the larger elements allow larger local steps, which is consistent with the intended local time-stepping behavior.

\begin{figure}[htbp]
    \centering
    \begin{minipage}[t]{0.47\textwidth}
        \centering
        \includegraphics[width=\linewidth]{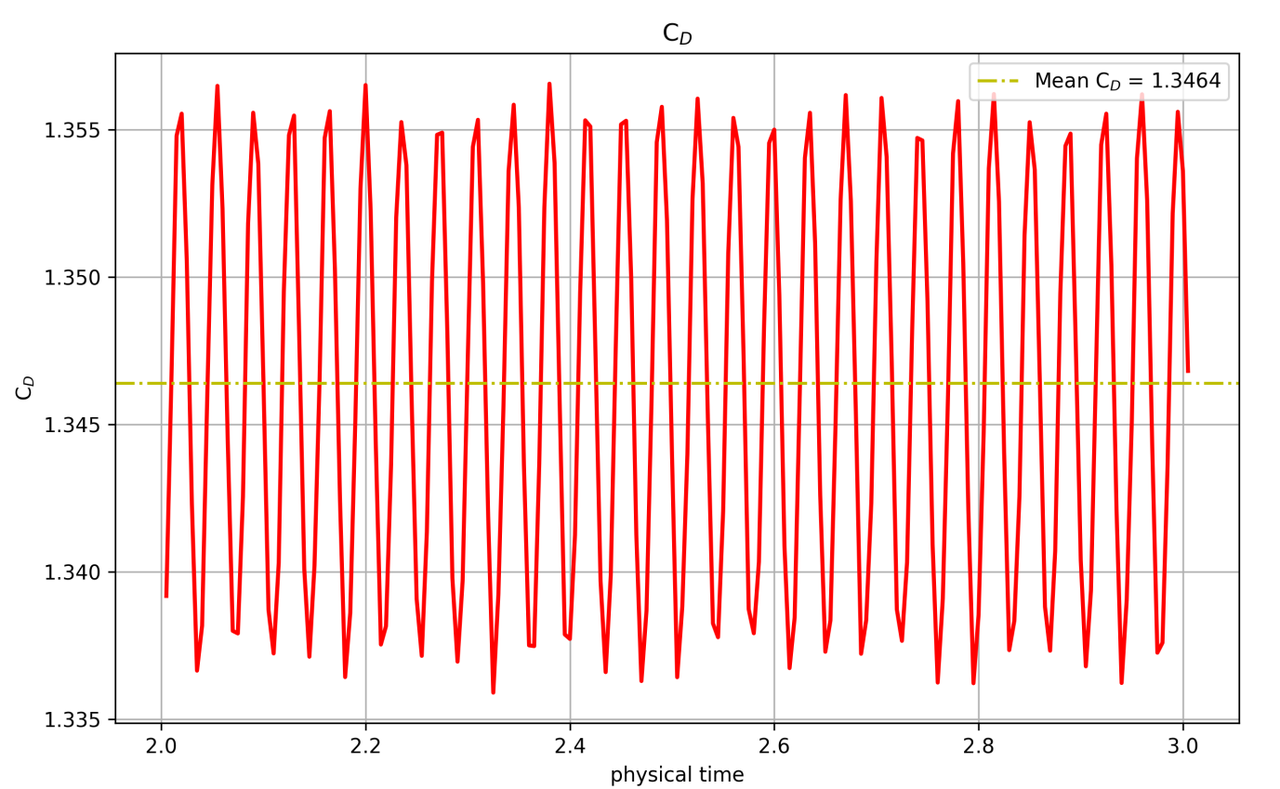}\\[-1mm]
        {\footnotesize (a) Drag coefficient, LF flux}
    \end{minipage}
    \hfill
    \begin{minipage}[t]{0.47\textwidth}
        \centering
        \includegraphics[width=\linewidth]{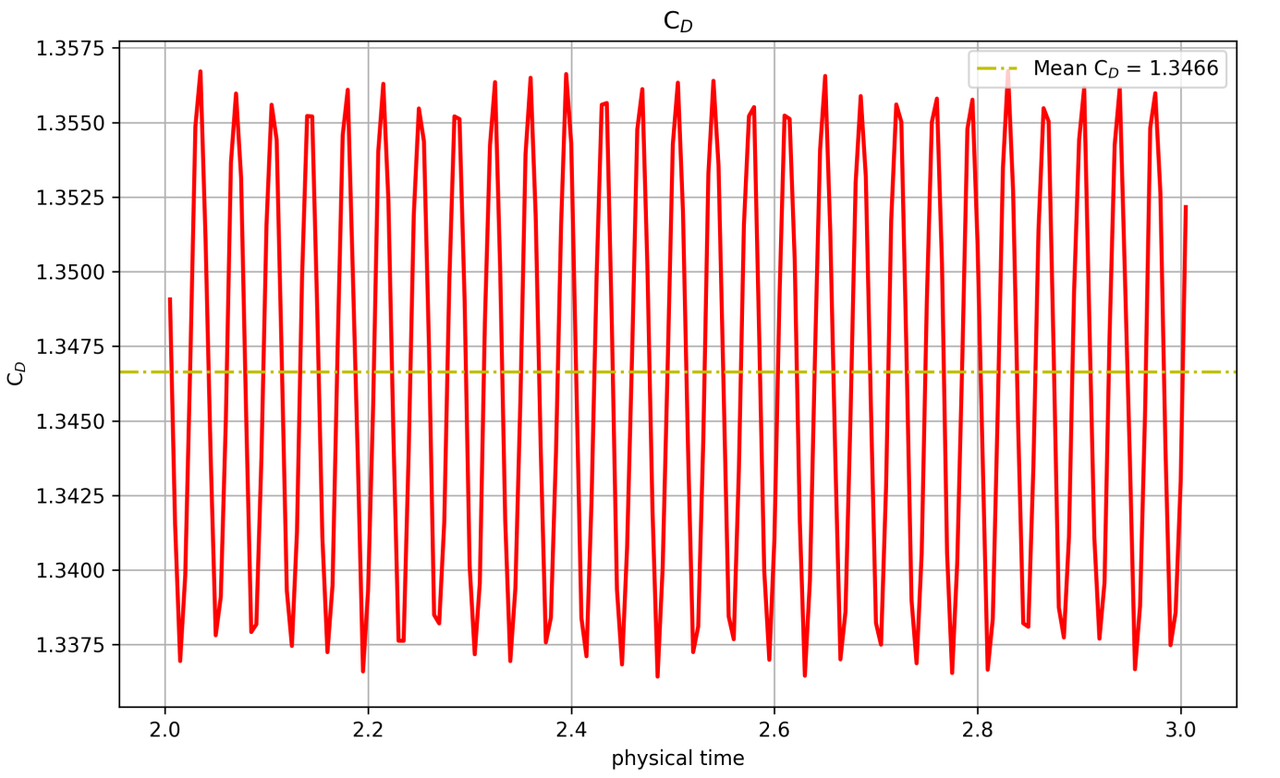}\\[-1mm]
        {\footnotesize (b) Drag coefficient, TR flux}
    \end{minipage}

    \vspace{0.15cm}

    \begin{minipage}[t]{0.47\textwidth}
        \centering
        \includegraphics[width=\linewidth]{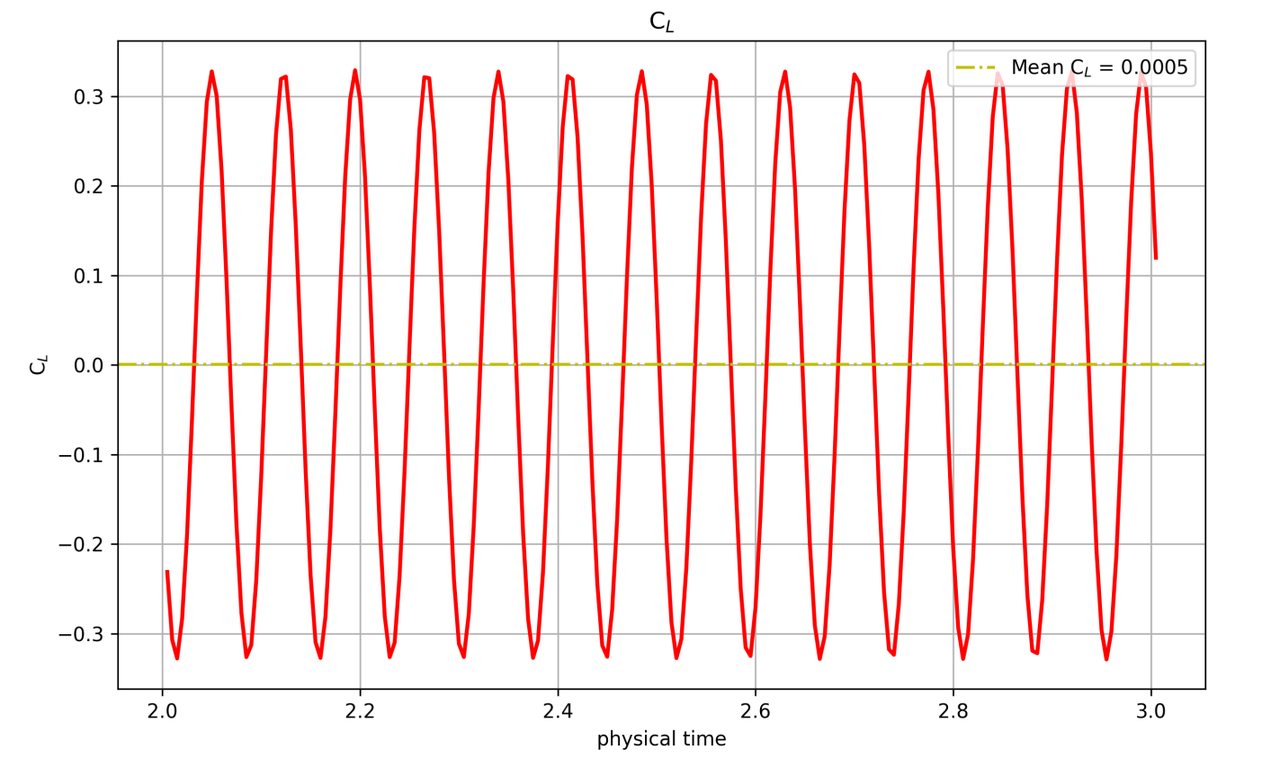}\\[-1mm]
        {\footnotesize (c) Lift coefficient, LF flux}
    \end{minipage}
    \hfill
    \begin{minipage}[t]{0.47\textwidth}
        \centering
        \includegraphics[width=\linewidth]{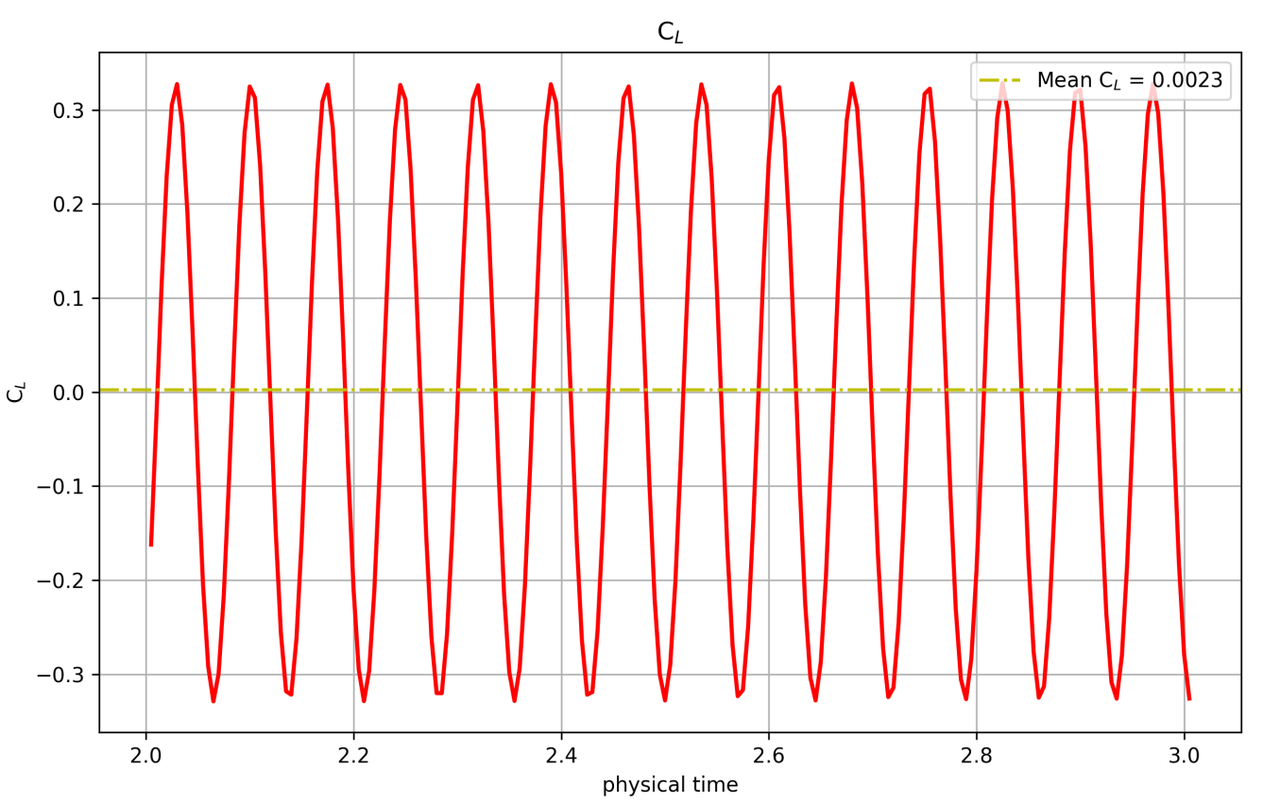}\\[-1mm]
        {\footnotesize (d) Lift coefficient, TR flux}
    \end{minipage}
    \caption{Time histories of the drag and lift coefficients for the $Re=100$ cylinder flow.}
    \label{fig:re100_force_history}
\end{figure}

Figure~\ref{fig:re100_force_history} gives the aerodynamic-coefficient histories after the periodic shedding state has developed. The drag coefficient oscillates around nearly identical mean levels, with $\overline{C_D}=1.3464$ for the LF flux and $\overline{C_D}=1.3466$ for the TR flux. The lift coefficient is approximately periodic with a small mean value for both fluxes, and the corresponding rms lift coefficients are $C_L'=0.227$ for LF and $C_L'=0.233$ for TR. The shedding frequency extracted from the lift histories gives $St=0.166$ for both LF and TR.

\begin{table}[htbp]
\centering
\caption{Comparison of force statistics and Strouhal number for flow past a circular cylinder at $Re=100$.}
\label{tab:re100_cd_st}
\begin{tabular}{l c c c}
\toprule
Source & $\overline{C_D}$ & $C_L'$ & $St$ \\
\midrule
Present, LF flux & 1.3464 & 0.227 & 0.166 \\
Present, TR flux & 1.3466 & 0.233 & 0.166 \\
Aarnes et al.~\cite{AarnesHaugenAndersson2018} & 1.346 & 0.234--0.235 & 0.166 \\
Li et al.~\cite{LiZhangShockChen2009} & 1.336 & -- & 0.164 \\
Posdziech and Grundmann~\cite{PosdziechGrundmann2007} & 1.350 & 0.234 & 0.167 \\
Pan~\cite{Pan2006} & 1.32 & 0.23 & 0.16 \\
Qu et al.~\cite{QuNorbergDavidsonPengWang2013} & 1.326 & 0.219 & 0.166 \\
\bottomrule
\end{tabular}
\end{table}

Table~\ref{tab:re100_cd_st} compares the present force statistics with the overset-grid result of Aarnes et al.~\cite{AarnesHaugenAndersson2018} and the reference data summarized therein. The present mean drag coefficients are essentially identical to the large-domain overset-grid value of Aarnes et al. The TR rms lift coefficient agrees closely with the Aarnes and Posdziech--Grundmann values, while the LF value is slightly lower but still within the range of published results. The computed Strouhal number agrees exactly with the Aarnes and Qu et al. values while remaining within the spread of the other published data. This confirms that the force histories are consistent with established $Re=100$ cylinder-flow benchmarks.

\begin{figure}[htbp]
    \centering
    \begin{minipage}[t]{0.48\textwidth}
        \centering
        \includegraphics[width=\linewidth]{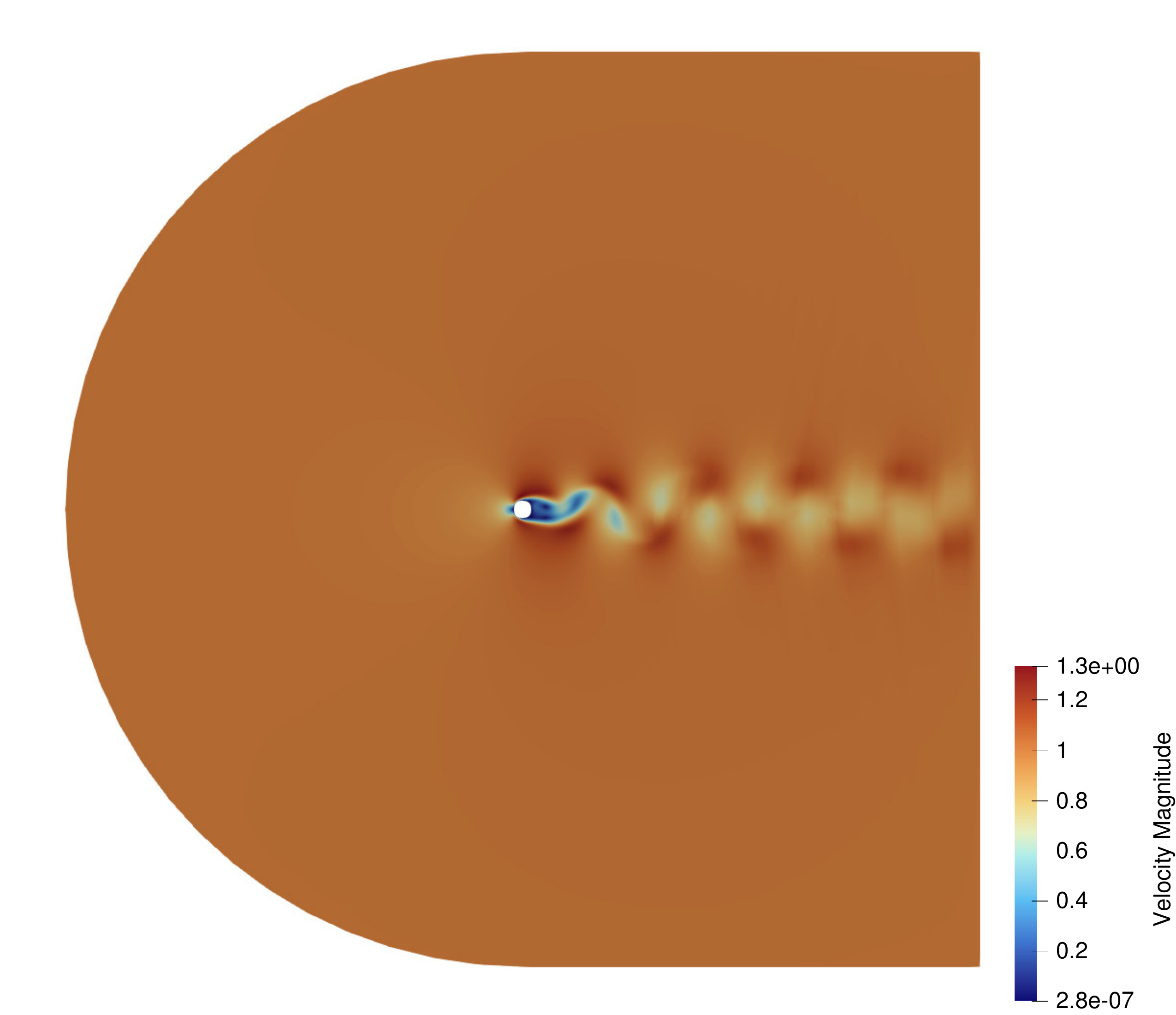}\\[-1mm]
        {\footnotesize (a) LF flux}
    \end{minipage}
    \hfill
    \begin{minipage}[t]{0.48\textwidth}
        \centering
        \includegraphics[width=\linewidth]{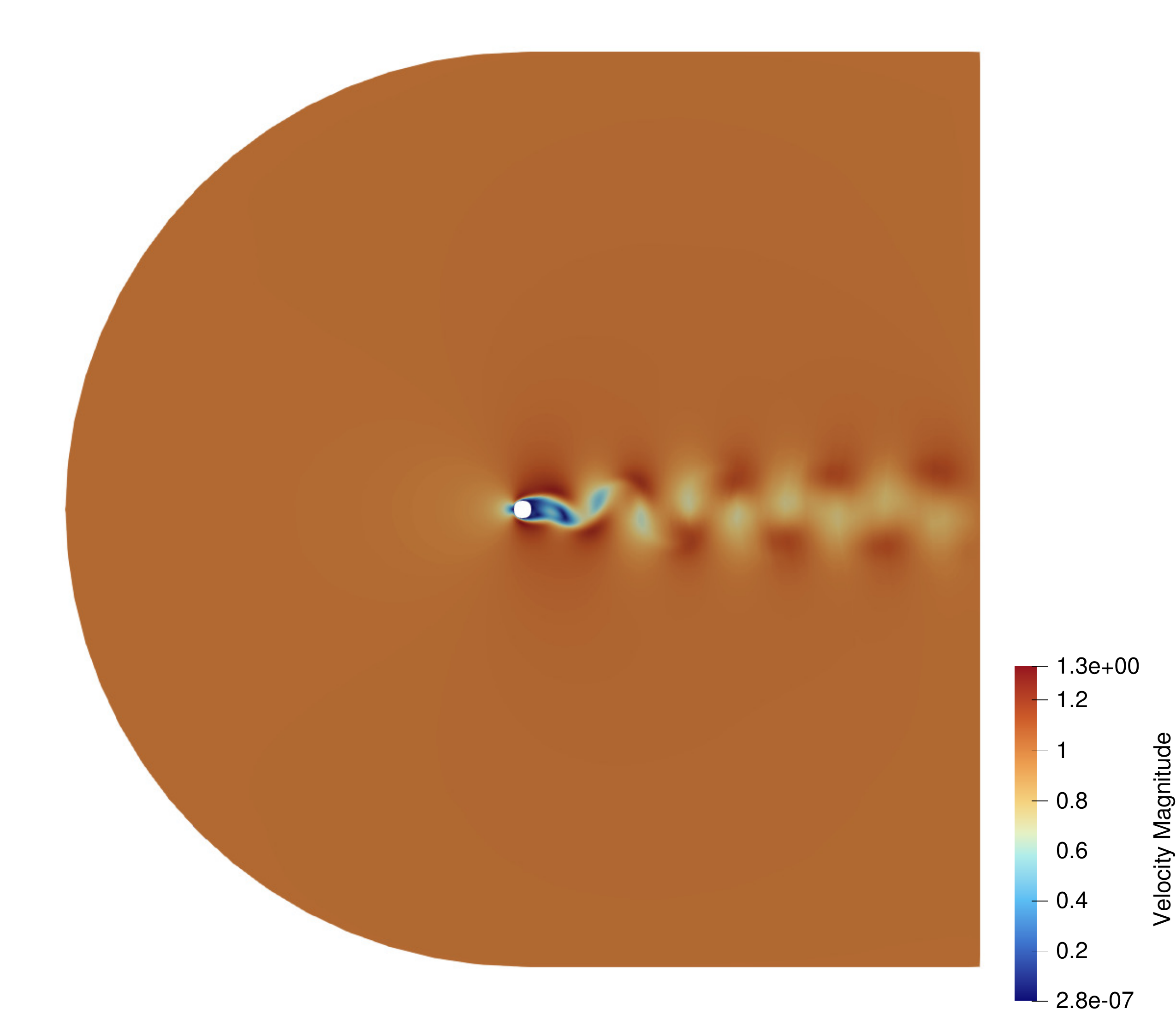}\\[-1mm]
        {\footnotesize (b) TR flux}
    \end{minipage}
    \caption{Velocity-magnitude fields, in $\mathrm{m\,s^{-1}}$, at $t=3.0~\mathrm{s}$ for the $Re=100$ cylinder flow.}
    \label{fig:re100_velocity}
\end{figure}

The velocity-magnitude fields in Fig.~\ref{fig:re100_velocity} show that the LF and TR fluxes produce comparable near-cylinder acceleration and wake development at the final time.

\begin{figure}[htbp]
    \centering
    \begin{minipage}[t]{0.48\textwidth}
        \centering
        \includegraphics[width=\linewidth]{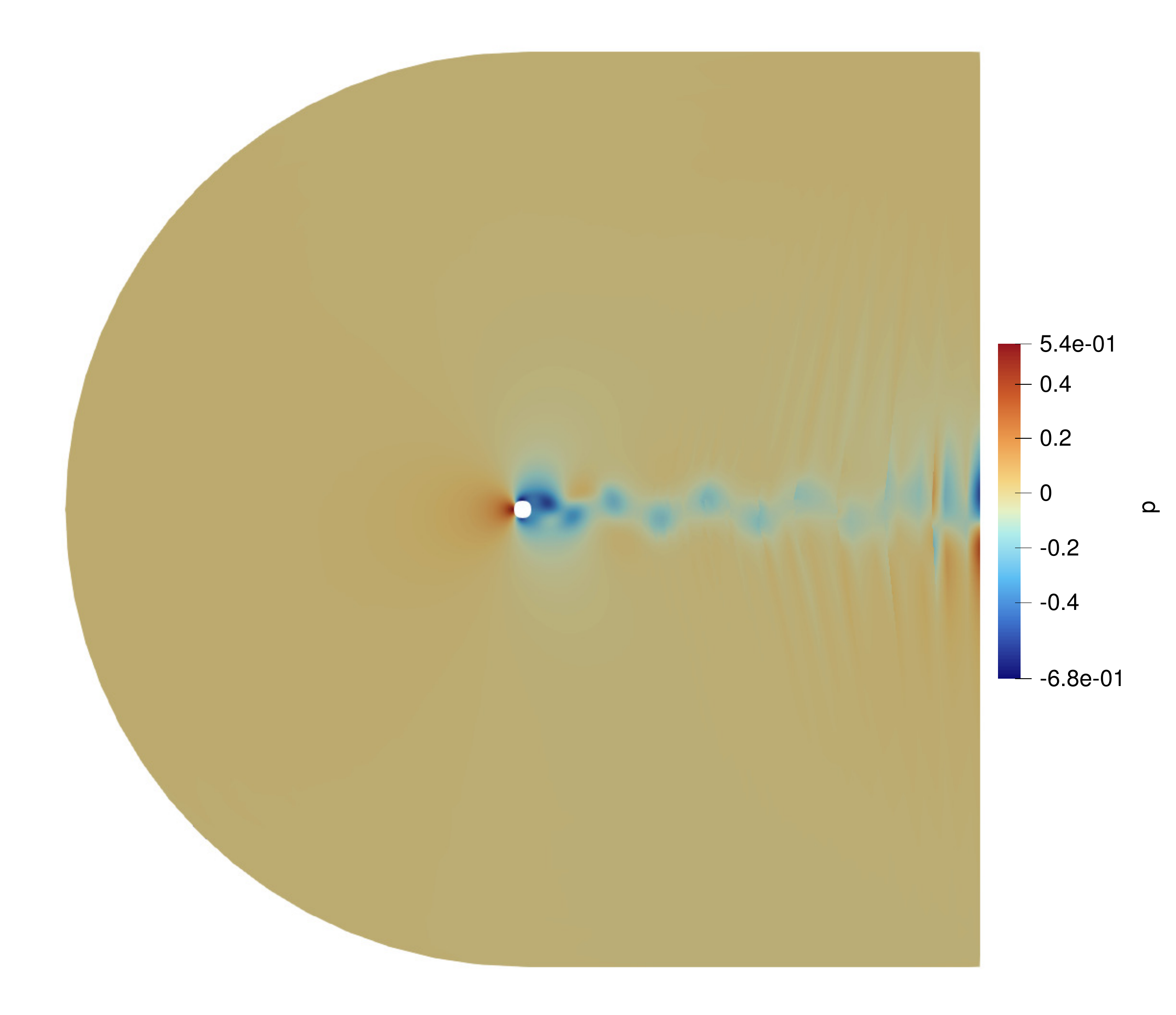}\\[-1mm]
        {\footnotesize (a) LF flux}
    \end{minipage}
    \hfill
    \begin{minipage}[t]{0.48\textwidth}
        \centering
        \includegraphics[width=\linewidth]{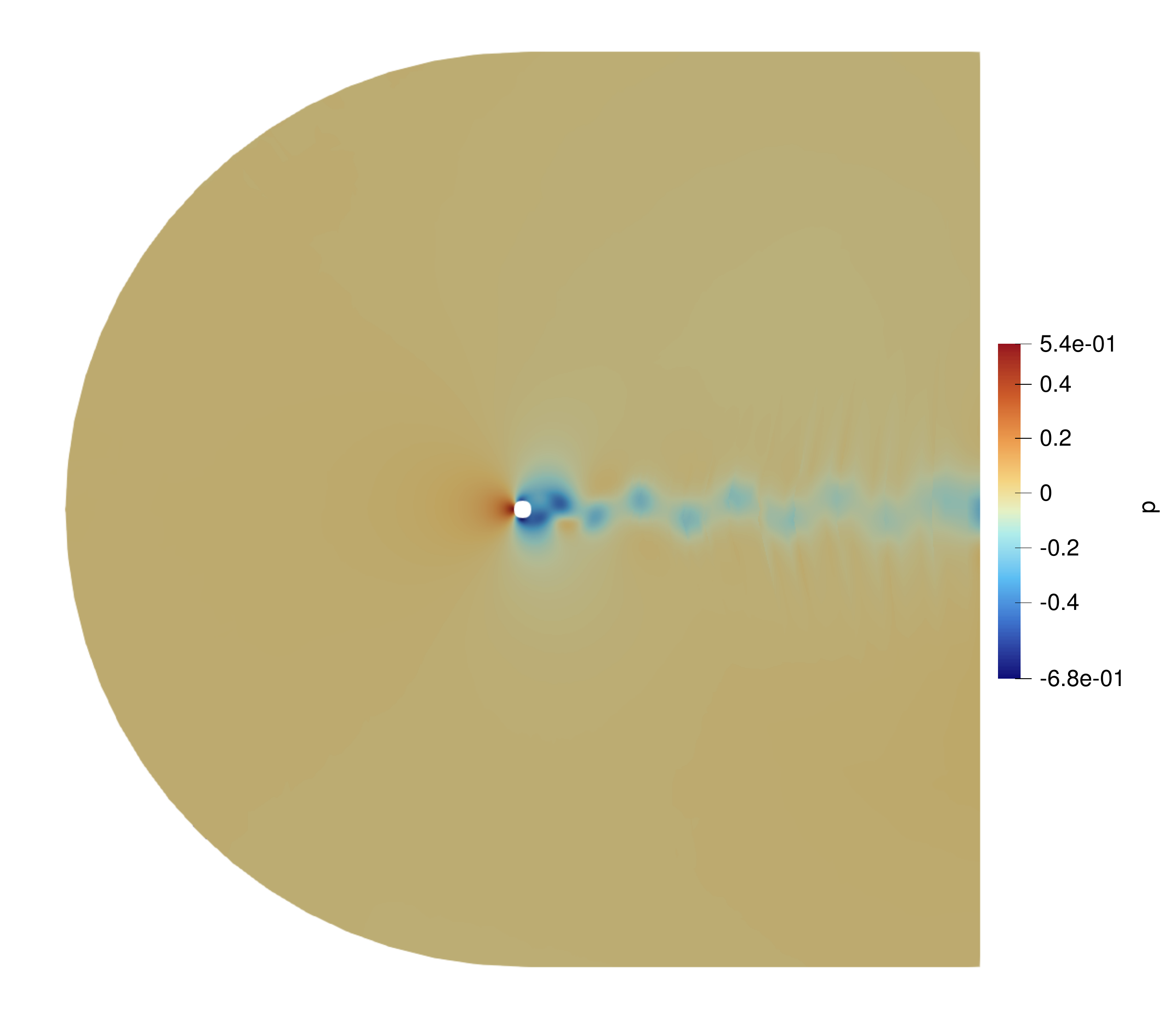}\\[-1mm]
        {\footnotesize (b) TR flux}
    \end{minipage}
    \caption{Pressure fields at $t=3.0~\mathrm{s}$ for the $Re=100$ cylinder flow.}
    \label{fig:re100_pressure}
\end{figure}

The pressure contours in Fig.~\ref{fig:re100_pressure} show a visible difference between the two fluxes. Compared with LF, the TR result exhibits weaker far-field pressure fluctuations, especially close to the downstream outflow boundary. In the present cylinder-flow test, this suggests a potential benefit of the two-rarefaction approximate flux for reducing spurious boundary reflections in this open-boundary setting.

\begin{figure}[htbp]
    \centering
    \begin{minipage}[t]{0.47\textwidth}
        \centering
        \includegraphics[width=\linewidth]{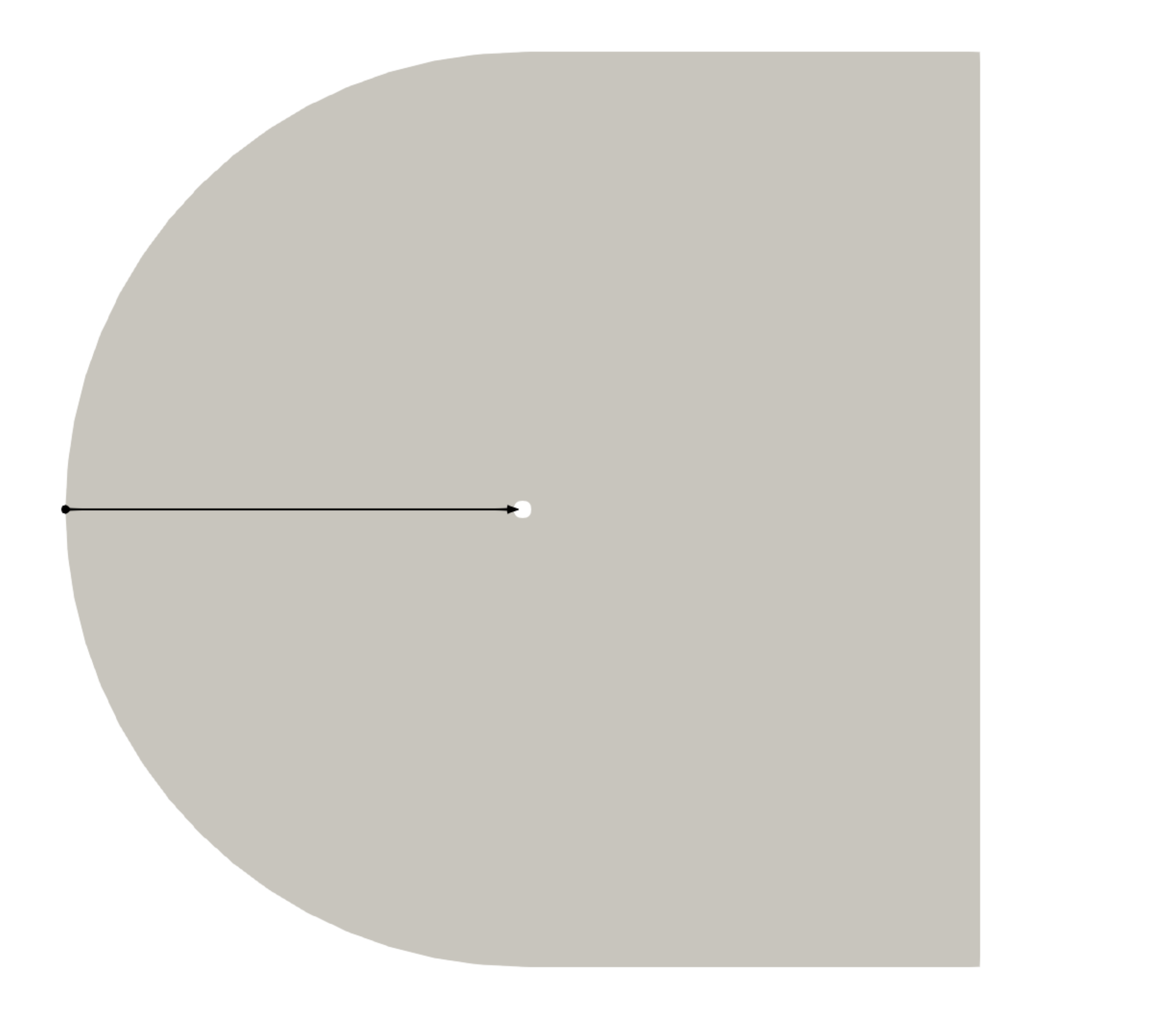}\\[-1mm]
        {\footnotesize (a) Sampling line}
    \end{minipage}
    \hfill
    \begin{minipage}[t]{0.50\textwidth}
        \centering
        \includegraphics[width=\linewidth]{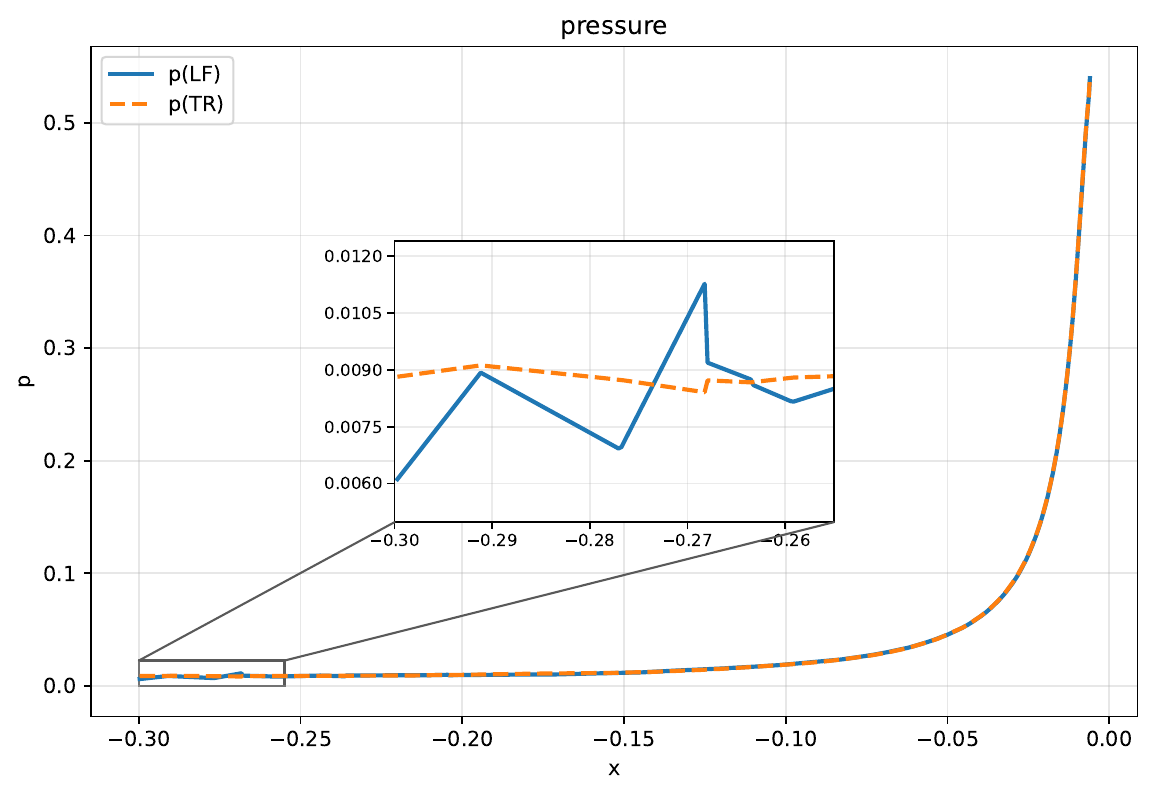}\\[-1mm]
        {\footnotesize (b) Pressure along the sampling line}
    \end{minipage}
    \caption{Pressure sampling line through the cylinder centerline and the corresponding pressure comparison at $t=3.0~\mathrm{s}$.}
    \label{fig:re100_pressure_line}
\end{figure}

To further compare the pressure response near the far-field boundary, Fig.~\ref{fig:re100_pressure_line} shows the pressure sampled along a horizontal line passing through the cylinder centerline. The pressure curve shows weaker pressure fluctuations for the TR flux near the upstream far-field boundary. Together with the pressure-contour comparison, this suggests a potential benefit of the TR flux in the present cylinder-flow test: it retains the near-body solution quality while producing a weaker far-field pressure response than the LF flux.

\subsection{Lid-driven cavity flow}
\label{subsec:cavity}

The lid-driven cavity flow is considered next to examine the capability of the proposed method for incompressible-like viscous recirculating flows. The computational domain is a square cavity with side length $L=1~\mathrm{m}$,
\begin{equation*}
(x,y)\in[0,1]~\mathrm{m}\times[0,1]~\mathrm{m}.
\end{equation*}
No-slip boundary conditions are imposed on all walls, and the top lid moves with a constant horizontal velocity
\begin{equation*}
u=U_{\mathrm{lid}}=1.0~\mathrm{m\,s^{-1}}, \qquad v=0~\mathrm{m\,s^{-1}}.
\end{equation*}
The Reynolds number is defined as
\begin{equation*}
Re=\frac{U_{\mathrm{lid}}L}{\mu/\rho_0},
\end{equation*}

The computations are performed at $Re=1000$ with polynomial degree $k=3$. A $100\times100$ nonuniform grid is used, as shown in Fig.~\ref{fig:cavity_grid}. To examine the sensitivity to the artificial sound speed in the weakly compressible model, three values are considered: $c_0=5~\mathrm{m\,s^{-1}}$, $10~\mathrm{m\,s^{-1}}$, and $20~\mathrm{m\,s^{-1}}$.

\begin{figure}[htbp]
    \centering
    \includegraphics[width=0.64\textwidth]{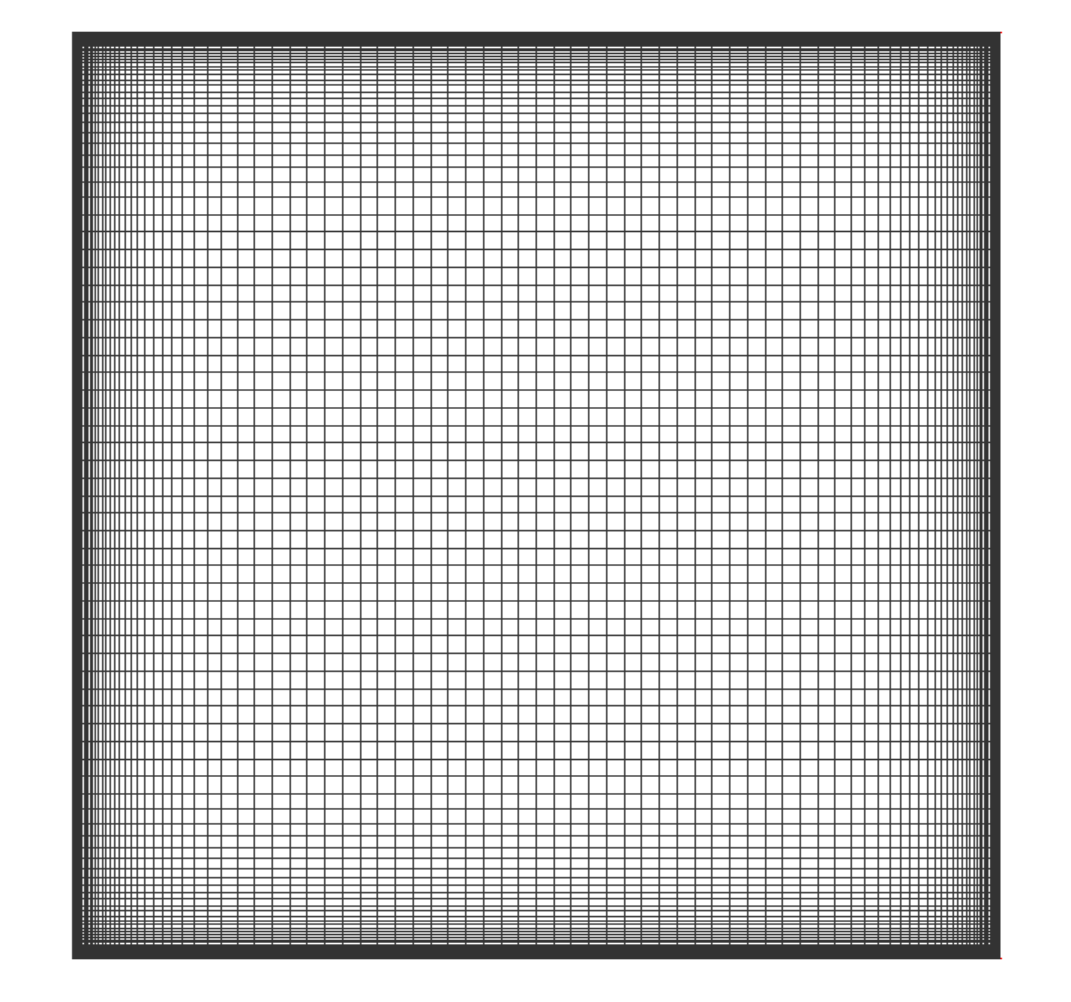}
    \caption{Computational grid for the lid-driven cavity calculation; coordinate lengths are in meters.}
    \label{fig:cavity_grid}
\end{figure}

The main post-processing results for the representative case with $c_0=10~\mathrm{m\,s^{-1}}$ are shown first in Fig.~\ref{fig:cavity_streamlines}. The streamline plot colored by velocity magnitude clearly identifies the primary recirculating vortex and the two lower-corner secondary vortices. The pressure contour shows the high-pressure region near the upper-right corner and the corresponding low-pressure region near the upper-left corner, while the vorticity contour highlights the thin shear layers generated by the moving lid and no-slip side walls.

\begin{figure}[htbp]
    \centering
    \subfigure[Streamlines and velocity magnitude ($\mathrm{m\,s^{-1}}$)]{
        \includegraphics[width=0.31\textwidth]{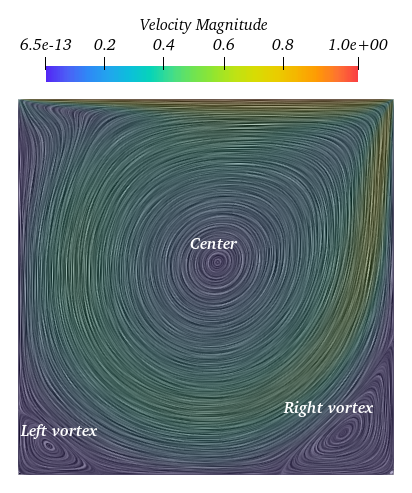}
    }
    \hfill
    \subfigure[Dimensionless pressure contour]{
        \includegraphics[width=0.31\textwidth]{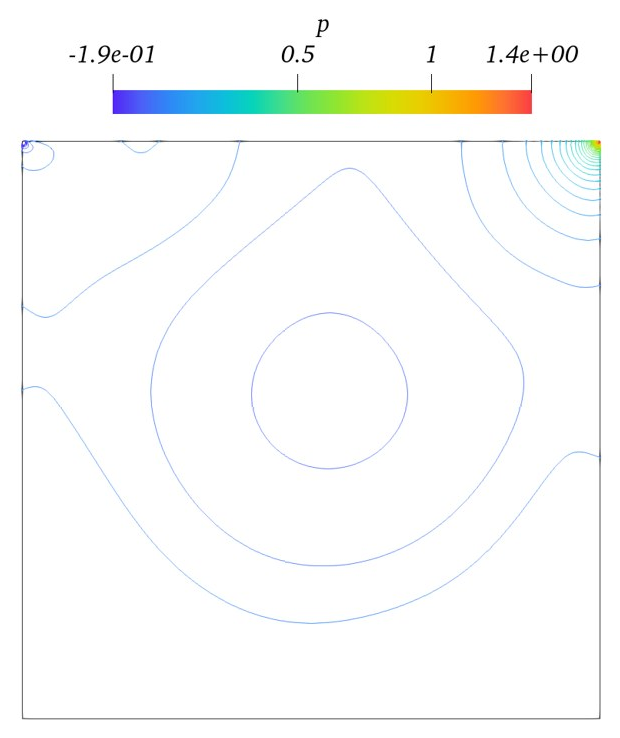}
    }
    \hfill
    \subfigure[Vorticity contour ($\mathrm{s^{-1}}$)]{
        \includegraphics[width=0.31\textwidth]{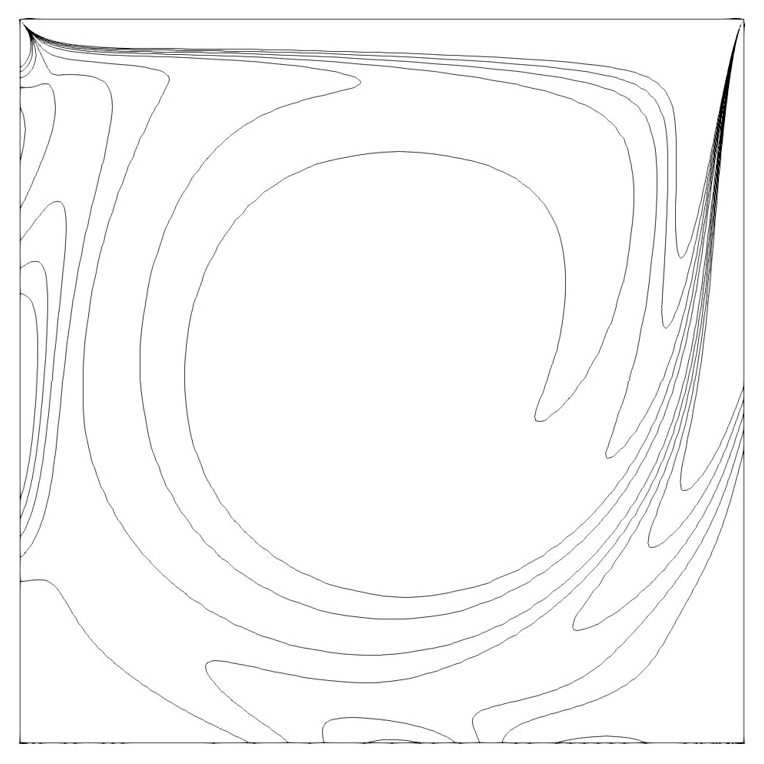}
    }
    \caption{Post-processed flow fields for the lid-driven cavity at $Re=1000$ with $c_0=10~\mathrm{m\,s^{-1}}$.}
    \label{fig:cavity_streamlines}
\end{figure}

For quantitative validation, Fig.~\ref{fig:cavity_centerline} compares the horizontal velocity along the vertical centerline and the vertical velocity along the horizontal centerline with benchmark data from Ghia et al.~\cite{GhiaGhiaShin1982} and Botella and Peyret~\cite{BotellaPeyret1998}. The plotted velocities are normalized by $U_{\mathrm{lid}}$; since $U_{\mathrm{lid}}=1.0~\mathrm{m\,s^{-1}}$, the numerical values are identical to the corresponding SI velocities in $\mathrm{m\,s^{-1}}$. The three values of $c_0$ give very similar velocity profiles on the global scale, and the agreement with the benchmark profiles is close near the extrema and near-wall turning regions.

\begin{figure}[htbp]
    \centering
    \subfigure[$u$ on the vertical centerline]{
        \includegraphics[width=0.47\textwidth]{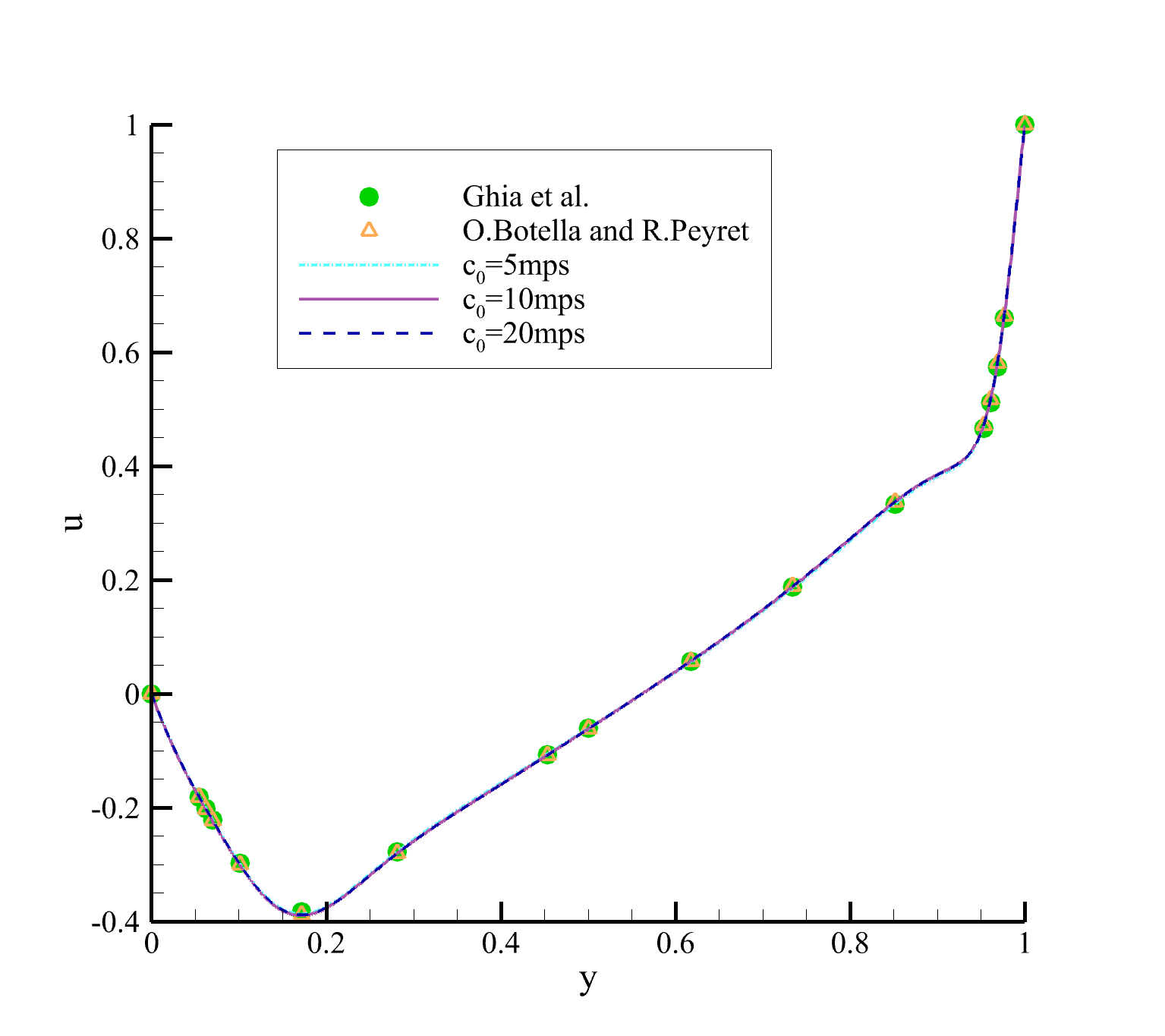}
    }
    \hfill
    \subfigure[$v$ on the horizontal centerline]{
        \includegraphics[width=0.47\textwidth]{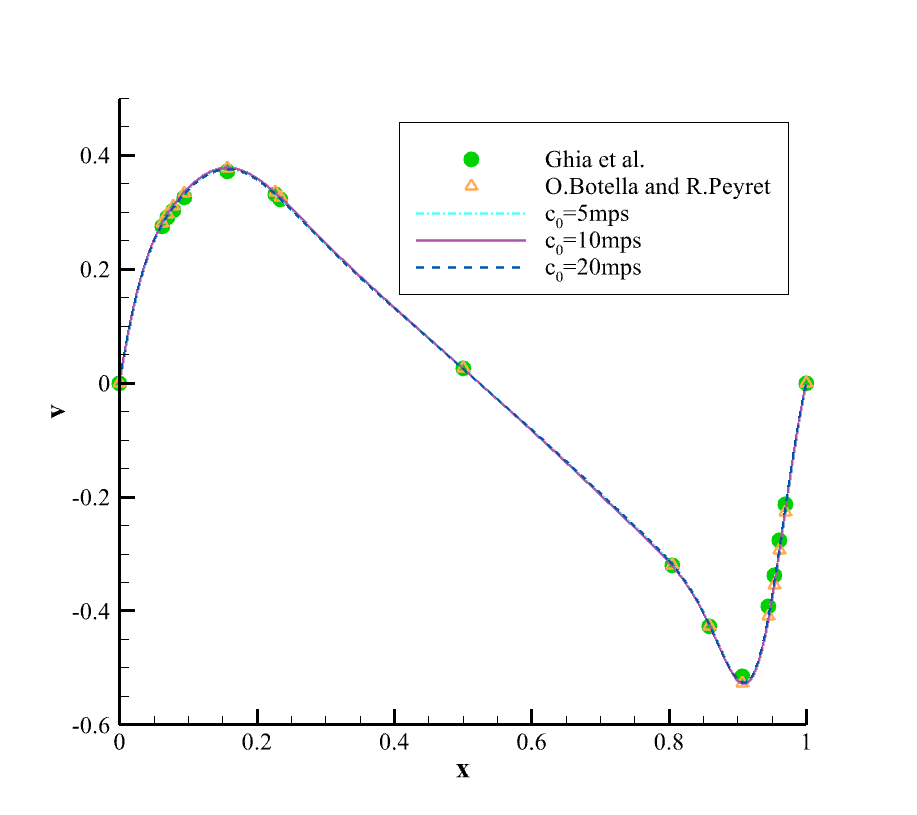}
    }
    \caption{Centerline velocity profiles normalized by $U_{\mathrm{lid}}$ for the lid-driven cavity at $Re=1000$.}
    \label{fig:cavity_centerline}
\end{figure}

The pressure profiles provide a more sensitive view of the influence of $c_0$. Figure~\ref{fig:cavity_pressure_profiles} plots the dimensionless pressure along both centerlines and includes local enlarged views. Although the full-scale curves remain close, the enlarged plots show visible differences among $c_0=5~\mathrm{m\,s^{-1}}$, $10~\mathrm{m\,s^{-1}}$, and $20~\mathrm{m\,s^{-1}}$. The $c_0=5~\mathrm{m\,s^{-1}}$ curve shows the largest deviation in the local pressure extrema, whereas the $c_0=10~\mathrm{m\,s^{-1}}$ and $20~\mathrm{m\,s^{-1}}$ curves are closer to each other and to the Botella--Peyret reference data. This behavior is consistent with the role of the artificial sound speed in the weakly compressible equation of state: the velocity solution is already well converged for all three values, while the pressure field is more sensitive to the compressibility parameter.

\begin{figure}[htbp]
    \centering
    \subfigure[Vertical centerline]{
        \includegraphics[width=0.42\textwidth]{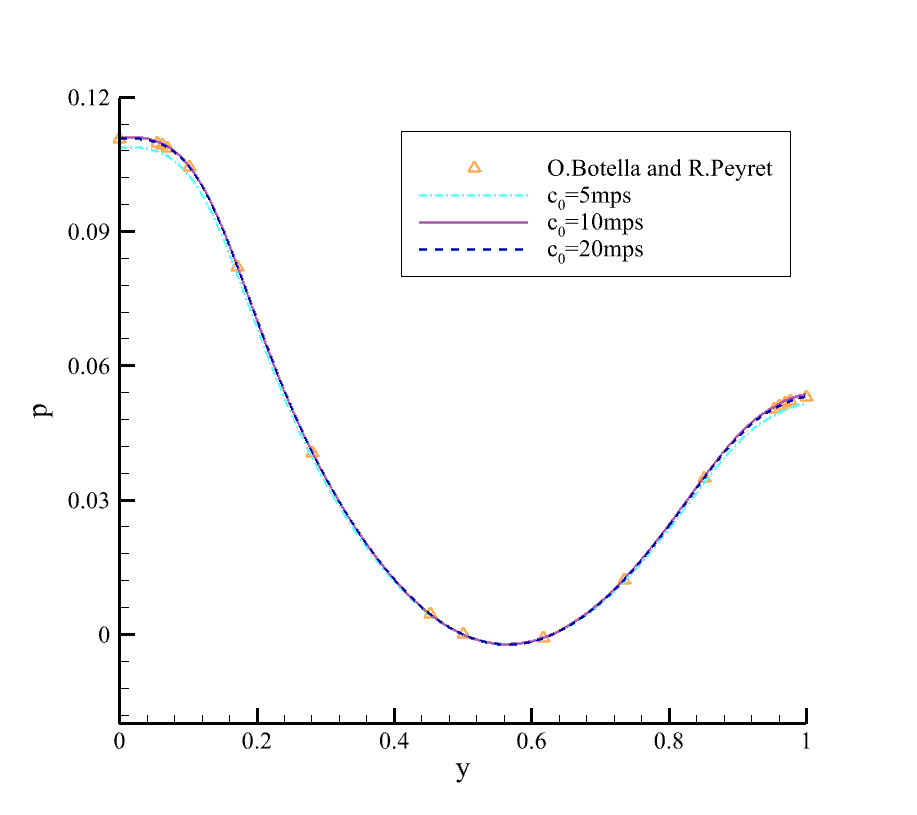}
    }
    \hfill
    \subfigure[Horizontal centerline]{
        \includegraphics[width=0.42\textwidth]{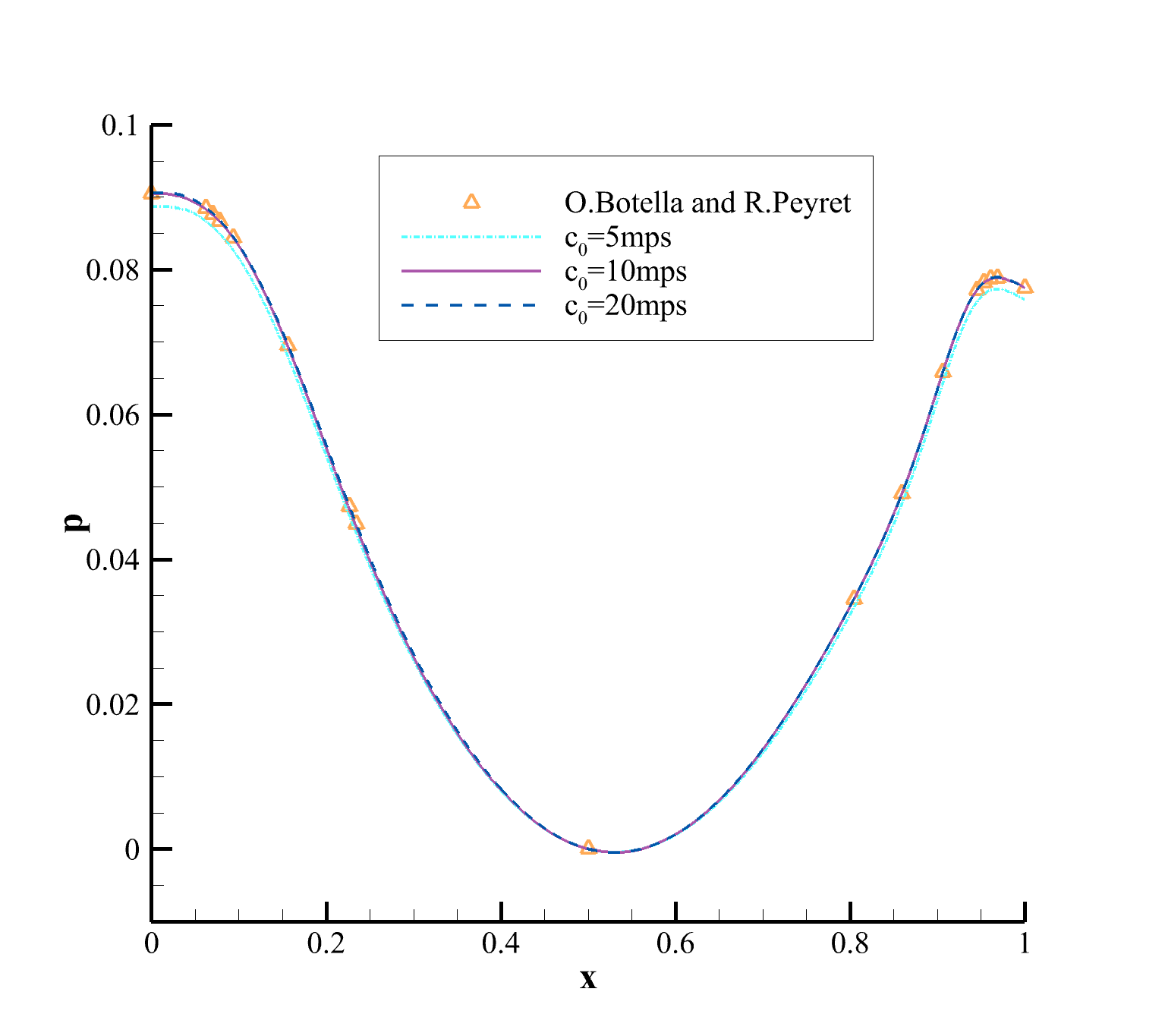}
    }
    
    \vspace{0.2cm}

    \subfigure[Vertical centerline, zoom I]{
        \includegraphics[width=0.42\textwidth]{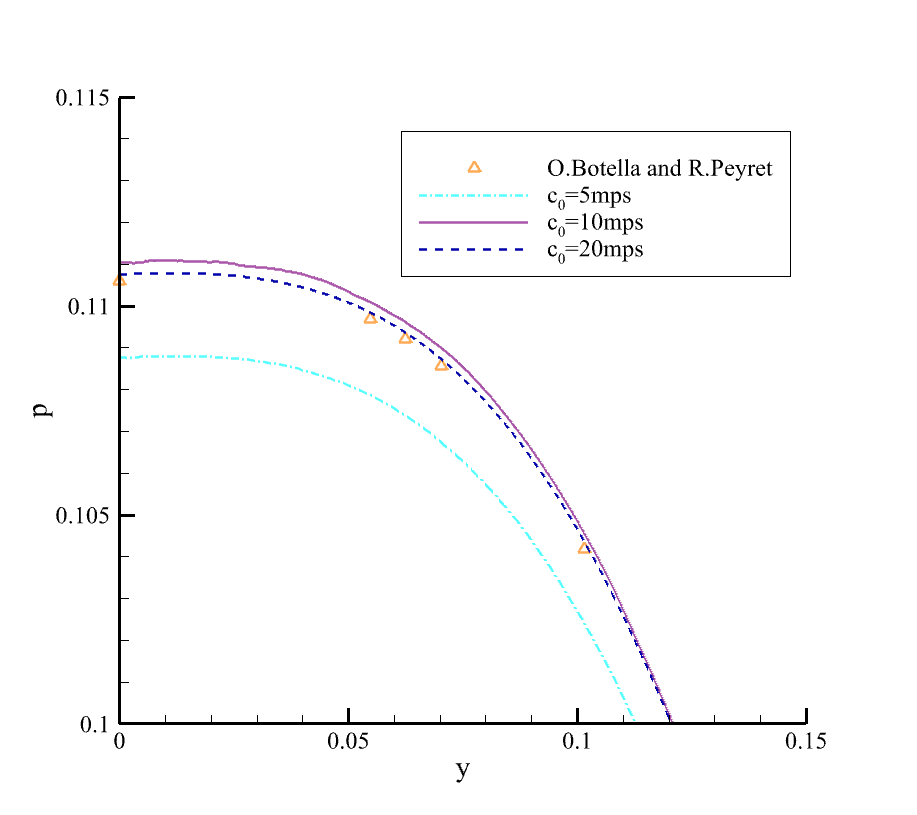}
    }
    \hfill
    \subfigure[Horizontal centerline, zoom I]{
        \includegraphics[width=0.42\textwidth]{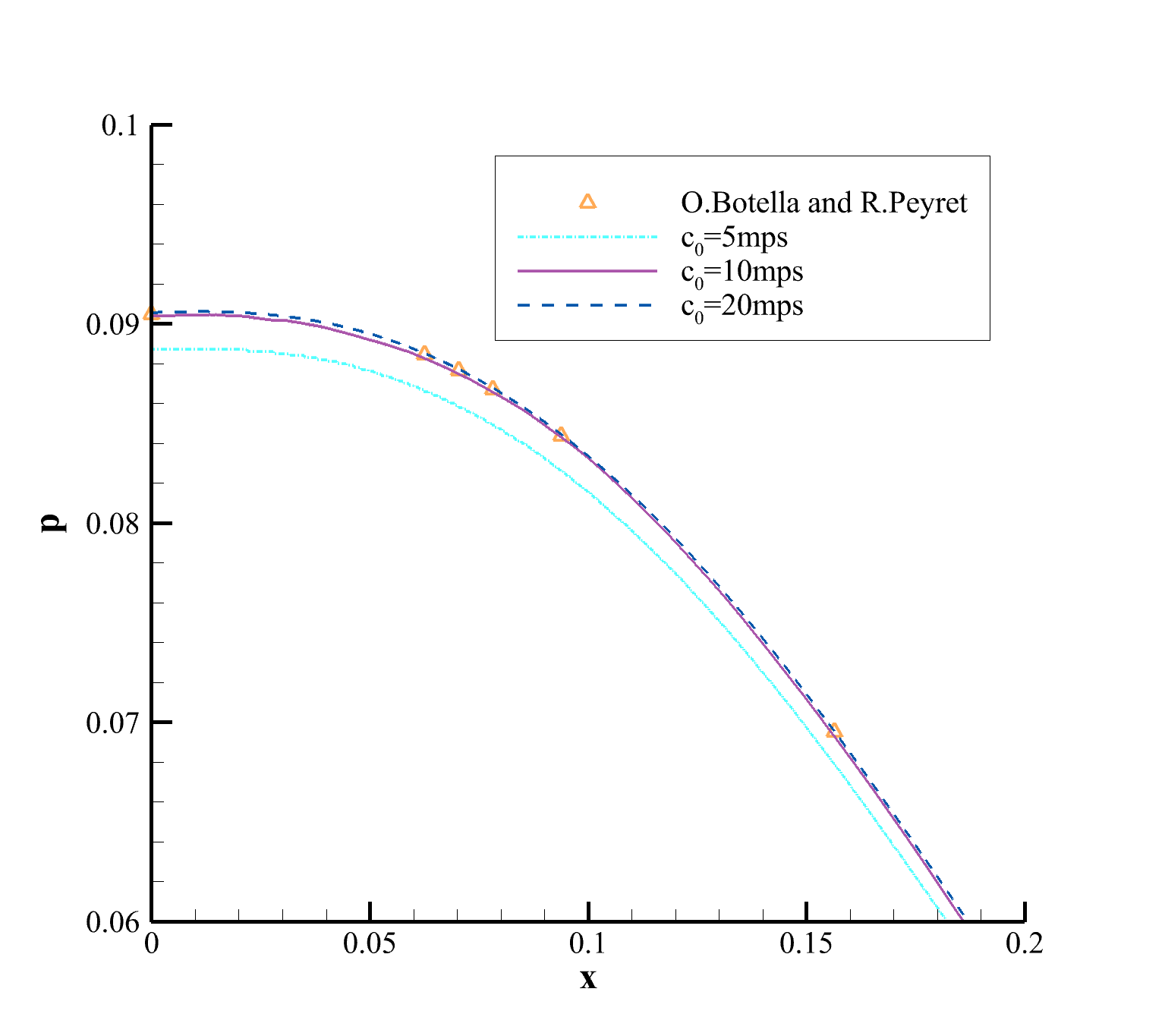}
    }

    \vspace{0.2cm}

    \subfigure[Vertical centerline, zoom II]{
        \includegraphics[width=0.42\textwidth]{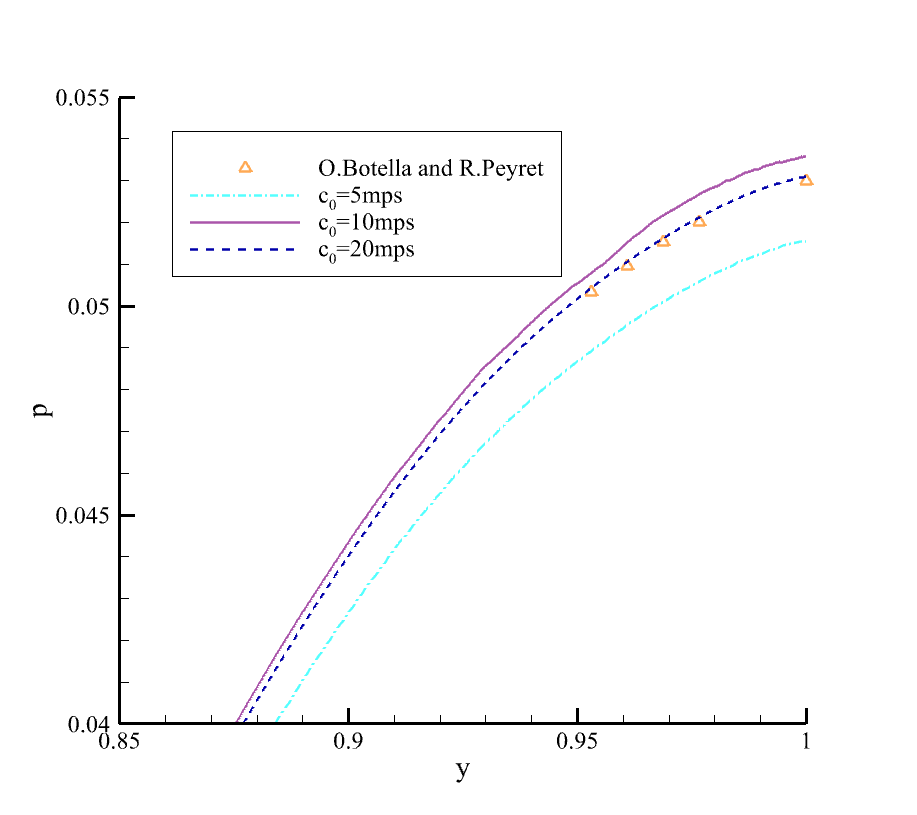}
    }
    \hfill
    \subfigure[Horizontal centerline, zoom II]{
        \includegraphics[width=0.42\textwidth]{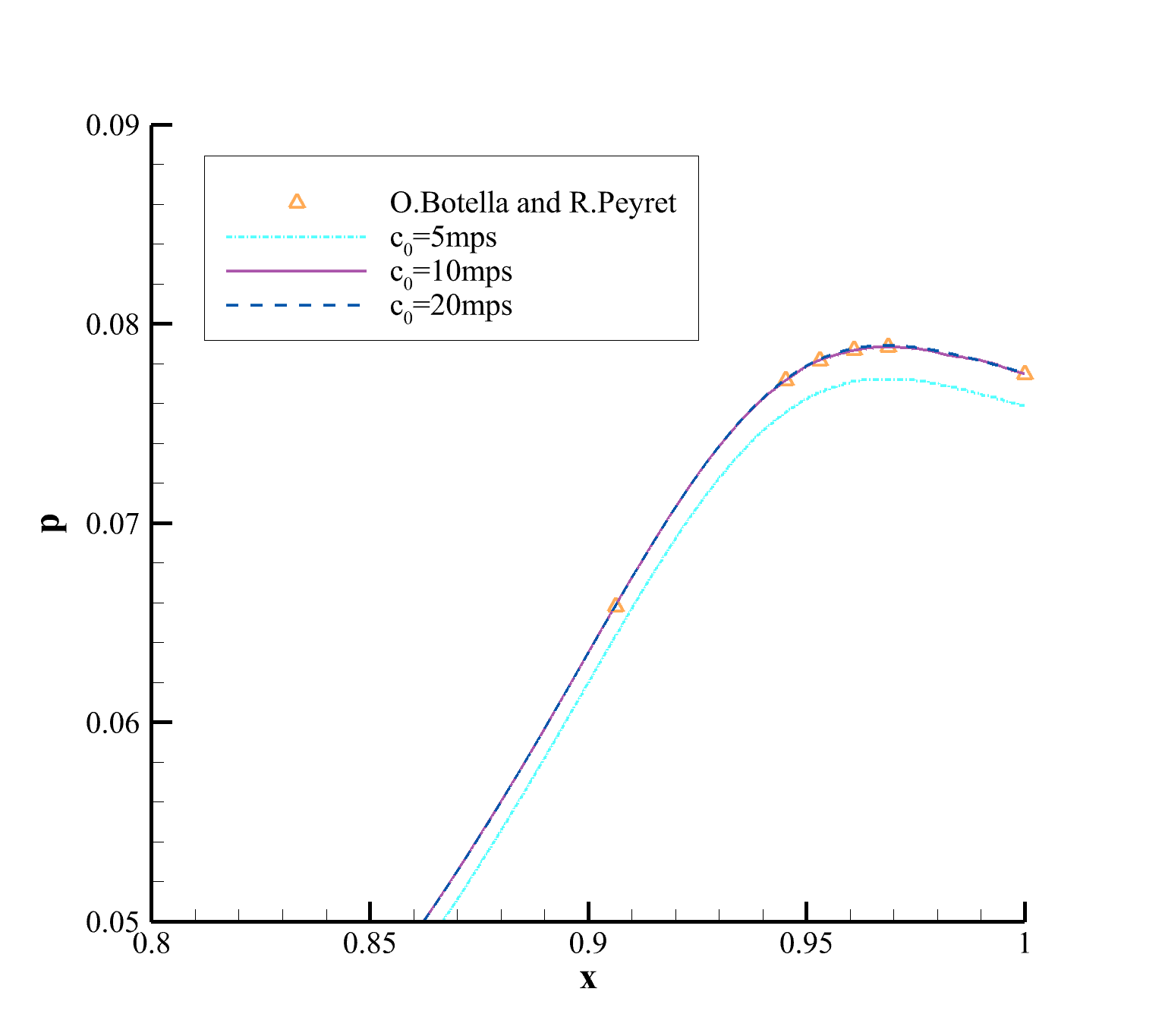}
    }
    \caption{Dimensionless centerline pressure profiles and local enlarged views for the lid-driven cavity at $Re=1000$.}
    \label{fig:cavity_pressure_profiles}
\end{figure}

Table~\ref{tab:cavity_extrema} lists the velocity extrema extracted from the centerline profiles and compares them with representative benchmark values from the literature. The velocity extrema are normalized by $U_{\mathrm{lid}}$. Compared with the spectral benchmark of Botella and Peyret, the $c_0=10~\mathrm{m\,s^{-1}}$ and $20~\mathrm{m\,s^{-1}}$ results reproduce the velocity extrema closely, while the $c_0=5~\mathrm{m\,s^{-1}}$ result remains within the spread of the reference data.

\begin{table}[htbp]
\centering
\caption{Velocity extrema normalized by $U_{\mathrm{lid}}$ on the cavity centerlines at $Re=1000$.}
\label{tab:cavity_extrema}
\resizebox{0.9\textwidth}{!}{%
\begin{tabular}{l c c c c}
\toprule
Reference & Grid & $u_{\min}/U_{\mathrm{lid}}$ & $v_{\max}/U_{\mathrm{lid}}$ & $v_{\min}/U_{\mathrm{lid}}$ \\
\midrule
Present, $c_0=5~\mathrm{m\,s^{-1}}$  & $100\times100$ & -0.386911 & 0.375245 & -0.525427 \\
Present, $c_0=10~\mathrm{m\,s^{-1}}$ & $100\times100$ & -0.389368 & 0.377974 & -0.527792 \\
Present, $c_0=20~\mathrm{m\,s^{-1}}$ & $100\times100$ & -0.388598 & 0.377014 & -0.527080 \\
Botella and Peyret~\cite{BotellaPeyret1998} & $N=48$ & -0.388527 & 0.376899 & -0.527017 \\
Botella and Peyret~\cite{BotellaPeyret1998} & $N=64$ & -0.388570 & 0.376944 & -0.527076 \\
Botella and Peyret~\cite{BotellaPeyret1998} & $N=96$ & -0.388570 & 0.376945 & -0.527077 \\
Botella and Peyret~\cite{BotellaPeyret1998} & $N=128$ & -0.388570 & 0.376945 & -0.527077 \\
Botella and Peyret~\cite{BotellaPeyret1998} & $N=160$ & -0.388570 & 0.376945 & -0.527077 \\
Deng et al., CPI~\cite{DengPiquetQueuteyVisonneau1994} & $64\times64$ & -0.37436 & 0.36364 & -0.51015 \\
Deng et al., CPI~\cite{DengPiquetQueuteyVisonneau1994} & $96\times96$ & -0.38233 & 0.37109 & -0.51947 \\
Deng et al., CPI~\cite{DengPiquetQueuteyVisonneau1994} & $128\times128$ & -0.38511 & 0.37369 & -0.52280 \\
Deng et al., CPI~\cite{DengPiquetQueuteyVisonneau1994} & extrap. & -0.38867 & 0.37702 & -0.52724 \\
Deng et al., staggered~\cite{DengPiquetQueuteyVisonneau1994} & $64\times64$ & -0.375726 & 0.34556 & -0.48858 \\
Deng et al., staggered~\cite{DengPiquetQueuteyVisonneau1994} & $96\times96$ & -0.37441 & 0.36271 & -0.50982 \\
Deng et al., staggered~\cite{DengPiquetQueuteyVisonneau1994} & $128\times128$ & -0.38050 & 0.36884 & -0.51727 \\
Deng et al., staggered~\cite{DengPiquetQueuteyVisonneau1994} & extrap. & -0.38855 & 0.37705 & -0.52690 \\
Ghia et al.~\cite{GhiaGhiaShin1982} & $129\times129$ & -0.38289 & 0.37095 & -0.51550 \\
Bruneau and Jouron~\cite{BruneauJouron1990} & $256\times256$ & -0.3764 & 0.3665 & -0.5208 \\
Vanka~\cite{Vanka1986} & $321\times321$ & -0.387 & -- & -- \\
\bottomrule
\end{tabular}
}
\end{table}

The instantaneous element time-step distributions in Fig.~\ref{fig:cavity_dt} illustrate the local time-stepping behavior for the three artificial sound speeds, with the local time step reported in seconds. For all three cases, smaller time steps are concentrated near the walls and corner regions where viscous layers and stronger gradients appear, whereas the interior of the cavity permits larger local steps. The distribution also reflects the dependence of the acoustic time scale on $c_0$: increasing $c_0$ reduces the admissible local step, so the $c_0=20~\mathrm{m\,s^{-1}}$ case is globally more restrictive than the $c_0=5~\mathrm{m\,s^{-1}}$ case. The spatial pattern, however, remains consistent, confirming that the local time-step allocation follows the flow and mesh constraints rather than introducing irregular numerical artifacts.

\begin{figure}[htbp]
    \centering
    \begin{minipage}[t]{0.32\textwidth}
        \centering
        \includegraphics[width=\linewidth]{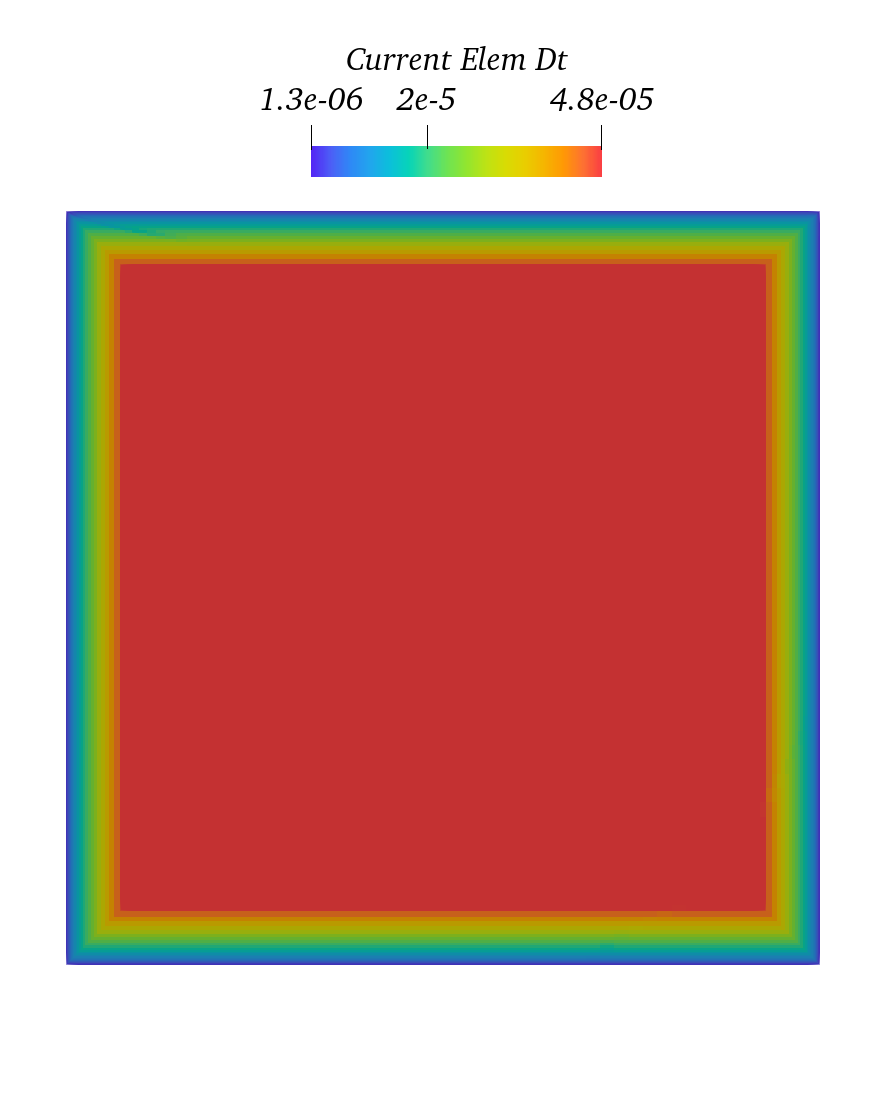}\\[-1mm]
        {\footnotesize (a) $c_0=5~\mathrm{m\,s^{-1}}$}
    \end{minipage}
    \hfill
    \begin{minipage}[t]{0.32\textwidth}
        \centering
        \includegraphics[width=\linewidth]{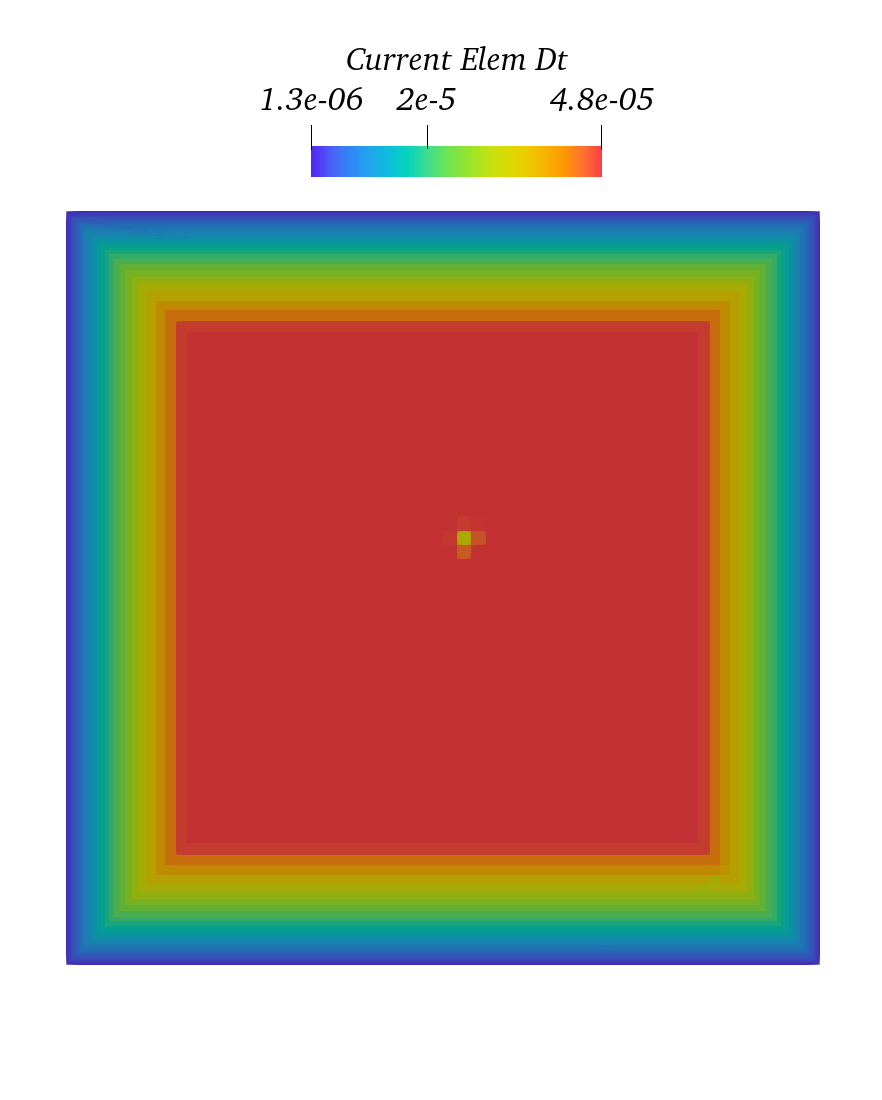}\\[-1mm]
        {\footnotesize (b) $c_0=10~\mathrm{m\,s^{-1}}$}
    \end{minipage}
    \hfill
    \begin{minipage}[t]{0.32\textwidth}
        \centering
        \includegraphics[width=\linewidth]{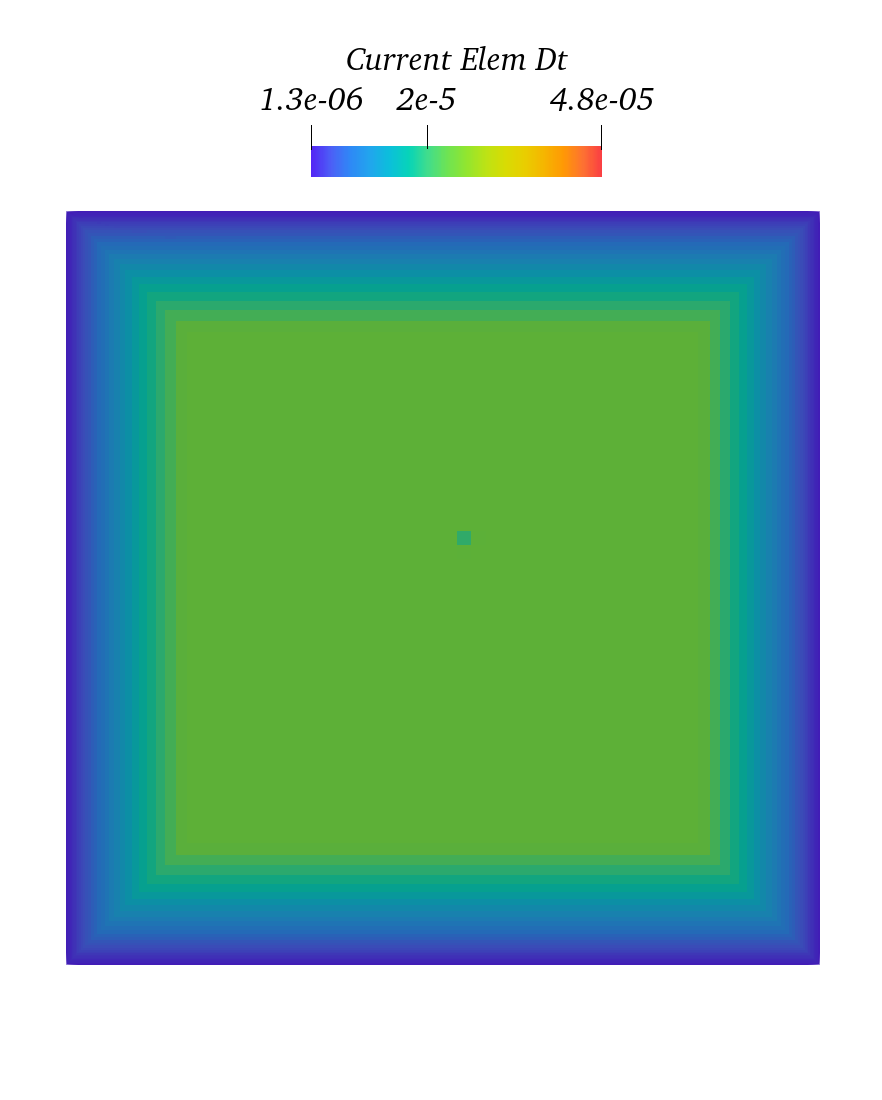}\\[-1mm]
        {\footnotesize (c) $c_0=20~\mathrm{m\,s^{-1}}$}
    \end{minipage}
    \caption{Instantaneous element time-step distributions, in seconds, for the lid-driven cavity calculation.}
    \label{fig:cavity_dt}
\end{figure}

Overall, the cavity-flow computation confirms that the weakly compressible DG formulation with the proposed local time-stepping treatment can reproduce the canonical $Re=1000$ cavity solution, including the primary recirculation, secondary corner vortices, centerline velocity extrema, pressure variation. The velocity profiles are already nearly insensitive to the tested values of $c_0$, while the pressure profiles and local time-step distributions retain visible and physically consistent differences among $c_0=5~\mathrm{m\,s^{-1}}$, $10~\mathrm{m\,s^{-1}}$, and $20~\mathrm{m\,s^{-1}}$.

\subsection{Three-dimensional 30P30N high-lift benchmark}
\label{subsec:30p30n}

Finally, the proposed method is evaluated using the 30P30N high-lift benchmark to validate its performance in complex geometries and unsteady low-speed environments. This classic three-element configuration (slat, main element, and flap), developed by McDonnell Douglas, is widely used to study slat noise \cite{choudhari2015assessment,zhang2026sweep,zhang2026multiple}. It provides a rigorous test for high-fidelity CFD methods due to its complex flow physics, characterized by strong shear-layer interactions and multiple separation zones \cite{zhang2024slat}.

The present study employs the modified 30P30N configuration developed by the Japan Aerospace Exploration Agency (JAXA) \cite{murayama2014experimental,murayama2018experimental}, as illustrated in Fig.~\ref{fig:30p30n}. The stowed chord length is $c_s=0.4572$ m. Both slat and flap are deflected at $30^\circ$ with chord lengths of $0.15c_s$ and $0.3c_s$, respectively. The computational setup follows the $3^{rd}$ AIAA Workshop on Benchmark Problems for Airframe Noise Computations (BANC-III) conditions in terms of inflow parameters \cite{choudhari2015assessment}. The free-stream Mach number, Reynolds number (based on $c_s$), and incidence angle are set to $Ma = 0.17$, $Re = 1.71 \times 10^6$, and $5.5^\circ$, respectively, corresponding to the nominal BANC-III case. The artificial sound speed is set to $c_0=341~\mathrm{m\,s^{-1}}$.

\begin{figure}[htbp]
\centering
\includegraphics[width=0.55\textwidth]{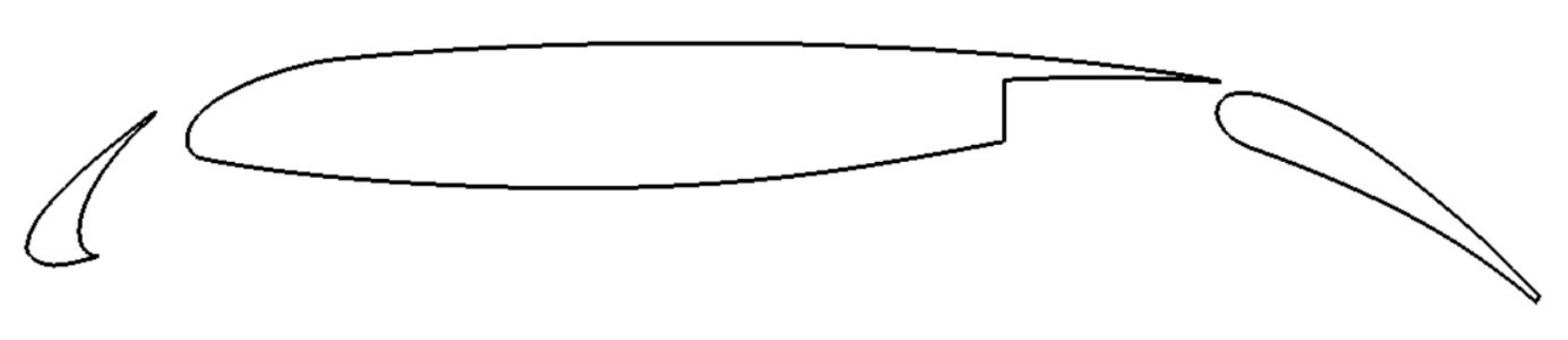}
\caption{Geometric profile of the JAXA-modified 30P30N high-lift airfoil.}
\label{fig:30p30n}
\end{figure}

The three-dimensional mesh contains 250,784 elements / 31.3 million degrees of freedom under polynomial degree $k=4$, with a boundary layer and geometric mesh growth rate of 1.3. Figure~\ref{fig:30p30n_mesh} illustrates the computational mesh and representative surface details. The same mesh from \cite{zhang2026sweep}, validated for FR method simulations, is utilized here. The computational domain extends $50c_s$ in both forward and vertical directions, reaching $60c_s$ aftward. The z-axis spans $1/9c_s$, following BANC-III guidelines \cite{choudhari2015assessment}. 

\begin{figure}[htbp]
    \centering
    \subfigure[Overview]{
        \includegraphics[width=0.45\textwidth]{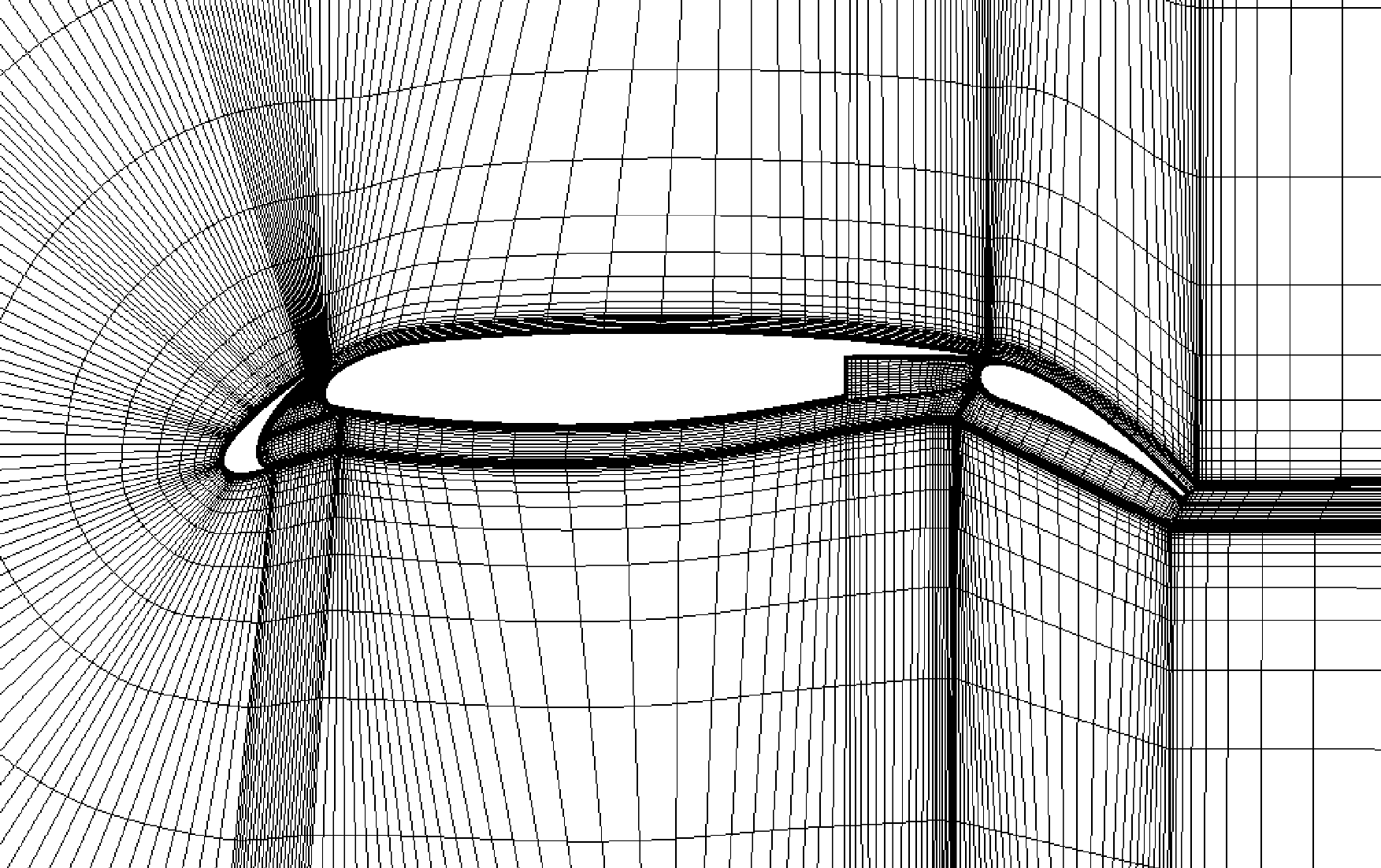}
    }
    \hfill
    \subfigure[Around the slat]{
        \includegraphics[width=0.45\textwidth]{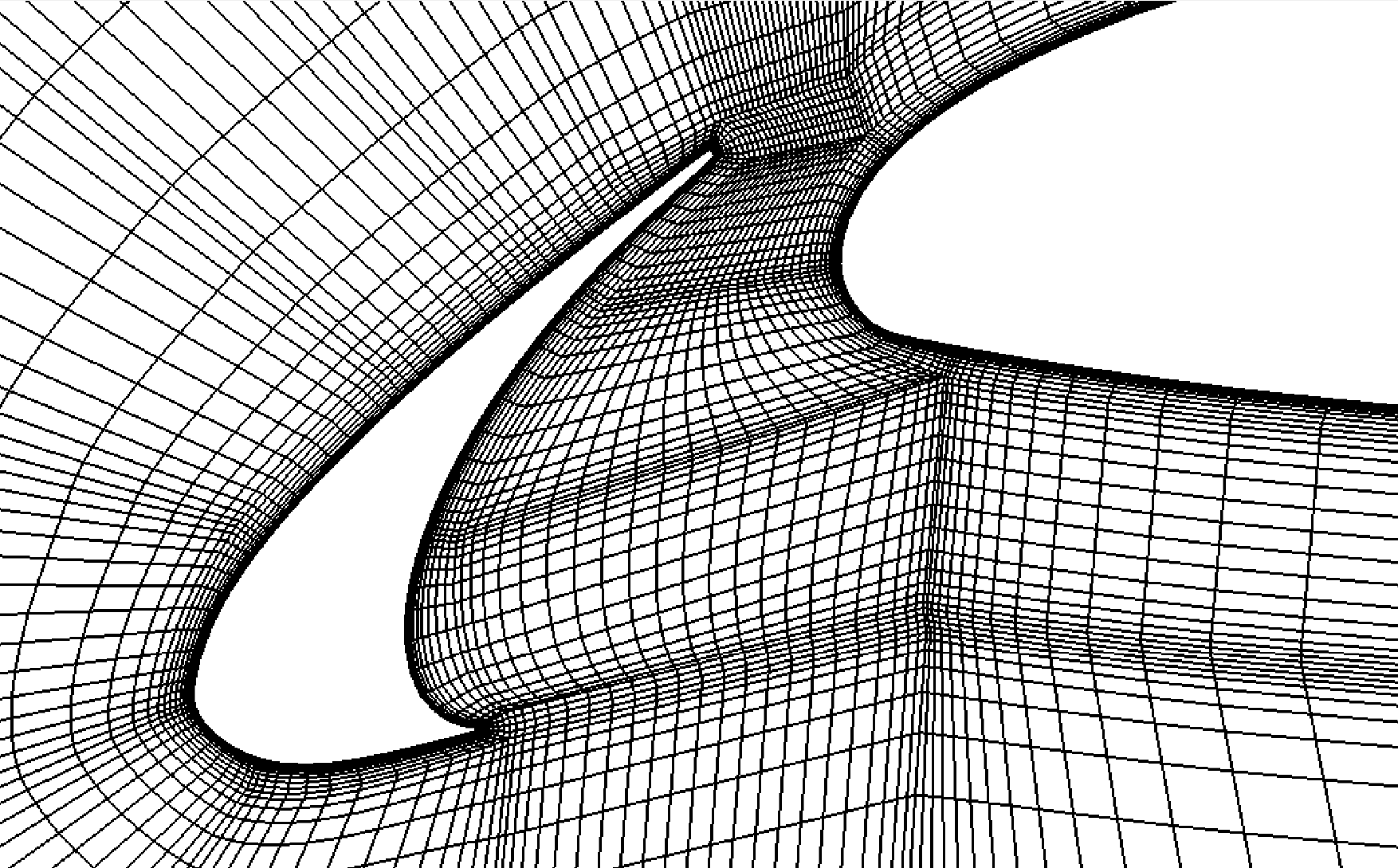}
    }
    
    \vspace{0.3cm}
    
    \subfigure[Around the main wing]{
        \includegraphics[width=0.45\textwidth]{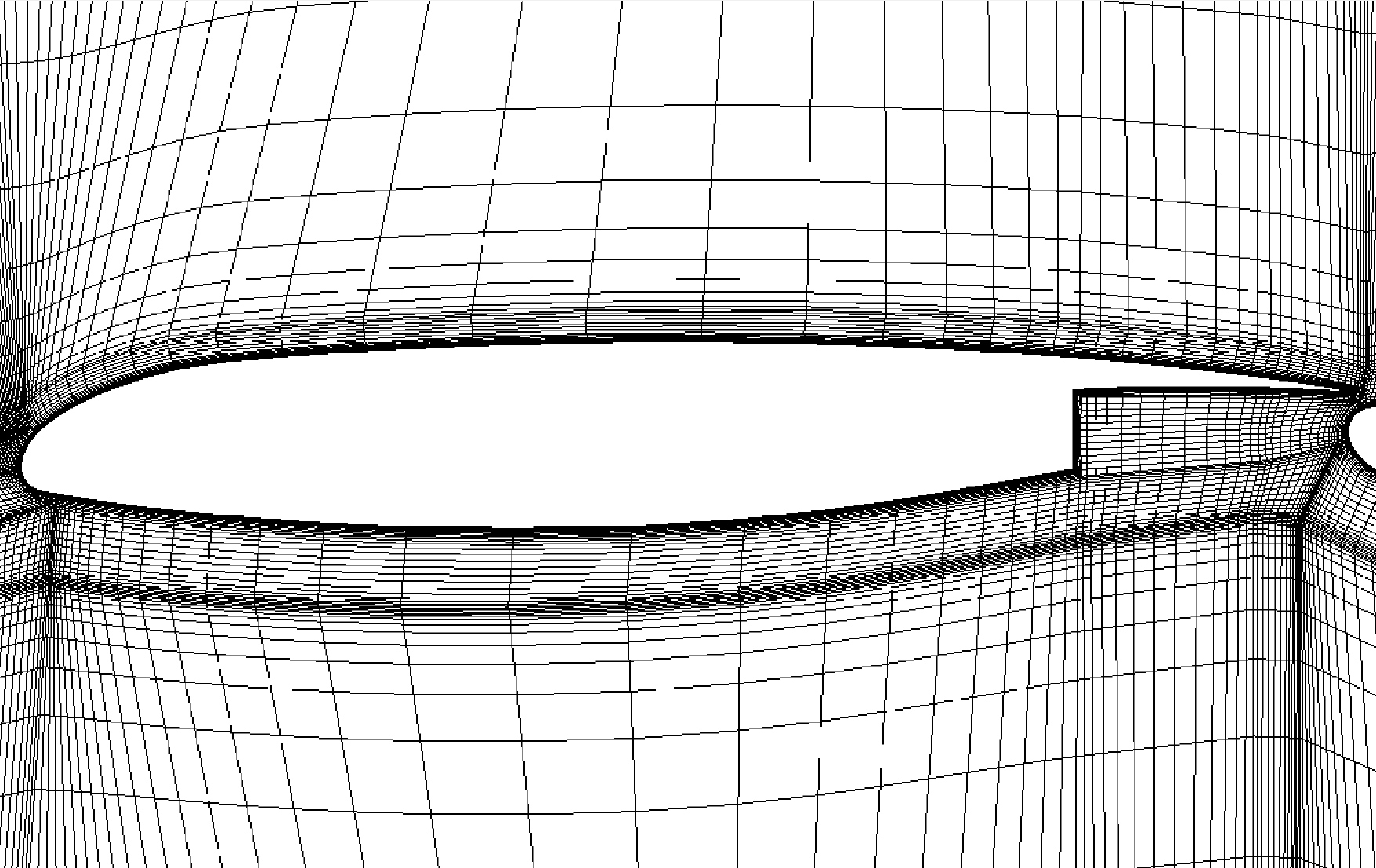}
    }
    \hfill
    \subfigure[Around the flap]{
        \includegraphics[width=0.45\textwidth]{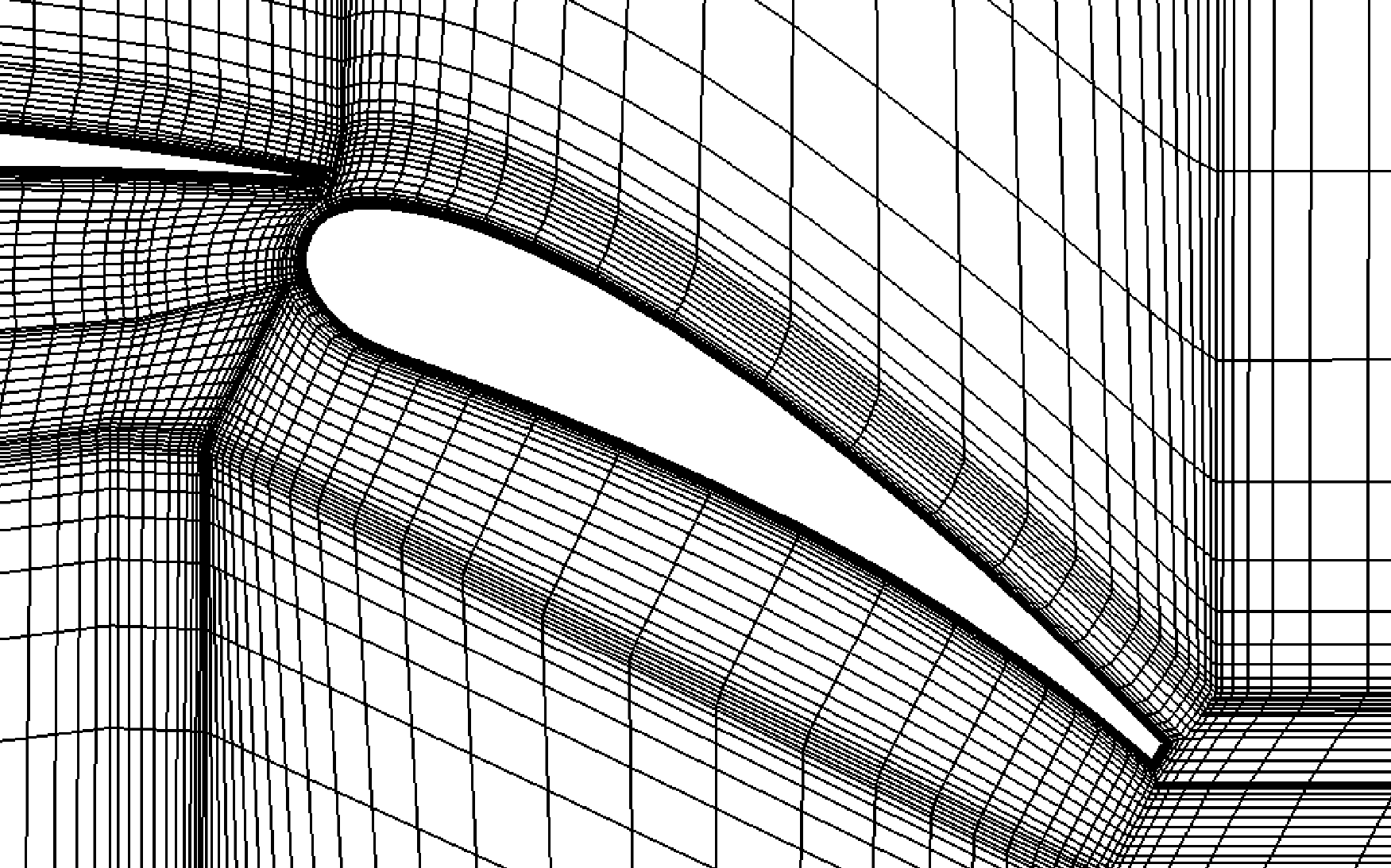}
    }
    \caption{Computational mesh of the 30P30N high-lift airfoil.}
    \label{fig:30p30n_mesh}
\end{figure}

Starting from a steady-RANS solution, the simulation is advanced until the flow reaches statistical convergence. The simulation spans a total of 30 flow pass times (FPTs, based on $c_s$), requiring 252 hours on four RTX 4090D GPUs. Samples are collected during the final 10 FPTs at 100 kHz. Spectra are obtained via the periodogram method with a Hanning window and 50\% overlap, giving a frequency resolution of 100 Hz.

Figure~\ref{fig:30p30n_Q} presents the instantaneous iso-surfaces of the $Q$ criterion within the slat cove region. The method successfully captures the main features of the multielement configuration: a shear layer originates at the slat cusp and transitions from a two-dimensional transverse vortex state to a three-dimensional vortical system, driven by Kelvin–Helmholtz instability; this system subsequently impinges on the slat lower surface, while a fraction of the large vortices breaks down into smaller structures within the slat cove. The majority are advected and stretched by the high-speed gap flow, eventually bursting towards the slat trailing edge.

\begin{figure}[htbp]
\centering
\includegraphics[width=0.55\textwidth]{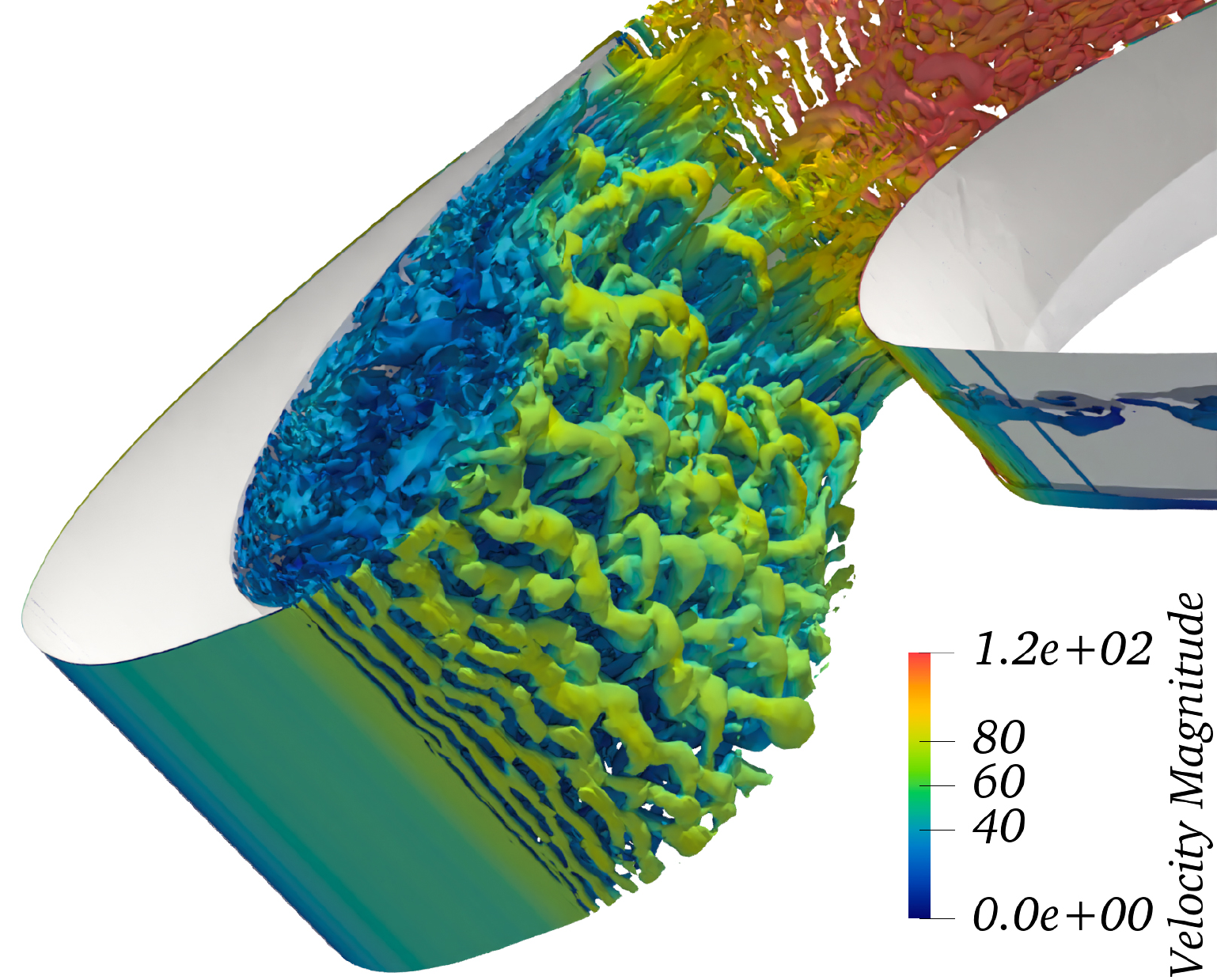}
\caption{Instantaneous Q-criterion iso-surface ($Q=2\times10^7$) in the slat cove region, colored by velocity magnitude.}
\label{fig:30p30n_Q}
\end{figure}

For quantitative validation, the surface pressure coefficient ($C_p$) distributions on the airfoil midplane are presented in Fig.~\ref{fig:30p30n_Cp} alongside experimental data. It should be noted that the angle of attack was corrected to account for wall interference inherent in the closed test section, due to the significant blockage effect induced by the high-lift devices; thus, the free-stream value of $5.5^\circ$ corresponds to $7^\circ$ and $10^\circ$ in the JAXA \cite{murayama2018experimental} and Florida State University (FSU) \cite{pascioni2016aeroacoustic} tests, respectively. The computed distributions are in good agreement with the measurements, although the suction pressures on the upper surface are slightly underestimated. The predicted lift coefficient $C_L = 2.654$ is consistent with the BANC-III workshop mean of 2.643 \cite{choudhari2015assessment,chen2021noise,xie2026noise}, with a relative difference of merely 0.42\%.

\begin{figure}[htbp]
    \centering
    \subfigure[All three elements]{
        \includegraphics[width=0.45\textwidth]{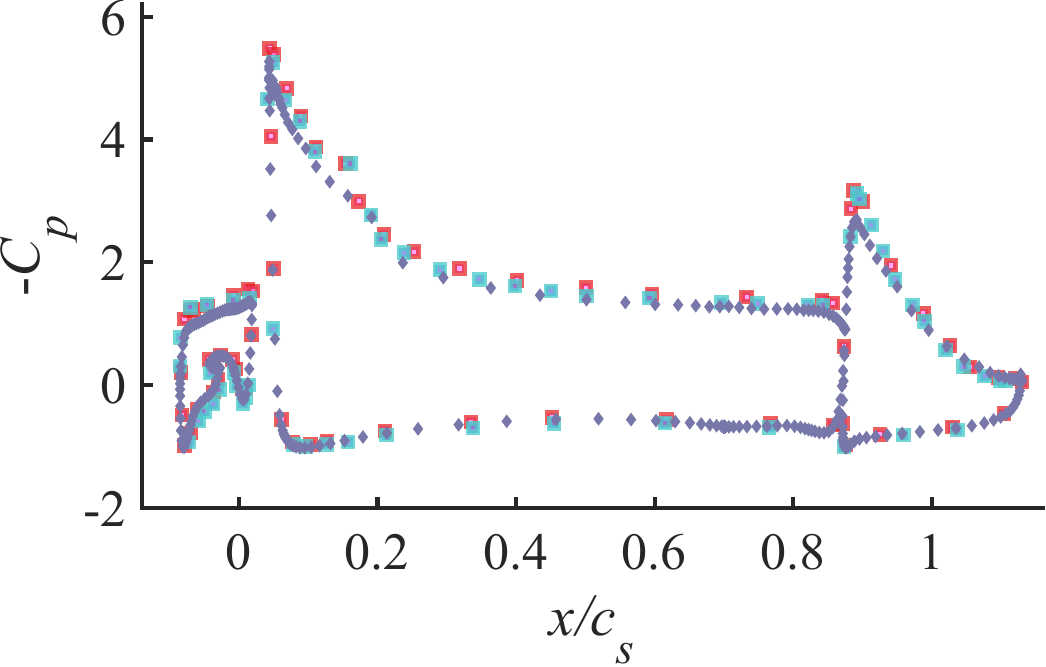}
    }
    \hfill
    \subfigure[Around the slat]{
        \includegraphics[width=0.45\textwidth]{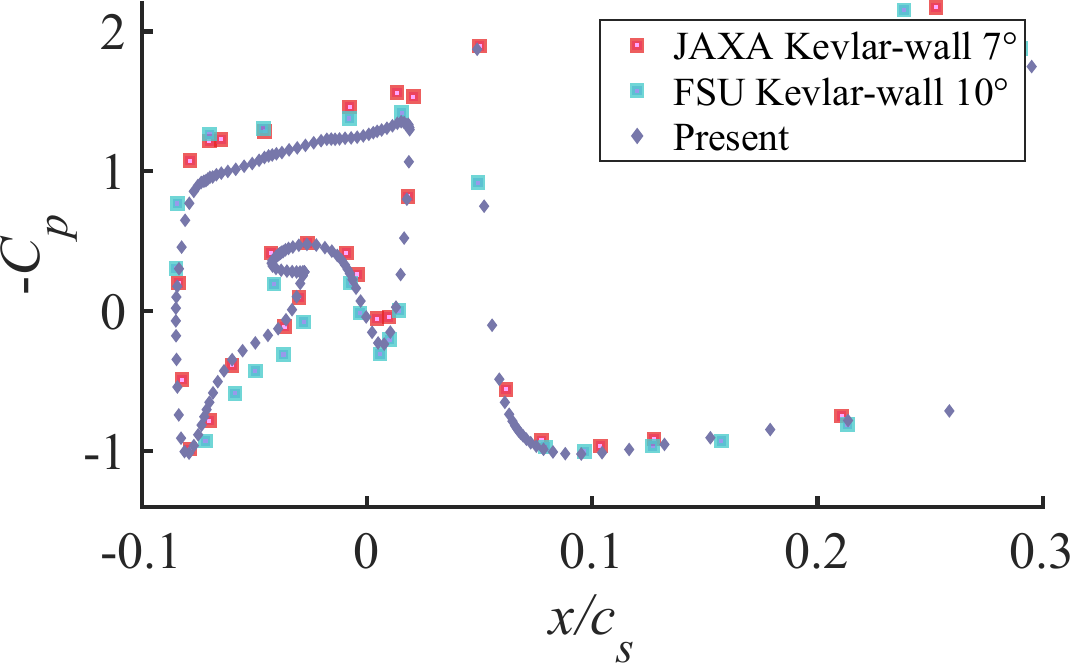}
    }
    \caption{Time-averaged $C_p$ distributions on the airfoil midplane.}
    \label{fig:30p30n_Cp}
\end{figure}

Figure~\ref{fig:probe} presents the wall surface pressure probes used in this study. Wall pressure spectra are compared with experimental datasets \cite{murayama2014experimental,murayama2018experimental,pascioni2014experimental,pascioni2016aeroacoustic} in Fig.~\ref{fig:30p30n_nf}. The numerical results reproduce the broadband spectrum characteristics at the S3 location. At the other three probe locations, the numerical simulation accurately captures the discrete tones in the mid-low frequency ranges, while the predicted broadband levels and hump amplitudes fall between the two experimental datasets. Overall, the comparison with the measured spectra is favorable.

\begin{figure}[htbp]
\centering
\includegraphics[width=0.5\textwidth]{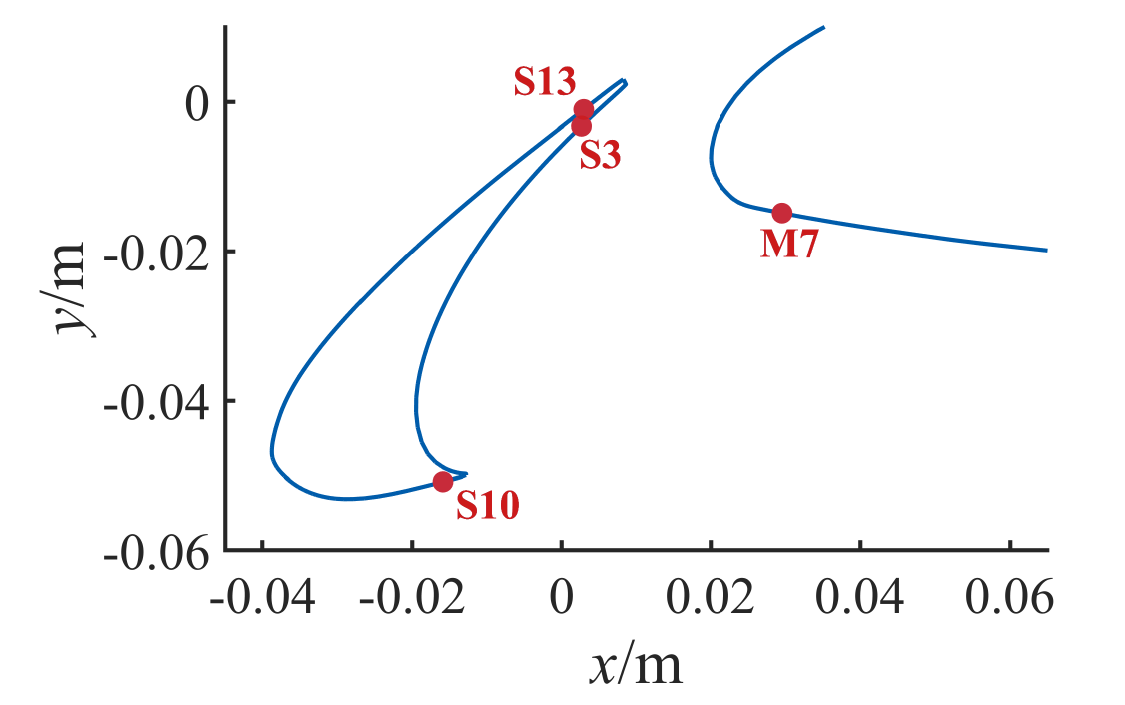}
\caption{Schematic of the wall surface pressure probes used in this study.}
\label{fig:probe}
\end{figure}

\begin{figure}[htbp]
    \centering
    \subfigure[S3]{
        \includegraphics[width=0.45\textwidth]{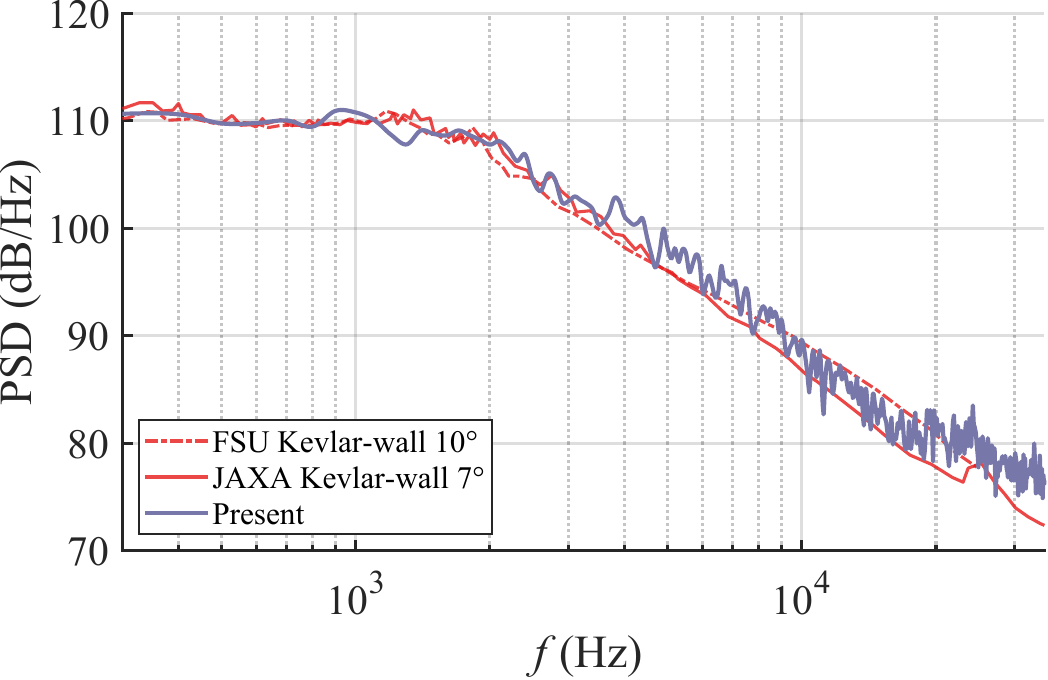}
    }
    \hfill
    \subfigure[S10]{
        \includegraphics[width=0.45\textwidth]{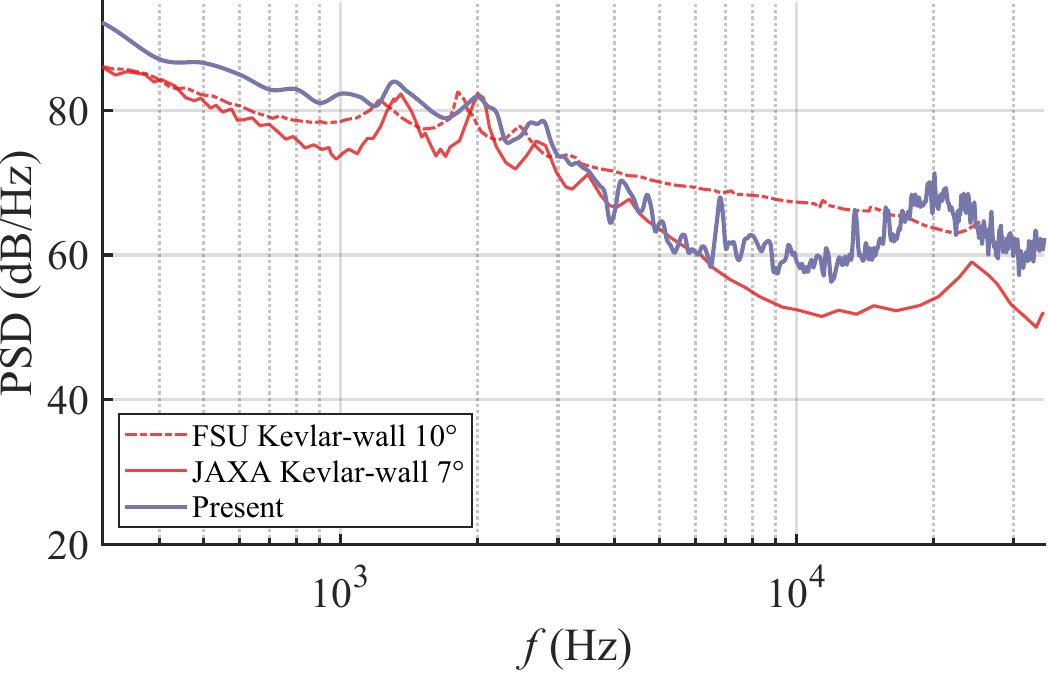}
    }
    
    \vspace{0.3cm}
    
    \subfigure[S13]{
        \includegraphics[width=0.45\textwidth]{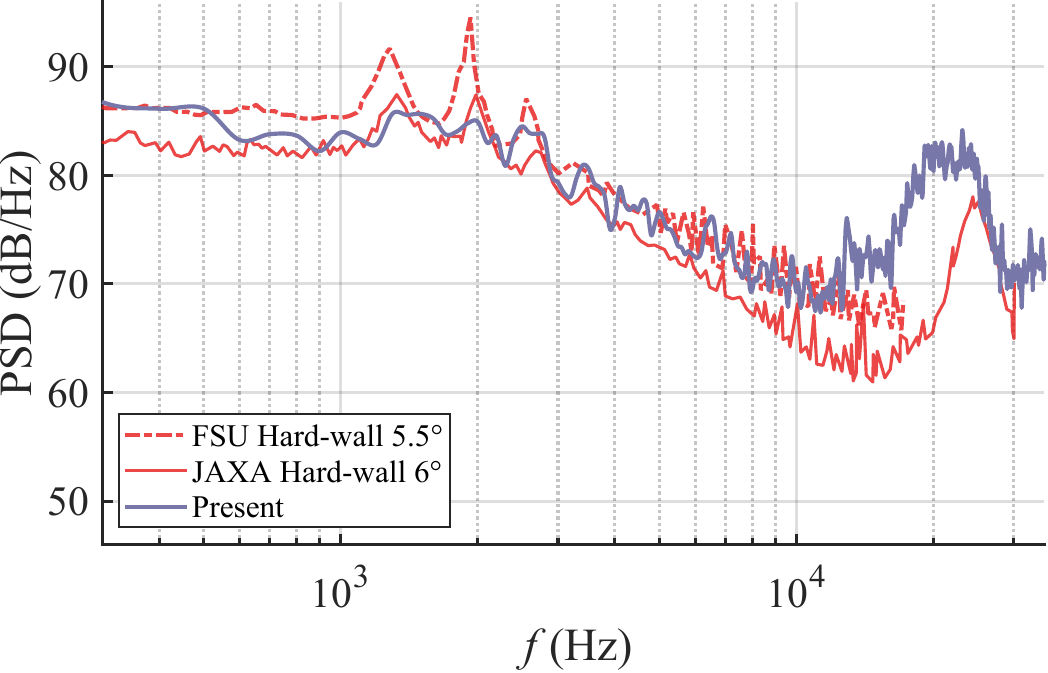}
    }
    \hfill
    \subfigure[M7]{
        \includegraphics[width=0.45\textwidth]{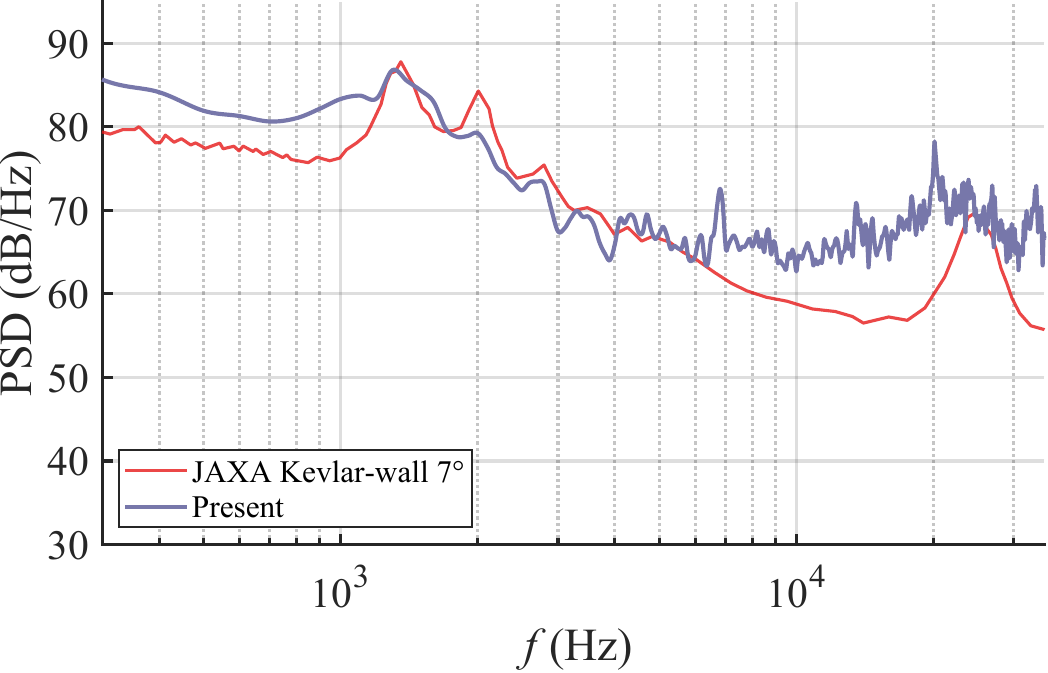}
    }
    \caption{Wall pressure spectra compared with experimental datasets.}
    \label{fig:30p30n_nf}
\end{figure}

The time step distribution is shown in Fig.~\ref{fig:30p30n_dt}. Local time stepping assigns larger admissible time steps to most mesh cells, while critical small time steps persist only within boundary layers and locally refined regions. The two-order-of-magnitude spread between extreme time steps illustrates the motivation for local time stepping, although no wall-clock speedup is quantified in the present study.
\begin{figure}[htbp]
\centering
\includegraphics[width=0.7\textwidth]{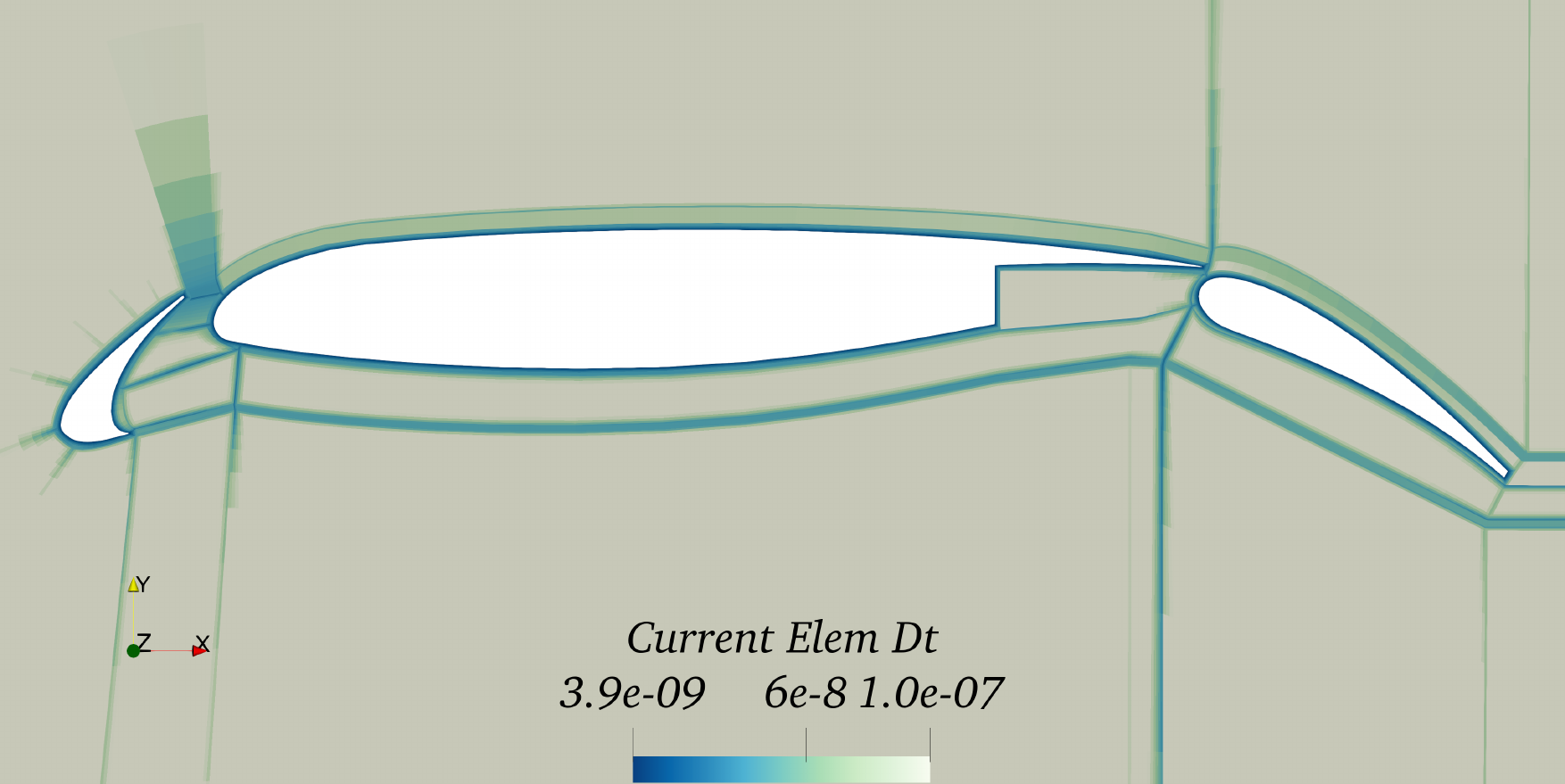}
\caption{Instantaneous time step distribution for the 30P30N high-lift benchmark.}
\label{fig:30p30n_dt}
\end{figure}

The present three-dimensional computation highlights two important features of the proposed method. First, the weakly compressible DG framework remains robust on a geometrically complex multielement configuration. Second, the local time-stepping strategy is compatible with the strongly nonuniform time-step distribution induced by the near-wall layers and inter-element gaps. Therefore, the 30P30N benchmark simulation further confirms the practical potential of the proposed method for high-fidelity simulations of realistic low-Mach-number aerodynamic flows.

\subsection{Summary of numerical findings}
\label{subsec:numerical_summary}

From the above numerical experiments, several observations can be made. First, the one-dimensional weakly compressible simple wave and the two-dimensional weakly compressible shear-wave exact-solution tests verify the expected convergence behavior on nonuniform meshes, while the mass variation remains at roundoff level for the conservative split LTS formulation. In the two-dimensional shear-wave test, the additional comparison also shows that directly applying the Gassner-type unsplit update to the nodal DG discretization destroys this conservation property, increasing the mass drift from about $10^{-15}$ to about $10^{-8}$. Second, the $Re=100$ circular-cylinder case shows that the force statistics agree with established benchmark data, while the TR flux used here produces weaker far-field pressure fluctuations than the LF flux in the present cylinder-flow test. Third, the lid-driven cavity flow shows that the method can accurately predict canonical viscous low-speed flows with strong wall effects. Finally, the three-dimensional 30P30N benchmark confirms the robustness and practical applicability of the method to a realistic complex configuration.

Overall, the numerical results indicate that the proposed weakly compressible DG method with time-accurate local time stepping preserves conservation and maintains the expected accuracy on the tested nonuniform meshes, while retaining robustness and local time-step flexibility for low-Mach-number flow simulations.

\section{Conclusions and Future Work}
\label{sec:conclusions}

This paper has developed a time-accurate local time-stepping DG method for unsteady low-Mach-number flows based on a weakly compressible formulation. The method combines a strong-form nodal DGSEM discretization on quadrilateral and hexahedral elements, IIPG treatment of viscous terms, a two-rarefaction approximate Riemann flux specialized to the weakly compressible convective system, and a CERK-based continuous predictor for local time stepping. For the nodal DGSEM discretization, the face contribution is split into a purely interior part and a common-flux part so that the volume term and the purely interior face contribution are integrated at the same element-local quadrature points, while the conservative exchange across interfaces is handled by face-wise time matching of the common flux.

The numerical results show that the proposed method achieves the expected convergence behavior on nonuniform meshes for the one-dimensional simple wave and the two-dimensional weakly compressible shear wave, while preserving mass to roundoff accuracy in the periodic verification problems. The comparison with the Gassner-type unsplit update further confirms the need for the split treatment in the present nodal DG formulation. For the $Re=100$ cylinder-flow case, the TR flux used here gives force statistics consistent with reference data and produces weaker far-field pressure fluctuations than the LF flux in the present test. The lid-driven cavity calculation confirms that the method can resolve canonical viscous low-speed flow features, including centerline velocity extrema, and pressure variation. The three-dimensional 30P30N computation further demonstrates the applicability of the method to realistic complex geometries and large-scale unsteady low-Mach-number flow simulations.

Several directions remain for future work. First, more systematic performance studies will be carried out on large CPU--GPU heterogeneous platforms to quantify the wall-clock performance and possible gains of the local time-stepping strategy for different mesh nonuniformities and polynomial degrees. Second, implicit--explicit (IMEX) hybrid time integration will be investigated within the present conservative LTS-DG framework. This extension can be considered from two complementary viewpoints. Equation-based splitting will emphasize implicit treatment of stiff source terms, while convective fluxes and other nonstiff contributions remain under explicit local time stepping. In terms of spatial-region splitting, near-wall, small-cell, or strongly refined regions may be advanced by implicit or locally implicit treatments, while the remaining regions continue to use explicit LTS. Key issues include conservative common-flux exchange across explicit--implicit interfaces, time matching between different temporal treatments, and consistency with the CERK continuous predictor. Third, the LTS treatment of viscous terms will be extended from the present IIPG discretization to a broader IPDG family, in particular the SIPG and NIPG variants. A key issue is to formulate the time integration of the symmetric or nonsymmetric consistency terms, penalty terms, and viscous numerical fluxes so that the element-local quadrature, face-wise time matching, stability, and conservation requirements remain mutually consistent. This will also allow systematic comparisons of IIPG, SIPG, and NIPG in terms of accuracy, robustness, and computational cost under local time stepping for viscous weakly compressible flows.

\appendix

\section{CERK Coefficients}
\label{app:cerk_coefficients}

Following the continuous explicit Runge--Kutta notation of Owren and Zennaro \cite{OwrenZennaro1992}, the stage values and the continuous predictor are written as
\begin{equation*}
Y_i
=
y_n+\Delta t\sum_{j=1}^{i-1}a_{ij}K_j,
\qquad
K_i=f(Y_i),
\end{equation*}
and
\begin{equation*}
P(t_n+\theta\Delta t)
=
y_n+\Delta t\sum_{i=1}^{s}b_i(\theta)K_i,
\qquad 0\le \theta\le 1.
\end{equation*}
The first- to fourth-order CERK coefficients used in the local predictor are listed below. For $p=1,2$, the standard minimal continuous extensions are used; for $p=3,4$, the coefficients are those reported by Owren and Zennaro \cite{OwrenZennaro1992}.

\paragraph{First order.}
\begin{equation*}
c_1=0,\qquad b_1^{(1)}(\theta)=\theta .
\end{equation*}

\paragraph{Second order.}
\begin{equation*}
\begin{array}{c|cc}
0 &  &  \\
1 & 1 &  \\
\end{array}
\qquad
b_1^{(2)}(\theta)=\theta-\frac{1}{2}\theta^2,\qquad
b_2^{(2)}(\theta)=\frac{1}{2}\theta^2 .
\end{equation*}

\paragraph{Third order.}
{\small
\renewcommand{\arraystretch}{1.35}
\setlength{\arraycolsep}{5pt}
\[
\begin{array}{c|cccc}
0 &  &  &  &  \\
\frac{12}{23} & \frac{12}{23} &  &  &  \\
\frac{4}{5} & -\frac{68}{375} & \frac{368}{375} &  &  \\
1 & \frac{31}{144} & \frac{529}{1152} & \frac{125}{384} & 
\end{array}
\]
\begin{align*}
b_1^{(3)}(\theta)&=\frac{41}{72}\theta^3-\frac{65}{48}\theta^2+\theta,\\
b_2^{(3)}(\theta)&=-\frac{529}{576}\theta^3+\frac{529}{384}\theta^2,\\
b_3^{(3)}(\theta)&=-\frac{125}{192}\theta^3+\frac{125}{128}\theta^2,\\
b_4^{(3)}(\theta)&=\theta^3-\theta^2 .
\end{align*}
}

\paragraph{Fourth order.}
{\small
\renewcommand{\arraystretch}{1.35}
\setlength{\arraycolsep}{3.5pt}
\[
\begin{array}{c|cccccc}
0 &  &  &  &  &  &  \\
\frac{1}{6} & \frac{1}{6} &  &  &  &  &  \\
\frac{11}{37} & \frac{44}{1369} & \frac{363}{1369} &  &  &  &  \\
\frac{11}{17} & \frac{3388}{4913} & -\frac{8349}{4913} & \frac{8140}{4913} &  &  &  \\
\frac{13}{15} & -\frac{36764}{408375} & \frac{767}{1125} & -\frac{32708}{136125} & \frac{210392}{408375} &  &  \\
1 & \frac{1697}{18876} & 0 & \frac{50653}{116160} & \frac{299693}{1626240} & \frac{3375}{11648} & 
\end{array}
\]
\begin{align*}
b_1^{(4)}(\theta)&=-\frac{866577}{824252}\theta^4+\frac{1806901}{618189}\theta^3-\frac{104217}{37466}\theta^2+\theta,\\
b_2^{(4)}(\theta)&=0,\\
b_3^{(4)}(\theta)&=\frac{12308679}{5072320}\theta^4-\frac{2178079}{380424}\theta^3+\frac{861101}{230560}\theta^2,\\
b_4^{(4)}(\theta)&=-\frac{7816583}{10144640}\theta^4+\frac{6244423}{5325936}\theta^3-\frac{63869}{293440}\theta^2,\\
b_5^{(4)}(\theta)&=-\frac{624375}{217984}\theta^4+\frac{982125}{190736}\theta^3-\frac{1522125}{762944}\theta^2,\\
b_6^{(4)}(\theta)&=\frac{296}{131}\theta^4-\frac{461}{131}\theta^3+\frac{165}{131}\theta^2 .
\end{align*}
}

\end{document}